\documentclass{article}
\usepackage[utf8]{inputenc}

\usepackage{graphics}
\usepackage[margin=1in]{geometry}
\usepackage{amsmath}
\usepackage{amsfonts}
\usepackage{algorithm, algorithmic}
\usepackage{mathtools}
\usepackage{bm}

\usepackage{amssymb}
\usepackage{amsthm}
\usepackage{enumerate}
\usepackage{float}
\usepackage{lscape}

\usepackage{thmtools}
\usepackage{thm-restate}
\usepackage[colorlinks, hypertexnames=false]{hyperref}
\hypersetup{linkcolor=blue, urlcolor=blue, citecolor=red}
\usepackage[capitalise]{cleveref}

\newtheorem{thm}{Theorem}[subsection]
\newtheorem{prop}[thm]{Proposition}
\newtheorem{cor}[thm]{Corollary}
\newtheorem{lemma}[thm]{Lemma}
\newtheorem{defi}[thm]{Definition}
\newenvironment{de}{\begin{defi} \rm}{\end{defi}}
\newtheorem{exam}[thm]{Example}
\newenvironment{ex}{\begin{exam} \rm}{\end{exam}}
\newtheorem{remark}[thm]{Remark}

\newtheorem*{lemma-no-num}{Lemma}

\newcommand{\lex}{\operatorname{lex}}

\RequirePackage{doi}
\usepackage{hyperref}
\hypersetup{linkcolor=blue, urlcolor=blue, citecolor=red}

\newcommand{\N}{\mathbb{N}}

\usepackage{tikz}
\usetikzlibrary{calc}
\definecolor{color0}{RGB}{127, 0, 255}
\definecolor{color1}{RGB}{128, 128, 128}
\definecolor{color2}{RGB}{255, 0, 127}
\definecolor{color3}{RGB}{255, 0, 255}

\title{An explicit algorithm for normal forms in small overlap monoids}
\author{James D. Mitchell and Maria Tsalakou}
\date{\today}

\begin{document}

\maketitle
\begin{abstract}
We describe a practical algorithm for computing normal forms for semigroups and monoids with finite presentations satisfying so-called small overlap conditions. 
Small overlap conditions are natural conditions on the relations in a presentation, 
which were introduced by J. H. Remmers and subsequently studied extensively by M. Kambites. Presentations satisfying these conditions are ubiquitous; Kambites showed that a randomly chosen finite presentation satisfies the $C(4)$ condition with probability tending to 1 as the sum of the lengths of relation words tends to infinity. Kambites also showed that several key problems for finitely presented semigroups and monoids are tractable in $C(4)$ monoids: the word problem is solvable in $O(\min\{|u|, |v|\})$ time in the size of the input words $u$ and $v$; the uniform word problem for $\langle A|R\rangle$ is solvable in $O(N ^ 2 \min\{|u|, |v|\})$ where $N$ is the sum of the lengths of the words in $R$;
 and a normal form for any given word $u$ can be found in $O(|u|)$ time.  Although Kambites' algorithm for solving the word problem in $C(4)$ monoids is highly practical, it appears that the coefficients in the linear time algorithm for computing normal forms are too large in practice. 
 
 In this paper, we present an algorithm for computing normal forms in $C(4)$ monoids that has time complexity $O(|u| ^ 2)$ for input word $u$, but where the coefficients are sufficiently small to allow for practical computation. Additionally, we show that the uniform word problem for small overlap monoids can be solved in $O(N \min\{|u|, |v|\})$ time.

\end{abstract}

\section{Introduction}

In this paper we present an explicit algorithm for computing normal forms of words in so-called small overlap monoids.
The problem of finding normal forms for finitely presented monoids, semigroups, and groups, is classical, and is widely studied in the literature; some classical examples can be found in~\cite{Gilman, sims_1994}, and more recently in~\cite{overlaps2}.  In a monoid $M$ defined by a presentation $\left\langle A \, | \,R \right\rangle$, the \textit{word problem} asks if given $u, v\in A ^ *$, does there exist an algorithm deciding whether or not $u$ and $v$ represent the same element of $M$ (i.e. an algorithm which outputs ``yes'' if $(u, v)$ belongs to the least congruence on $A ^ *$ containing $R$ and ``no'' otherwise)? The word problem is said to be \textit{decidable} if such an algorithm exists, and \textit{undecidable} if it does not.
If an algorithm for computing normal forms for a monoid presentation  $\mathcal{P} = \left\langle A \, | \,R \right\rangle$ is available, then the word problem in the monoid defined by $\mathcal{P}$ is decidable by computing normal forms of $u$ and $v$,  
and checking if these two words coincide. In 1947, Markov~\cite{markov} and Post~\cite{post_1947} independently proved that the word problem for monoids is undecidable, in general, and, as such, the problem of finding normal forms for arbitrary finitely presented monoids is also undecidable.

Although undecidable in general, there are many special cases where it is possible to determine the structure of a finitely presented monoid, and to solve the word problem, and there are several well-known algorithms for doing this. Of course, because the word problem is undecidable in general, none of these algorithms can solve the word problem for all finitely presented monoids.
The Todd-Coxeter Algorithm~\cite{todd-coxeter} terminates if and only if the input finite presentation defines a finite monoid. If it does terminate, the output of the Todd-Coxeter Algorithm is the (left or right) Cayley graph of the monoid defined by the input presentation. The word problem is thus solved by following the paths in the Cayley graph starting at the identity and labelled by any $u, v\in A ^ *$, and checking whether these paths end at the same node. The Knuth-Bendix Algorithm~\cite{knuth-bendix}, can terminate even when the monoid $M$ defined by the input (again finite) presentation is infinite. The output of the Knuth-Bendix Algorithm is a finite noetherian complete rewriting system defining the same monoid $M$ as the input presentation. The normal forms of such a rewriting system are just the words to which no rewriting rule can be applied, and a normal form can be obtained for any word, by arbitrarily applying relations in the rewriting system until no further relations apply. 
This permits the word problem to be solved in monoids where the Knuth-Bendix Algorithm terminates via the computation of normal forms. 

In this paper we are concerned with a class of finitely presented monoids, introduced by Remmers~\cite{Remmers}, and studied further by Kambites~\cite{overlaps1}, ~\cite{overlaps2}, and~\cite{kambites2011}. 

If $\mathcal{P} = \left\langle A \, | \,R \right\rangle$ is a monoid presentation, then we will refer to the left or right hand side of any pair $(u, v) \in R$ as a \textit{relation word}. 
A word $w\in A ^*$ is said to be a \textit{piece} of $\mathcal{P}$ if $w$ is a factor of at least two distinct relation words, or $w$ occurs more than once as a factor of a single relation word (possibly overlapping). Note that if a relation word $u$ appears as one side of more than one relation in the presentation, then $u$ is not considered a piece.
A monoid presentation $\mathcal{P}$ is said to satisfy the condition $C(n)$,  $n \in \mathbb{N}$, 
if the minimum number of pieces in any factorisation of a relation word is at 
 least $n$. If no relation word in $\mathcal{P}$ equals the empty word, then $\mathcal{P}$ satisfies $C(1)$. If no relation word can be written as a product of pieces, then we say that $\mathcal{P}$ satisfies $C(n)$ for all $n\in \N$. 
If a presentation satisfies $C(n)$, then it also satisfies $C(k)$ for every $k\in \N$ such that $1\leq k < n$.

For example, the presentation $\mathcal{P}=\left\langle a, b, c \, | \,abc = cba \right\rangle$ satisfies $C(3)$. The set of pieces is $P = \{\varepsilon, a, b, c\}$ and each relation word can be written as a product of exactly 3 pieces. Hence $\mathcal{P}$ does not satisfy $C(4)$. Similarly, for $\mathcal{P}=\left\langle a, b, c \, | \,acba = a^2bc \right\rangle$ the set of pieces is $P = \{\varepsilon, a, b, c\}$ and $\mathcal{P}$ satisfies $C(4)$ but not $C(5)$. If $\mathcal{P}=\left\langle a, b, c, d\, | \,acba = a^2bc, acba = db^3d \right\rangle$, the set of pieces is $P = \{\varepsilon, a, b, c, d, b^2\}$ and $\mathcal{P}$ satisfies $C(4)$ but not $C(5)$, since the relation words $acba$ and $a^2bc$ can be written as the product of 4 pieces. For the presentation $\mathcal{P}=\left\langle a, b, c, d \, | \,a^2bc = a^2bd \right\rangle$ the set of pieces is $P=\{\varepsilon, a, b, a^2, a^2b\}$ and none of the relation words can be written as a product of pieces since neither $c$ nor $d$ are pieces. 

If a finite monoid presentation satisfies the condition $C(4)$, then we will refer to the monoid defined by the presentation as a \textit{small overlap monoid}. Remmers initiated the study of $C(3)$ monoids in the paper~\cite{Remmers}; see also~\cite[Chapter 5]{higgins}.
If a monoid presentation $\langle A \mid R\rangle$ satisfies $C(3)$, then the number of words in any class of $R ^ {\#}$ is finite, and so the monoid defined by the presentation is infinite; see~\cite[Corollary 5.2.16]{higgins}. The word problem is solvable in $C(3)$ monoids but the algorithm described in \cite[Theorem 5.2.15]{higgins} has exponential complexity.
Groups with similar combinatorial conditions have also been studied and such groups are called \textit{small cancellation groups}. Small cancellation groups have decidable word problem; see~\cite[Chapter 5]{lyndon} for further details. 

In~\cite{Kambites2011aa}, Kambites' showed that the probability that a randomly chosen finite monoid presentation is $C(4)$ tends to $1$ as the length of the presentation tends to infinity; and the rate of convergence appears to be rather high; see \cref{tab:converge}. Hence, in some sense, algorithms for small overlap monoids are widely applicable. In Kambites~\cite{kambites2011}, an explicit algorithm (\textbf{WpPrefix}) is presented for solving the word problem for finitely presented   monoids satisfying $C(4)$. If $u, v\in A ^ *$, then, provided that certain properties of the presentation are known already, Kambites' Algorithm requires $O(\min\{|u|, |v|\})$ time. In Kambites~\cite{overlaps2}, among many other results, it is shown that there exists a linear time algorithm for computing normal forms in $C(4)$ monoids, given a preprocessing step that requires polynomial time in the size of the alphabet and the maximum length of a relation word. 
\begin{table}[]
    \centering
    \begin{tabular}[t]{cc}
    \begin{tabular}{l|l|l|l}
        $n$  & $C(4)$ monoids & monoids      &ratio \\ \hline
        1  & 0              & 1            & 0.0 \\
        2  & 0              & 14           & 0.0\\
        3  & 0              & 76           & 0.0\\
        4  & 0              & 344          & 0.0 \\
        5  & 0              & 1,456        & 0.0\\
        6  & 0              & 5,984        & 0.0\\
        7  & 2              & 24,256       & 0.000082 \\
        8  & 26             & 97,664       & 0.000266\\
        9  & 760            & 391,936      & 0.001939 \\
        10 & 17,382         & 1,570,304    & 0.011069 \\
        11 & 217,458        & 6,286,336    & 0.034592\\
        12 & 1,994,874      & 25,155,584   & 0.079301\\
        13 & 14,633,098     & 100,642,816  & 0.145396 \\\hline
        14 & -              & 4.026122e+08 & 0.186342\\
        15 & -              & 1.610531e+09 & 0.280811\\
        16 & -              & 6.442287e+09 & 0.374679\\
        17 & -              & 2.576948e+10 & 0.473369\\
        18 & -              & 1.030786e+11 & 0.594068\\
        19 & -              & 4.123155e+11 & 0.681053\\
        20 & -              & 1.649265e+12 & 0.732404\\
        21 & -              & 6.597065e+12 & 0.801495\\
        22 & -              & 2.638827e+13 & 0.843976\\
        23 & -              & 1.055531e+14 & 0.884619\\
        24 & -              & 4.222124e+14 & 0.929988\\
        25 & -              & 1.688850e+15 & 0.941666\\
     \end{tabular}            
     & 
    \begin{tabular}{l|l|l|l}
        $n$  & $C(4)$ monoids & monoids      & ratio \\ \hline
        26 & -              & 6.755399e+15 & 0.963477\\
        27 & -              & 2.702160e+16 & 0.970104\\
        28 & -              & 1.080864e+17 & 0.977796\\
        29 & -              & 4.323456e+17 & 0.990216\\
        30 & -              & 1.729382e+18 & 0.989878\\
        31 & -              & 6.917529e+18 & 0.994861\\
        32 & -              & 2.767012e+19 & 0.995684\\
        33 & -              & 1.106805e+20 & 0.996879\\
        34 & -              & 4.427219e+20 & 0.998821\\
        35 & -              & 1.770887e+21 & 0.996402\\
        36 & -              & 7.083550e+21 & 0.999423\\
        37 & -              & 2.833420e+22 & 0.999776\\
        38 & -              & 1.133368e+23 & 0.997902\\
        39 & -              & 4.533472e+23 & 0.999928\\
        40 & -              & 1.813389e+24 & 0.99953 \\
        41 & -              & 7.253555e+24 & 0.999982\\
        42 & -              & 2.901422e+25 & 0.999964\\
        43 & -              & 1.160569e+26 & 0.999986\\
        44 & -              & 4.642275e+26 & 0.99999 \\
        45 & -              & 1.856910e+27 & 0.999972\\
        46 & -              & 7.427640e+27 & 1       \\
        47 & -              & 2.971056e+28 & 1\\
        48 & -              & 1.188422e+29 & 0.999998 \\
        49 & -              & 4.753690e+29 & 1\\
        50 & -              & 1.901476e+30 & 1 \\
    \end{tabular}
    \end{tabular}
    \caption{The number of $2$-generated $1$-relation monoids with the $C(4)$ condition where the maximum length of a relation word is $n$. The values for $n\geq 14$ were obtained from a uniform sample of $1000$ pairs of words $(l, r)$ of length where $|l| = n$ and $|r| \in \{1, \ldots, n\}$.}
    \label{tab:converge}
\end{table}
The normal form algorithm from~\cite{overlaps2} is not stated explicitly in~\cite{overlaps2},
and it appears that the constants in the polynomial time preprocessing step are rather large; see~\cref{section-normal-forms} for further details. The purpose of this paper is to provide an explicit algorithm for computing normal forms in $C(4)$ monoids with sufficiently small coefficients to permit its practical use. 
If it is already known that the input presentation satisfies $C(4)$, and a certain decomposition of the relation words is known, then 
the time complexity of the algorithm we present is $O(|w| ^ 2)$ for input word $w$, and the space complexity is the sum of $|A|$ and the lengths of all of the relation words in $R$.  We will show that it is possible to show that the $C(4)$ condition holds, and that the required decomposition of the relation words can be found, in $O(N + n)$ time where $N$ is the sum of the lengths of the relation words, and $n$ is the number of relation words in \cref{section-uniform-wp}. 

In \cref{prerequisites}, we present some necessary background material, and establish some notation. In \cref{section-uniform-wp}, we show that it is possible to determine the greatest $n\in \mathbb{N}$ such that a presentation $\mathcal{P}$ satisfies $C(n)$ in linear time in the sum of the lengths of the relation words in $\mathcal{P}$ using Ukkonen's Algorithm~\cite{Ukkonen}.  
In \cref{section-normal-forms}, we discuss the normal form algorithm of Kambites from~\cite{overlaps2}. In \cref{section-possible-prefix}, we describe, prove correct, and analyse the complexity of, a subroutine that is required in the practical normal form algorithm that is the main focus of this paper. 
Finally, in \cref{section-main}
we present our normal form algorithm, prove that it is correct, and analyse its complexity. 

The algorithm for solving the word problem in $C(4)$ monoids given in~\cite{kambites2011}, and the main algorithm from the present paper, were implemented by the authors in the C++ library \textsf{libsemigroups}~\cite{libsemigroups}.

\section{Prerequisites}\label{prerequisites}
\addtocounter{subsection}{1}
In this section we provide some of the prerequisites for understanding small overlap conditions and properties of small overlap monoids.
	
Let $A$ be a non-empty set, called an \textit{alphabet}. 
A word $w$ over $A$ is a finite sequence $w=a_0a_1 \cdots a_m$, $m\geq 0$ of elements of $A$.
The set of all words (including the empty word, denoted by $\varepsilon$) over $A$ with concatenation of words is called the \textit{free monoid} on $A$ and is denoted by $A^*$.
A \textit{monoid presentation} is a pair $\left\langle A \, |\, R \right\rangle$  where $A$ is an alphabet and $R\subseteq A^* \times A^*$ is a set of \textit{relations} on $A^*$. A monoid $M$ is defined by the presentation
$\left\langle A \, | \,R \right\rangle$ if $M$ is isomorphic to $A^*/R^{\#}$ where $R^{\#} \subseteq A^* \times A^*$ is the least congruence on $A^*$ containing $R$.
A \textit{finitely presented monoid} is any monoid defined by a presentation $\left\langle A \, | R \right\rangle$ where $A$ and $R$ are finite, and such a presentation is called a \textit{finite monoid presentation}.
For the rest of the paper, $\mathcal{P}=\left\langle A\,|\,R \right\rangle $ will denote a finite monoid presentation where 
\[
R = \{(W_0, W_1), (W_2, W_3), \ldots, (W_{n - 2}, W_{n - 1})\}.
\] 

If $s, t\in A ^ *$ are such that there exist $x_i, y_i\in A ^ *$ and $(W_i, W_{i + 1})$ or $(W_{i + 1}, W_i)\in R$ with $s = x_iW_iy_i$ and $t = x_{i}W_{i + 1}y_i$, then we write $s\to t$. 
If there exists a sequence of words $s=w_0, w_1, \ldots, w_n=t$ such that $w_i \to w_{i + 1}$ for all $i\in \{0, \ldots, n-1\}$, then we write $s\stackrel{*}{\rightarrow}t$ and we refer to such a sequence as a \textit{rewrite sequence}.
It is routine to verify that $(s, t) \in R ^ {\#}$ if and only if $s\stackrel{*}{\rightarrow} t$ .

We say that a relation word $V$ is a \textit{complement} of a relation word $W$ if there are relation words $V = r_0,r_1, \ldots, r_{n-1} = W$ such that either $(r_i,r_{i+1}) \in R$ or $(r_{i+1},r_i) \in R$ for $0 \leq i \leq n-1$. We say that a complement $V$ of $W$ is a \textit{proper complement} of $W$ if $V\neq W$. The equivalence relation defined by the complements of the relation words is a subset of the congruence $R ^ \#$.
We will write $u \equiv v$ to indicate that the words $u, v \in A^*$ represent the same element of the monoid presented by $\mathcal{P}$ (i.e. that $u/R^\# = v/R^\#$).

The \textit{relation words} of the presentation $\mathcal{P}$ are $W_0, W_1, \ldots, W_{n-1}$. A word $p \in A^*$ is called a \textit{piece} if it occurs as a factor of $W_i$ and $W_j$ where $W_i\not=W_j$, or in two different places (possibly overlapping) in the same relation word in $R$.  Note that the definition allows for the case when $i\not= j$ but $W_i = W_j$. In this case, neither $W_i$ nor $W_j$ is considered a piece (unless for other reasons), because although $W_i$ is a factor of $W_j$, it is not the case that $W_i \neq W_j$. 
By convention the empty word $\varepsilon$ is always a piece.

\begin{de}[cf. \cite{overlaps1}]\label{de-c-n}
 We say that a monoid presentation satisfies the condition $C(n), \, n \in \mathbb{N}$, if no relation word can be written as the product of strictly less than $n$ pieces. 
	The condition $C(1)$ describes those presentations where no relation word is equal to the empty word.
\end{de}

Having given the definition of the condition $C(4)$, we suppose for the remainder of the paper, that our fixed presentation $\mathcal{P} = \langle A| R\rangle$ satisfies the condition $C(4)$.

The following terms are central to the algorithms for $C(4)$ in~\cite{overlaps2, kambites2011} and are used extensively throughout the current paper. 

We say that $s \in A^*$ is a \textit{possible prefix} of a word $w\in A ^ *$ if $s$ is a prefix of some word $w_0\in A ^ *$ such that $w \equiv w_0$. The \textit{maximal piece prefix} of $u$ is the longest prefix of $u$ that is also a piece; we denote the maximal piece prefix of $u$ by $X_u$.
      The \textit{maximal piece suffix} of $u$, $Z_u$, is the longest suffix of $u$ that is also a piece; denoted $Z_u$.
      The word $Y_u$ such that $u=X_uY_uZ_u$ is called the \textit{middle word} of $u$. Since $u$ is a relation word of a presentation satisfying condition $C(4)$, $u$ cannot be written as a product of three pieces, and so the middle word $Y_u$ of $u$ cannot be a piece. In particular, the only relation word containing $Y_u$ as a factor is $u$.

Using the above notation every relation word $u$ in a $C(4)$ presentation can be written as a product of the form $X_uY_uZ_u$. Assume that $\overline{u}$ is a complement of $u$. Then $\overline{u}=X_{\overline{u}}Y_{\overline{u}}Z_{\overline{u}}$. We will write $\overline{X_u}$ instead of $X_{\overline{u}}$, $\overline{Y_u}$ instead of $Y_{\overline{u}}$, and $\overline{Z_u}$ instead of $Z_{\overline{u}}$. We say that $\overline{X_u}$ is a \textit{complement} of $X_u$, $\overline{Y_u}$ is a complement of $Y_u$ and similarly for $Z_u$, $X_uY_u$ and $Y_uZ_u$.

A  prefix of $w\in A ^ *$ that admits a factorization of the form $aXY$, for $XYZ$ a relation word, $a \in A^*$ and $X$ and $Y$ the maximal piece prefix and middle word of $XYZ$ respectively, is called a \textit{relation prefix}. If $w \in A ^ *$ has relation prefixes $aXY$ and $a'X'Y'$ such that $|aXY| = |a'X'Y'|$ for some $a, a'\in A ^*$, then $a = a'$, $X = X'$, and $Y = Y'$ as a direct consequence of the $C(4)$ condition.
A relation prefix of the form $p = bX_0Y_0'X_1Y_1'\cdots X_{n-1}Y_{n-1}'X_nY_n$, $n \geq 1$ and $b\in A ^ *$, is called an \textit{overlap prefix} if it satisfies the following:
\begin{enumerate}[\rm (i)]
	\item 
	$Y_i'$ is a proper non-empty prefix of the middle word $Y_i$ of some relation word $X_iY_iZ_i$; and 
	
	\item there does not exist a factor in $p$ of the form $X_mY_m$ beginning before the end of  $b$.
\end{enumerate} 
A relation prefix $aXY$ of a word $u$ is called a \textit{clean relation prefix of $u$} if $u$ does not have a prefix of the form $aXY'X_0Y_0$, where $Y'$ is a proper, non-empty prefix of $Y$.
 An overlap prefix of $u$ that is also a clean relation prefix is called a \textit{clean overlap prefix of $u$}.
If $p$ is a piece, then the word $u$ is called \textit{$p$-active} if $pu$ has a relation prefix $aXY$ for some $a\in A^ *$ such that $|a|<|p|$.

\begin{ex}
Suppose that $$\mathcal{P}=\left\langle a,b,c,d \, | \, a^2bc = acba,\, adca = bd^2b \right\rangle.$$ 
The relation words are $W_0 = a^2bc, W_1 = acba, W_2 = adca, W_3 = bd^2b$ and the set of pieces of this presentation is $P =\{\varepsilon, a, b, c, d\}$.
Then $W_0$ is a proper complement of $W_1$, and $W_2$ is a proper complement of $W_3$. Since none of the relation words can be written as the product of less than 4 pieces, $\mathcal{P}$ is a $C(4)$ presentation. We have $X_{W_0} = a, \, Y_{W_0} = ab, \, Z_{W_0} = c, \, X_{W_1} = a, \, Y_{W_1} = cb, \, Z_{W_1} = a$, $X_{W_2} = a, \, Y_{W_2} = dc, \, Z_{W_2} = a$ and $X_{W_3} = b, \, Y_{W_3} = d^2, \, Z_{W_3} = b$.

Let $w = cba^2bd^2a$. The word $w$ has two relation prefixes: $cba^2b = cbX_{W_0}Y_{W_0}$ and $cba^2bd^2 = cba^2X_{W_3}Y_{W_3}$. The relation prefix $cbX_{W_0}Y_{W_0}$ is not clean since $w$ has a prefix of the form $cbX_{W_0}Y_{W_0}'X_{W_3}Y_{W_3}$, where $Y_{W_0}' = a$. In addition, $cbX_{W_0}Y_{W_0}'X_{W_3}Y_{W_3}$ is an overlap prefix since there is no factor of the form $X_{W_i}Y_{W_i}$ beginning before the end of $cb$.
Let $p = a$. The word $w$ is $p$-active since $pw$ has the relation prefix $X_{W_1}Y_{W_1}$ and clearly $|\varepsilon| < |p|$. 
\end{ex}

The following results describe some properties of presentations satisfying the $C(4)$ condition mentioned in \cite{overlaps1} as weak cancellativity properties.

\begin{prop}[Proposition 1 in \cite{overlaps1}]
	Let $w$ be a word in $A^*$ and $aX_0Y_0'X_1Y_1'\dots X_nY_n$ be an overlap prefix of $w$. Then there is no relation word contained in this prefix except possibly $X_nY_n$, in case $Z_n=\varepsilon$.
\end{prop}

	In a word $w \in A ^ *$, an overlap prefix $aX_0Y_0'X_1Y_1'\dots X_tY_t$ is always contained in some clean overlap prefix $aX_0Y_0'X_1Y_1'\dots X_sY_s$ for $s \geq t$. In addition, if a word has a relation prefix, then the shortest relation prefix will be an overlap prefix.  If a word $w$ contains a relation word $u$ as a factor, then it has a relation prefix of the form $aX_uY_u$ for some $a \in A ^ *$. It follows that it also has an overlap prefix, this is its shortest relation prefix. Since any overlap prefix is contained in a clean overlap prefix it also follows that $w$ has a clean overlap prefix. Hence, taking the contrapositive, if a word in $A^*$ does not have a clean overlap prefix, then it contains no relation words as factors.

	If $aXYZ$ is a prefix of a word $w$, with $aXY$ an overlap prefix and $XYZ$ a relation word, then $aXY$ is a clean overlap prefix of $w$.
	If $aXY$ is a clean overlap prefix of a word $w$ and $\overline{XY}$ is a complement of $XY$, then $aXY$ and $a\overline{XY}$ are not necessarily clean overlap prefixes of words equivalent to $w$. However, it can be shown that such clean overlap prefixes are always overlap prefixes of words equivalent to $w$.
	\begin{lemma}[Lemma 2 in \cite{overlaps1}]
	\label{lem-2}
		If a word $w \in A^*$ has clean overlap prefix $aXY$ and $w\equiv v$ for some $v \in A^*$, then $v$ either has  $aXY$ or $a\overline{XY}$ for $\overline{XY}$ some complement of $XY$ as an overlap prefix; and no relation word in $v$ overlaps this prefix, unless it is $XYZ$ or $\overline{XYZ}$.
	\end{lemma}

In \cite{kambites2011}, Kambites describes an algorithm that takes as input two words and a piece of a given presentation that satisfies condition $C(4)$ and returns Yes if the words are equivalent and the piece is a possible prefix of the words and No if either of these does not happen. We will refer to this algorithm in the following sections of this paper and \textbf{WpPrefix}$(u,v,p)$ will be used to denote the result of this algorithm with input the words $u$ and $v$ and the piece $p$. In \cite{kambites2011} it is shown that for a fixed $C(4)$ presentation, this algorithm decides whether $u$ and $v$ are equivalent and whether $p$ is a possible prefix of $u$ in time $O(\text{min}(|u|, |v|))$ given the decomposition of relation words into the form $XYZ$ is known.

\section{A linear time algorithm for the uniform word problem}
\label{section-uniform-wp}
\addtocounter{subsection}{1}

In \cite{overlaps1}, Kambites explores the complexity of the so-called  \textit{uniform word problem} for $C(4)$ presentations. 
Given a finite monoid presentation $\left\langle A \, | \, R \right\rangle$ and two words in $A^*$, the uniform word problem asks whether the two words represent the same element of the monoid defined by $\left\langle A \, | \, R \right\rangle$. In particular, determining that the presentation satisfies the $C(4)$ condition is part of the uniform word problem for $C(4)$ presentations.
In \cite{overlaps1} it is shown that the uniform word problem for $C(4)$ presentations can be solved, in the RAM model of computation, in $O(|R|^2 \min(|u|, |v|))$ time for $u, v$ the two words in $A^*$ and $|R|$ the sum of the lengths of the distinct relation words in the presentation. 

We will show that the uniform word problem for $C(4)$ presentations can be solved in $O(|R| \min(|u|, |v|))$ time, where $|R|$ is the sum of the lengths of the relation words, by using a generalized suffix tree to represent the relation words. In the RAM model, we may assume that the following operations are constant time: random access to the letters of a word $w\in A ^ *$ from an index; concatenation of words; comparison of letters from $A$ for a given total order on $A$.

Assume $A$ is an alphabet and $s = a_0a_1\cdots a_{m - 1}\in A^*$ such that $|s| = m$. We use the notation $s[i,j)$ for the factor $a_i\ldots a_{j - 1}$ of $s$ that starts at position $i$ and ends at position $j - 1$ (inclusive). A word of length $m$ has $m$ non-empty suffixes $s[0, m), \ldots, s[m - 1, m)$. A word $x$ is a factor of $s$ if and only if it is a prefix of one of the suffixes of $s$.

\begin{de}[cf. Section 5.2 in \cite{suffix_trees}]
A \textit{suffix tree} for a word $s$ of length $m$ is a rooted directed tree with exactly $m$ leaf nodes numbered $0$ to $m - 1$. 
The nodes of a suffix tree are of exactly one of the following types:
the root; a leaf node; or an internal node.
Each internal node has at least two children. Each edge of the tree is labelled by a nonempty factor of $s$ and no two edges leaving a node are labelled by words that begin with the same character. For any leaf $i$, $0 \leq i < m$ the label of the path that starts at the root and ends at leaf $i$ is $s[i,m)$.
\end{de}

A \textit{generalized suffix tree} is a suffix tree for a sequence of words $S = \{s_0, s_1, \dots , s_{n-1}\}$. In a generalized suffix tree the leaf nodes are numbered by ordered pairs $(i, j)$ for $0 \leq i \leq n-1$ and $0 \leq j \leq |s_i|$. The label of the path that starts at the root node and ends at leaf $(i,j)$ is $s_i[j,m)$ for $m = |s_i|$. In a generalized suffix tree, a special unique character $\$_i$ is attached at the end of each word $s_i$ to ensure that each suffix corresponds to a unique leaf node in the tree. Thus a generalized suffix tree has exactly $N+n$ leaf nodes, where $N$ is the sum of the lengths of the words in $S$.  See \cref{figure-gst} for an example of generalized suffix tree. 

A suffix tree for a word $w$ over an alphabet $A$ of length $m$ can be constructed in $O(m)$ time for constant-size alphabets and in $O(m \log|A|)$ time in the general case with the use of Ukkonen's algorithm~\cite{Ukkonen}.  
If $N$ is the sum of the lengths of the words in the set of words $S$, then, a generalized suffix tree for $S$ can also be constructed in $O(N+n)$ time; see~\cite[Section 6.4]{suffix_trees} for further details. 

A generalized suffix tree for a set of words $S$ of total length $N$ has at most $2(N+n)$ nodes. By definition, such a suffix tree has exactly $N+n$ leaf nodes and one root. In addition, each internal node of the tree has at least two edges leaving it. These edges belong to paths that will eventually terminate at some leaf node. Hence there can exist at most $N+n-1$ internal nodes and the total number of nodes in a suffix tree is at most $2(N+n)$.

Generalized suffix trees can be constructed and queried in linear time to provide various information about a set of words. 
For example, for a sequence of $n$ words $S$ of total length $N$ we can find the longest subwords that appear in more than one word in $O(N+n)$ time, find the longest common prefix of two strings in $O(N+n)$ time, check if a word of length $m$ is a factor of some word in $S$ in $O(m)$ time; see Sections~2.7 to~2.9 in \cite{suffix_trees}. We are interested in utilizing generalized suffix trees in the study of $C(4)$ presentations. Since generalized suffix trees can be queried to find the longest subwords that appear in more than one word in the set of distinct relation words in $R$, we can use them to find maximal piece prefixes, for example. In order to do this, however, we need to build the generalized suffix tree of distinct relation words of a presentation since a word $p$ is a piece if it occurs as a factor of $W_i$ and $W_j$ where $W_i\not=W_j$, or in two different places (possibly overlapping) in the same relation word in $R$.

Given the set of relation words $R$ of a presentation we want to construct the generalized suffix tree of the set of distinct relation words in $R$ without altering the complexity of Ukkonen's algorithm. In practice, Ukkonen's algorithm for the generalized suffix tree of the set $S = \{s_0\$_0, s_1\$_1, \ldots s_{n-1}\$_{n-1}\}$ starts by constructing the suffix tree $T$ for $s_0\$_0$. Then, for each $i \in \{1, \ldots, n-1\}$, the edges and nodes that correspond to the word $s_i$ are added to $T$; see Section 6.4 in \cite{suffix_trees} for more detail. We denote this step of the procedure by \texttt{AddWord}($T, s_i$). In order to avoid adding the same word twice, we add the following step to the procedure for each $i \in \{1, \ldots, n-1\}$: before calling \texttt{AddWord}($T, s_i$), we follow the path in $T$ that starts at the root node and is labeled by $s_i$. If this path ends at an internal node $\nu$ and one of the children of $\nu$ is a leaf node labeled by $(s_j, 0)$ for some $j$, then $s_j = s_i$ and hence we do not call \texttt{AddWord}($T, s_i$) for $s_i$. This additional step only requires traversing at most $|s_i| + 1$ nodes of $T$ for each $i \in \{1, \ldots, n-1\}$. The total number of nodes of $T$ is bounded above by $2(N+n)$ and hence this step does not alter the complexity of the construction of the generalized suffix tree.

\begin{prop}\label{pieces-complexity}
Let $\left\langle A \, | \, R \right\rangle$ be a finite monoid presentation such that the number of relation words in $R$ is $n$ for some $n\in \N$,  and let $N$ be the sum of the lengths of the relation words in $R$. Then from the input presentation $\left\langle A \, | \, R \right\rangle$ the set of maximal piece prefixes, and suffixes, of $R$ can be computed in $O(N + n)$ time.
\end{prop}

\begin{proof}
Using Ukkonen's Algorithm, for example, a generalized suffix tree for the set $\{W_0\$_0, W_1\$_1, \ldots,$ $W_{n-1}\$_{n-1}\}$ of distinct relation words in $R$ can be constructed in $O(N + n)$ time.
The maximal piece prefix of a relation word $W_r\in R$ can be found as follows. The path in the tree labelled by $W_r\$_r$ is followed from the root to the (unique) leaf node $(r, 0)$. Suppose that 
 $v_0, v_1, \ldots, v_m$ are the nodes in the path from the root node $v_0$ to the leaf node $v_m =(r, 0)$. Then the maximal piece prefix of $u_r$ corresponds to the path $v_0, v_1, \ldots, v_{m-1}$. In other words, the maximal piece prefix of $W_r$ corresponds to the internal node that is the parent of the leaf node labelled $(r, 0)$. 
Hence the maximal piece prefix of each relation word can be determined in $O(|W_r|)$ time, and so every maximal piece prefix can be found in $O(N+n)$ time. 

The maximal piece suffixes of the relation words can be found as follows. A generalized suffix tree for the 
set $\Tilde{R}$ of reversals of the relation words in $R$ can be constructed in $O(N + n)$ time, and then used, as described above, to compute the maximal piece prefixes of the reversed relation words in $O(N + n)$ time also. 

Alternatively, the generalized suffix tree for $R$ can be used to directly compute the maximal piece suffix of a given relation word $W_r$ by finding the maximum distance, from the root, of any internal node $n$ that is a parent of a leaf node labelled $(r, i)$ for any $i$. This maximum is the length of the maximal piece suffix of $W_r$. In this way, the maximum piece suffix of every relation word $W_r$ can be found in a single traversal of the nodes in the tree. 
Again, since there are $2(N + n)$ nodes in the tree, and the checks on each node can be performed in constant time, the maximal piece suffixes of all the relation words can be found in $O(N + n)$ time using this approach also.
\end{proof}

For example, the generalized suffix tree for the relation words in the presentation \[\langle a, b, c, d, e \mid a^2ea^3 = abcd\rangle\] can be seen in \cref{figure-gst}.

Using the approach described in the proof of \cref{pieces-complexity}, the path from the root $0$ (corresponding to $a^2ea^3$) to the leaf node labelled by $(0, 0)$ consists of the root $0$, internal nodes $1$ and $2$,  and leaf node $(0,0)$. Hence the maximal piece prefix of $a^2ea^3$ is $aa$, being the label of the path from the root to the parent $2$ of the leaf node $(0,0)$. Similarly, the maximal piece prefix of $abcd$ is $a$ corresponding to the parent $1$ of the node $(1,0)$. For the maximal piece suffix of $a^2ea^3$, the leaf nodes labelled $(0, i)$ for any $i$ with edge labelled by $\$_0$ are $(0,4)$, $(0,5)$, and $(0,6)$. The parents of these nodes are $2$, $1$, and $0$, respectively, and hence the maximal piece suffix of $a^2ea^3$ is $aa$.
The only leaf node labelled $(1, i)$ and with edge labelled $\$_1$ is $(1,4)$, and so the maximal piece suffix of $abcd$ is $\varepsilon$.

\begin{prop}\label{decideC4-complexity}
Let $\left\langle A \, | \, R \right\rangle$ be a finite monoid presentation such that the number of relation words in $R$ is $n$ for some $n\in \N$ and let $N$ be the sum of the lengths of the relation words in $R$. Then from the input presentation 
$\left\langle A \, | \, R \right\rangle$ it can be determined whether or not the presentation satisfies $C(4)$ in $O(N + n)$ time.
\end{prop}

\begin{proof}
In order to decide whether the presentation satisfies $C(4)$ we start by computing the maximal piece prefix $X_r$ and the maximal piece suffix $Z_r$ for each relation word $W_r$. By \cref{pieces-complexity}, this step can be performed in $O(N + n)$ time. The presentation is $C(4)$ if for every relation word $W_r$, $|X_r| + |Z_r| < |W_r|$ and the middle word $Y_r$ is not a piece. It suffices to show that it can be determined in $O(N + n)$ time whether or not $Y_r$ is a piece for every $r$. 
 
Any word $w\in A ^ *$ is a piece if and only if $w$ equals the longest prefix of $w$ that is a piece. In the proof of \cref{pieces-complexity}, we showed how to compute the maximal piece prefix of the relation words in $R$ using a generalized suffix tree in $O(N + n)$ time. The longest prefix of $Y_r$ that is a piece can be determined in $O(|Y_r|)$ time by finding the last internal node $\nu$ on the path from the root of the same generalized suffix tree labelled by $Y_r$; the node $\nu$ is the parent of the leaf node $(r, |X_r|)$. If $\nu$ is the root node, then the longest prefix of $Y_r$ that is a piece is $\varepsilon$. Otherwise, the longest prefix of $Y_r$ that is a piece is the label of the path from the root node to $\nu$.
 Hence determining whether or not every $Y_r$ is a piece can also be completed in total $O(N + n)$ time.
\end{proof}

If any of the words in $R$ is empty, then the presentation $\left\langle A\,|\, R\right\rangle$ is not $C(4)$. If all of the words are non-empty, then the number of relation words $n$ is bounded above by the sum of the lengths of the relation words $N$. Hence the $O(N + n)$ time complexity in \cref{pieces-complexity} and~\cref{decideC4-complexity} becomes $O(N)$. 

The presentation 
 $\langle a, b, c, d, e \mid a^2ea^3 = abcd\rangle$
can be seen to be $C(4)$ as follows. The only internal node on the path from the root of the suffix tree depicted in \cref{figure-gst} labelled by $ea$ is the root itself. Hence the maximal piece prefix of $ea$ is $\varepsilon$ and so $ea$ is not a piece. Similarly, the only internal node on the path from the root labelled $bcd$ is the root itself, and so $bcd$ is not a piece either. Hence, by the proof of \cref{decideC4-complexity}, the presentation 
 $\langle a, b, c, d, e \mid a^2ea^3 = abcd\rangle$ is $C(4)$.

\cref{decideC4-complexity} allows us to prove the following theorem.
\begin{thm}\label{wp-complexity}
Let $\left\langle A \, | \, R \right\rangle$ be a finite monoid presentation such that the number of relation words in $R$ is $n$ for some $n\in \N$, let $N$ be the sum of the lengths of the relation words in $R$, and let $u, v\in A ^ *$ be arbitrary. Then the uniform word problem with input the presentation $\left\langle A \, | \, R \right\rangle$, and the words $u$ and $v$ can be solved in $O((N + n) \min(|u|,|v|))$ time.
\end{thm}
\begin{proof}
Given \cref{decideC4-complexity}, 
the proof of this theorem is essentially identical to the proof of~\cite[Theorem 2]{overlaps1}. 
\end{proof}

\begin{landscape}
\begin{figure}
\resizebox{720pt}{!}{%
\begin{tikzpicture}[level distance=3cm, sibling distance=3cm]
\node(0){\color{color0}0}
child{node(2){\color{color0}1}
child{node(5){\color{color0}2}
child{node(6){\color{color0}(0, 3)}
edge from parent [color=color0] node[above, sloped, color0, rotate=180]{$s_0[5, 7)$} node[below, sloped, color0, rotate=180] {$a\$_0$}
}
child{node(1){\color{color0}(0, 0)}
edge from parent [color=color0] node[above, sloped, color0]{$s_0[2, 7)$} node[below, sloped, color0] {$eaaa\$_0$}
}
child{node(7){\color{color0}(0, 4)}
edge from parent [color=color0] node[above, sloped, color0]{$s_0[6, 7)$} node[below, sloped, color0] {$\$_0$}
}
edge from parent [color=color0] node[above, sloped, color0, rotate=180]{$s_0[1, 2)$} node[below, sloped, color0, rotate=180] {$a$}
}
child{node(10){\color{color1}(1, 0)}
edge from parent [color=color1] node[above, sloped, color1, rotate=180]{$s_1[1, 5)$} node[below, sloped, color1, rotate=180] {$bcd\$_1$}
}
child{node(3){\color{color0}(0, 1)}
edge from parent [color=color0] node[above, sloped, color0]{$s_0[2, 7)$} node[below, sloped, color0] {$eaaa\$_0$}
}
child{node(8){\color{color0}(0, 5)}
edge from parent [color=color0] node[above, sloped, color0]{$s_0[6, 7)$} node[below, sloped, color0] {$\$_0$}
}
edge from parent [color=color0] node[above, sloped, color0, rotate=180]{$s_0[0, 1)$} node[below, sloped, color0, rotate=180] {$a$}
}
child{node(11){\color{color1}(1, 1)}
edge from parent [color=color1] node[above, sloped, color1, rotate=180]{$s_1[1, 5)$} node[below, sloped, color1, rotate=180] {$bcd\$_1$}
}
child{node(12){\color{color1}(1, 2)}
edge from parent [color=color1] node[above, sloped, color1, rotate=180]{$s_1[2, 5)$} node[below, sloped, color1, rotate=180] {$cd\$_1$}
}
child{node(13){\color{color1}(1, 3)}
edge from parent [color=color1] node[above, sloped, color1]{$s_1[3, 5)$} node[below, sloped, color1] {$d\$_1$}
}
child{node(4){\color{color0}(0, 2)}
edge from parent [color=color0] node[above, sloped, color0]{$s_0[2, 7)$} node[below, sloped, color0] {$eaaa\$_0$}
}
child{node(14){\color{color1}(1, 4)}
edge from parent [color=color1] node[above, sloped, color1]{$s_1[4, 5)$} node[below, sloped, color1] {$\$_1$}
}
child{node(9){\color{color0}(0, 6)}
edge from parent [color=color0] node[above, sloped, color0]{$s_0[6, 7)$} node[below, sloped, color0] {$\$_0$}
}
;\end{tikzpicture}
}
\caption{The generalized suffix tree for the words ${\color{color0}s_0 = a^2ea^3\$_0}$ and ${\color{color1}s_1 = abcd\$_1}$.}
\label{figure-gst}
\end{figure}
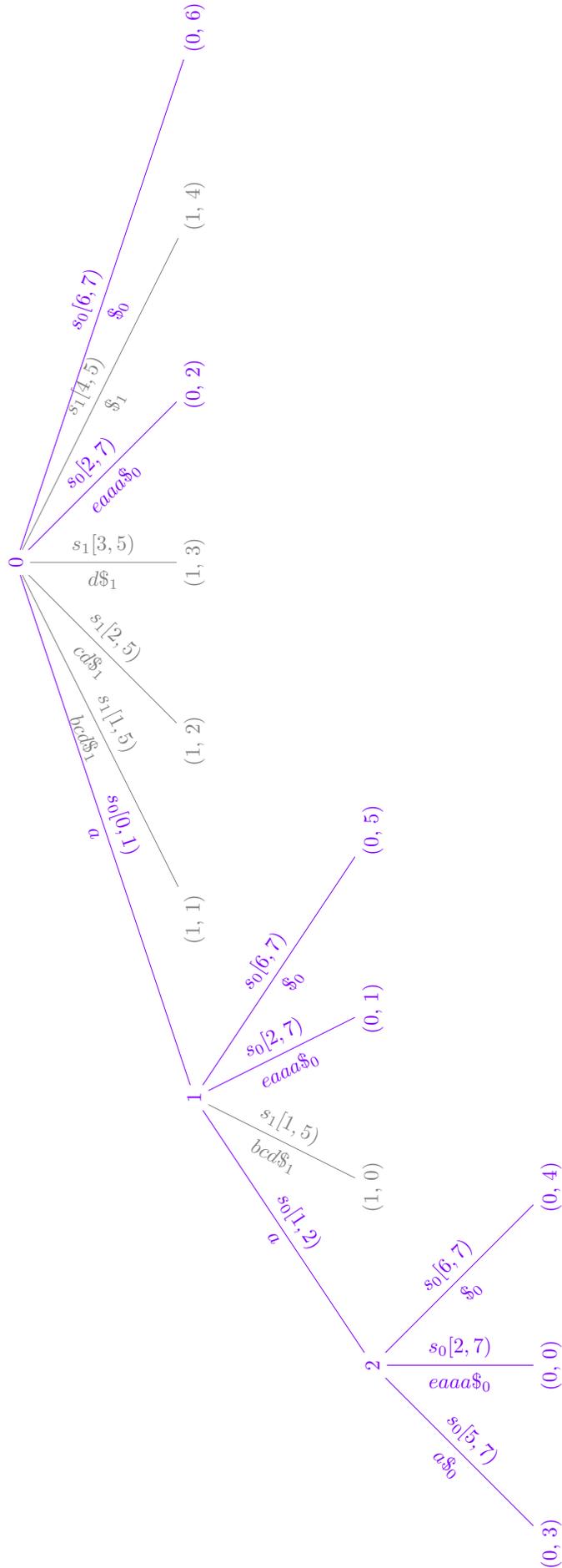
\end{landscape}

\section{Kambites' normal form algorithm}\label{section-normal-forms}
\addtocounter{subsection}{1}
 Let $A=\left\{ a_0,a_1,\dots,a_{n-1} \right\rbrace$ be a finite alphabet and define a total order $<$ on the elements of $A$ by $a_0<a_1<\dots<a_{n-1}$. We extend this to a total order over $A^*$, called the \textit{lexicographic order} as follows. The empty word $\varepsilon$ is less than every other word in $A ^ *$.
 If $u=a_iu_0$ and $v=a_jv_0$ are words in $A^+$, $a_i,a_j \in A$ and $u_0,v_0 \in A^*$, then $u<v$ whenever $a_i<a_j$, or $a_i=a_j$ and $u_0<v_0$. 
 
As mentioned above, Kambites in \cite{kambites2011} described an algorithm for testing the equivalence of words in $C(4)$ monoids. In \cite{overlaps2} it was shown that given a monoid $M$ defined by a $C(4)$ presentation $\left\langle A \, | \, R \right\rangle$ and a word $w \in A^*$ there exists an algorithm that computes the minimum representative of the equivalence class of $w$ with respect to the lexicographic order on $A ^ *$. This minimum representative is also known as the \textit{normal form} of $w$. It is not, perhaps, immediately obvious that such a minimal representative exists, because the lexicographic order is not a well order (it is not true that every non-empty subset of $A ^ *$ has a lexicographic least element).
  However, every  equivalence class of a word in a $C(3)$ monoid is finite; see, for example, in \cite[Corollary 5.2.16]{higgins}. Since any presentation that satisfies $C(4)$ also satisfies $C(3)$, there exists a lexicographically minimal representative for any $w \in A^*$. We denote the lexicographically minimal word equivalent to $w\in A ^ *$ by $\min w$.
  
  We note that since all of the equivalence classes of a $C(3)$ monoid are finite, any monoid satisfying $C(n)$ for $n \geq 3$, is infinite. Conversely, if $\left\langle A \, | \, R \right\rangle$ is a presentation for a finite monoid and this presentation satisfies $C(n)$ for some $n \in \mathbb{N}$, then $n \in \{1, 2 \}$.
  
In \cite{overlaps2} Kambites proved the following result.

\begin{prop}[Corollary 3 in \cite{overlaps2}] \label{nf-algorithm}
Let $\left\langle A \, | \, R \right\rangle$ be a finite monoid presentation satisfying $C(4)$ and suppose that $A$ is equipped with a total
order. Then there exists an algorithm which, given a word $w \in A^*$, computes in $O(|w|)$ time the corresponding lexicographic normal form for $w$.
\end{prop}

Although \cref{nf-algorithm} asserts the existence of an algorithm for computing normal forms, this algorithm is not explicitly stated in~\cite{overlaps2}. In the following paragraphs we briefly discuss the algorithm arising from~\cite{overlaps2}.  

We require a number of definitions; see~\cite{Berstel} for further details. A \textit{transducer} $\mathcal{T}=\left\langle A, \, B, \, Q, \,q_{\_}, \, Q_{+}, \, E\right\rangle$ is a 6-tuple that consists of an input alphabet $A$, an output alphabet $B$, a finite set of states $Q$, an initial state $q_{\_}$, a set of terminal states $Q_+$ that is a subset of $Q$, and a finite set of transitions or edges $E$ such that $E \subset Q\times A^* \times B^* \times Q$. A pair $(u,v) \in A^* \times B^*$ is \textit{accepted} by the transducer if there exist transitions $(q_{\_}, u_0, v_0, q_1), (q_1, u_1, v_1, q_2), \dots,(q_{n-1}, u_{n-1}, v_{n-1}, q_+) \in E$ such that $q_+ \in Q_+$ and $u=u_0u_1 \cdots u_{n-1}$, $v=v_0v_1 \cdots v_{n-1}$. The \textit{relation} accepted by $\mathcal{T}$ is the set of all pairs accepted by $\mathcal{T}$.
A relation accepted by a transducer is called a \textit{rational relation}. A rational relation that contains a single pair $(u,v)$ for each $u \in A^*$ is called a \textit{rational function}.

A \textit{deterministic 2-tape finite automaton} is an 8-tuple $\mathcal{A}=\left\langle A, \, B, \, Q_1, \, Q_2, \,q_{\_}, \, Q_{+}, \, \delta_1, \, \delta_2\right\rangle$ that consists of the tape-one alphabet $A$, the tape-two alphabet $B$, two disjoint state sets $Q_1$ and $Q_2$, an initial state $q_{\_} \in Q_1 \cup Q_2$, a set of terminal states $Q_+ \subset Q_1 \cup Q_2$ and two partial functions $\delta_1: Q_1 \times A \cup \{\$\} \rightarrow Q_1 \cup Q_2$, $\delta_2: Q_2 \times B \cup \{\$\} \rightarrow Q_1 \cup Q_2$ where $\$$ is a symbol not in $A$ and $B$. A path of length $n$ in $\mathcal{A}$ is a sequence of transitions of the form $$(q_0,t_0, a_0, \delta_{t_0}(q_0,a_0))(\delta_{t_0}(q_0,a_0), t_1, a_1, \delta_{t_1}(q_1,a_1)) \cdots (\delta_{t_{n-2}}(q_{n-2},a_{n-2}), t_{n-1}, a_{n-1}, \delta_{t_{n-1}}(q_{n-1},a_{n-1})),$$ where $t_i \in \{1,2\}$, $q_i \in Q_{t_i}$, $a_i \in A$ if $t_i=1$ and $a_i \in B$ if $t_i=2$ for $0\leq i \leq n-1$ and $\delta_{t_i}(q_i,a_i)= q_{i+1}$ for $0\leq i < n-1$. A path is called \textit{successful} if $q_0 = q_{\_}$ and $\delta_{t_{n-1}}(q_{n-1},a_{n-1}) \in Q_+$. A pair $(u,v)\in A^* \times B^*$ is the label of a path if $u$ is the concatenation of all the letters $a_i$, $0 \leq i \leq n-1$ that belong in $A$ and $v$ is the concatenation of all the letters $a_i$, $0 \leq i \leq n-1$ that belong in $B$. A pair $(u,v)\in A^* \times B^*$ is \textit{accepted} by $\mathcal{A}$ if it labels a successful path in $\mathcal{A}$. The \textit{relation accepted by} $\mathcal{A}$ is the set of all pairs accepted by $\mathcal{A}$.

 Let $\lex(R ^ {\#}) = \{(u,v) \in A^* \times A^*\,|\, u \equiv v$ and $v$ is a lexicographic normal form$\}$. In \cite{overlaps2} it is shown that $R ^ {\#}$ is a rational relation and $\lex(R ^ {\#})$ is a rational function.
  According to Lemma 5.3 in~\cite{Johnson1}, $\lex(R ^ {\#})$ can effectively be computed from a finite transducer for $R ^ {\#}$. Kambites describes the construction of a finite transducer for $\lex(R ^ {\#})$ in \cite{overlaps2}. The steps for constructing this transducer are the following:
 
  \begin{itemize}
     \item Starting from the $C(4)$ presentation $\langle A \mid R \rangle$, an abstract machine called a 2-tape deterministic prefix rewriting automaton with bounded expansion can be computed. The construction is given in the proof of Theorem 2 in \cite{overlaps2}. The relation accepted by this automaton is $R ^ {\#}$.
     \item Using the construction in the proof of Theorem 1 in \cite{overlaps2}, the 2-tape deterministic prefix rewriting automaton can be used to construct a transducer $\mathcal{T}$ realizing $R ^ {\#}$.
  \end{itemize}
  
  Let $\delta$ be the length of the longest relation word in $R$ and let $P$ be the set of  pieces of the presentation. The set $A^{\leq k}$ for $k\in \mathbb{N}$ consists of all words in $A^*$ with length less or equal to $k$. Similarly, $A^{< k}$ consists of all words in $A^*$ with length less than $k$. In addition, let $\$$ be a new symbol not in $A$. The set $A^{< k}\$$ consists of words $u\$$ such that $u \in A^{< k}$.
  
  The state set of the transducer $\mathcal{T}$, given in the proof of Theorem 1 in~\cite{overlaps2}, is the set $C \times C \times P$ where $C$ is the set $$A^{\leq3\delta} \cup A^{<3\delta}\$.$$ Hence, the number of states of the transducer is extremely large even for relatively small presentations. 
  
  For example, let $\left\langle a, b, c \,|\, a^2bc=acba \right\rangle $ be the presentation. In this case $|A|=3, \delta = 4$ and $P = \{ \varepsilon, a, b, c\}$. The size of the state set $Q = C \times C \times P$ of the corresponding transducer is $|C|^2 \cdot |P| = 4|C|^2$. Since $C = A^{\leq 12}\cup A^{<12}\$$, $$|C| = \sum_{i=0}^{12} 3^{i} + \sum_{i=0}^{11} 3^{i} = 1062881$$ and $|Q| = 4518864080644$.
  
  Another approach arising from~\cite{overlaps2} for the computation of normal forms is the construction of a deterministic 2-tape automaton accepting $\lex(R ^ {\#})$. This also begins by constructing the transducer $\mathcal{T}$.
 The process arising from~\cite{overlaps2} for the construction of the automaton is: perform the two steps given above to construct the transducer $\mathcal{T}$, then:
 \begin{itemize}
     \item using the construction in the proof of Proposition 1 in \cite{overlaps2}, a deterministic 2-tape automaton accepting $R ^ {\#}$ can be constructed starting from the transducer $\mathcal{T}$;
     \item the proof of Theorem 5.1 in \cite{Johnson2} describes the construction of a deterministic 2-tape automaton accepting $\lex(R ^ {\#})$, starting from the deterministic 2-tape automaton that accepts $R ^ {\#}$.
 \end{itemize}
 
 The state set $Q = Q_1 \cup Q_2$ of the 2-tape automaton that accepts $R ^ {\#}$ in the second step is the same as the state set of the transducer $\mathcal{T}$, partitioned in two disjoint sets $Q_1$ and $Q_2$. The state set $Q' = Q_1' \cup Q_2'$ of the 2-tape automaton that accepts $\lex(R ^ {\#})$ is the union of the sets $Q_1'=Q_1 \times 2^{Q_1}$ and $Q_2'=Q_2 \times 2^{Q_1}$, hence the number of states of this automaton is greater than the number of states of the transducer $\mathcal{T}$.
 
Although the approach described in \cite{overlaps2} allows normal forms for words in a $C(4)$ presentation to be found in linear time, 
it is impractical to use a transducer with such a large state set.
The current article arose out of a desire to have a practical algorithm for computing normal forms in $C(4)$ monoids.

\section{Possible prefix algorithm}
\label{section-possible-prefix}
\addtocounter{subsection}{1}
Before describing the procedure for finding normal forms,  we describe an algorithm that takes as input a word $w_0$ and a possible prefix piece $p$ of $w_0$ and returns a word equivalent to $w_0$ with prefix $p$.
As mentioned above the algorithm \textbf{WpPrefix}, described in \cite{kambites2011} can decide whether a piece $p$ is a possible prefix of some word $w_0$ by calling \textbf{WpPrefix}$(w_0,w_0,p)$. 

\begin{algorithm}
\caption{- \textbf{ReplacePrefix}($w_0, p$)}
\label{algorithm-prefix}
\textbf{Input:} A word $w_0$ and a piece $p$ such that \textbf{WpPrefix}$(w_0,w_0,p)=$Yes.\\
\textbf{Output:} A word equivalent to $w_0$ with prefix $p$.\\
\noindent \textbf{ReplacePrefix}($w_0, p$):
\begin{algorithmic}[1]
    \IF{$w_0$ does not have prefix $p$ and $w_0=aXYw'$ with $aXY$ a clean overlap prefix}
     \STATE $u\leftarrow$\textbf{ReplacePrefix}$(w',Z)$ \text{with $Z$ deleted}
     \STATE $w_0\gets a\overline{XYZ}u$ where $\overline{XYZ}$ is a proper complement of $XYZ$ such that $p$ is a prefix of $a\overline{X}$
  \ENDIF
 \RETURN $w_0$
\end{algorithmic}
\end{algorithm}

 \begin{lemma}\label{coprefix}
Let $w \in A^*$ be arbitrary. If there exists a piece $p$ that is a possible, but not an actual, prefix of $w$, then the shortest relation prefix of $w$ is a clean overlap prefix.
\end{lemma}

\begin{proof}
Since $p$ is a possible prefix but not a prefix of $w$, $w$ contains at least one relation word and hence has a relation prefix. Let $aXY$ be the shortest relation prefix of $w$. Then $aXY$ is an overlap prefix.
If $aXY$ is not clean, then $w$ has a prefix of the form $aXY'X_0Y_0$ such that $Y'$ is a proper non-empty prefix of $Y$. Hence the shortest clean overlap prefix of $w$ contains $aXY'$ and hence $aXY'$ is also a prefix of every $v$ such that $v \equiv w$ by \cref{lem-2}. Let $w_0$ be a word equivalent to $w$ that has prefix $p$. Then either $p$ is a prefix of $aXY'$ or $p$ contains $aXY'$. In the former case this would mean that $p$ is also a prefix of $w$, which is a contradiction. In the latter case $XY'$ is a factor of $p$. Since $p$ is a piece this implies that $XY'$ is also a piece which is a contradiction since $X$ is the maximal piece prefix of the relation word $XYZ$. We have shown that the shortest relation prefix of $w$ is a clean overlap prefix, as required. 
\end{proof}

\begin{lemma}\label{find-relation-prefix}
Let $w\in A ^ *$ be arbitrary. If $w$ has a piece $p$ as a possible, but not an actual, prefix, then the shortest relation prefix of $w$ can be found in constant time, given the suffix tree for the relation words in $R$.
\end{lemma}
\begin{proof}
Suppose that $S = \{W_0, W_1, \dots, W_{n-1}\}$ is the set of relation words and let $\delta$ be the length of the longest relation word in $R$.  We want to find the shortest relation prefix $tX_{W_i}Y_{W_i}$ for some $t \in A^*$, and $W_i \in S$. Since $X_{W_i}Y_{W_i}$ and $p$ are factors of relation words, 
$|p|, |X_{W_i}Y_{W_i}| \leq \delta$. 
Since $tX_{W_i}Y_{W_i}$ is the shortest relation prefix of $w$, $t$ is prefix of every word equivalent to $w$, and hence $t$ is a proper prefix of $p$. In particular, $|t| < |p| < \delta$. If $|w| \geq 2\delta$, then we define $v$ to be the prefix of $w$ of length $2\delta$; otherwise, we define $v$ to be $w$. In order to find the shortest relation prefix of $w$ it suffices to find the shortest relation prefix of $v$. For a given presentation, the length of $v$ is bounded above by the constant value $2\delta$.

In practice, in order to find the shortest relation prefix of $v$, we construct a suffix tree for all words $X_{W_i}Y_{W_i}$ such that $X_{W_i}Y_{W_i}Z_{W_i}$ is a relation word of the presentation. This is done in $O(N+n)$ time, for $N$ the sum of the lengths of the relation words in the presentation. A factor of $v$ has the form $X_{W_i}Y_{W_i}$ for some $i$ if and only if this factor labels a path that starts at the root node of the tree and ends at some leaf node labelled by $(i,0)$. Hence the shortest relation prefix of $v$ can be found by traversing the nodes of the tree at most $|v|$ times. Since the length of $v$ is at most $2\delta$ this can be achieved in constant time.

The complexity of this procedure is $O((N+n)|v|) = O(2\delta(N+n))$ which is independent of the choice of $w$.
\end{proof}

Next, we will show that \cref{algorithm-prefix} is valid.
 
 \begin{prop}
 If $w_0, p \in A ^*$ are such that $p$ is piece and a possible prefix of $w_0$, then \textbf{ReplacePrefix($w_0, p$)} returns a word that is equivalent to $w_0$ and has prefix $p$ in $O(|w_0|)$ time, given the suffix tree for the relation words in $R$.
 \end{prop}	
 		
 \begin{proof}
 We will prove that the algorithm returns the correct result using induction on the number $k$ of recursive calls in line 2.  
 Note that if $p$ is a possible prefix of $w_0$ and $w_0$ contains no relation words, then $p$ is a prefix of $w_0$. On the other hand, if $p$ is not a prefix of $w_0$, then $w_0$ must contain a relation word, and hence a clean overlap prefix. 
 
We first consider the base case, when $k=0$. Let $p$ be a piece and $w_0$ a word such that \textbf{ReplacePrefix}($w_0,p$) terminates without making a recursive call. This only happens in case $p$ is already a prefix of $w_0$ and the algorithm returns $w_0$ in line 5. Hence when $k = 0$ the word returned by \textbf{ReplacePrefix}($w_0,p$) is $w_0$ and has prefix $p$.

Next, we let $k>0$ and assume that the algorithm returns the correct result when termination occurs after strictly fewer than $k$ recursive calls. Now let $p$ be a piece and $w_0$ a word such that \textbf{ReplacePrefix}($w_0,p$) terminates after $k$ recursive calls. It suffices to prove that the first recursive call returns the correct output.

If $p$ is already a prefix of $w_0$ a recursive call does not happen, hence we are in the case where $p$ is not a prefix of $w_0$. Since $p$ is a possible prefix of $w_0$, there exists a word that is equivalent but not equal to $w_0$ and that has $p$ as a prefix. This means that $w_0$ has a relation prefix and hence it has a clean overlap prefix of the form $aXY$. By \cref{lem-2}, every word equivalent to $w_0$ has $a\overline{XY}$ for $\overline{XY}$ a complement of $XY$, as a prefix. Hence since $p$ is not a prefix of $w_0$, $p$ must be a prefix of $a\overline{XY}$, for $\overline{XY}$ a proper complement of $XY$. Since $p$ is a piece, $|p| \leq |a\overline{X}|$ because otherwise a prefix of $\overline{XY}$ longer than $\overline{X}$ would be a piece. Hence $p$ is a prefix of $a\overline{X}$.
It follows that there exists a word equivalent to $w_0$ in which $aXY$ is followed by $Z$ and we can rewrite $XYZ$ to $\overline{XYZ}$. This implies that if $w'$ is the suffix of $w_0$ following $aXY$, then $Z$ is a possible prefix of $w'$. In particular, by the inductive hypothesis, \textbf{ReplacePrefix}($w', Z)$ is $Zu$ for some $u\in A ^ *$ and $a\overline{XYZ}u$ is a word equivalent to $w_0$ that has prefix $p$. Therefore, by induction, the algorithm will return $a\overline{XYZ}u$ in line 5 after making the recursive call in line 2.

It remains to show that the output of \textbf{ReplacePrefix}($w_0, p$) can be computed in $O(|w_0|)$ time. The recursive calls within \textbf{ReplacePrefix}($w_0, p$) always have argument which is a factor, even a suffix, of $w_0$. Hence if \textbf{WpPrefix}($w_0,w_0,p$)=Yes, then $p$ is a possible prefix of $w_0$, and the number of recursive calls in \cref{algorithm-prefix} is bounded above by the length of $w_0$. 

 Let $\delta$ be the length of the longest relation word of our presentation. In line 1, we begin by checking if $p$ is a prefix of $w_0$. Clearly, this can be done in $|p|$ steps and since $p$ is a piece, $|p| < \delta$. In line 1, we also search for the clean overlap prefix of $w_0$. As shown in Lemmas \ref{coprefix} and \ref{find-relation-prefix}, this can be done in constant time. Next, in line 2 we delete a prefix of length $|Z|$ from the output of \textbf{ReplacePrefix}($w', Z$). Since $|Z| < \delta$, the complexity of this step is also constant for a given presentation. The search for a complement $\overline{XYZ}$ of $XYZ$ such that $p$ is a prefix of $a\overline{X}$ can be performed in constant time since the number of relation words is constant for a given presentation and $|p|< \delta$. In line 3, we concatenate words to obtain a word equivalent to $w$. In every recursive call we concatenate three words hence the complexity of this step is also constant. 
As we have already seen, the number of recursive calls of the algorithm is bounded above by the length of $w_0$, hence \textbf{ReplacePrefix}($w_0, p$) can be computed in $O(|w_0|)$ time.
\end{proof}

For the following examples we will use the notation $w_i$ and $p_i$ for the parameters of the $i$th recursive call of \textbf{ReplacePrefix}($w,p$) and we let $w_0=w$ and $p_0=p$.

 	\begin{ex}
 	Let $\mathcal{P}$ be the presentation $$\left\langle a,b,c,d \, | \, acba=a^2bc,\, acba=db^2d\right\rangle $$ and we let $w=acbdb^2d$. 
 	The set of pieces of $\mathcal{P}$ is $P =\{\varepsilon, a, b, c, d\}$. Let $W_0 = acba, \, W_1 = a^2bc, \, W_2 = db^2d$. Clearly $X_{W_0} = a, \, Y_{W_0} = cb, \, Z_{W_0} = a, \, X_{W_1} = a, \, Y_{W_1} = ab, \, Z_{W_1} = c $ and $X_{W_2} = d, \, Y_{W_2} = b^2, \, Z_{W_2} = d$.  The algorithm \textbf{WpPrefix}$(w,w, d)$ returns Yes and we want to find \textbf{ReplacePrefix}($w, d$). 
 	
 	We begin with $w_0=acbdb^2d$, $p_0=d$ and $u_0=\varepsilon$. Clearly $w_0$ does not begin with $d$ but using the process described in \cref{find-relation-prefix}, we can find the clean overlap prefix of $w_0$ which is $acb = X_{W_0}Y_{W_0}$ and hence $w_0$ satisfies the conditions of line 1. In line 2, $w_1\gets db^2d$, $p_1 \gets a$ and in order to compute $u_1$ we need to compute \textbf{ReplacePrefix}$(db^2d,a)$. Since $w_1$ does not begin with $a$ we need to find the clean overlap prefix of $w_1$ which is $db^2=X_{W_2}Y_{W_2}$. Now $w_2 =d$, $p_2 =d$ and \textbf{ReplacePrefix}($d,d$) returns $d$. Now $w_1$ will be rewritten to a complement of $db^2d$ that begins with $a$, hence we choose one of $W_0$ and $W_1$. If we choose $W_0$, $w_1 \gets acba$ and $w_0 \gets db^2dcba$. If we choose $W_1$, $w_1 \gets a^2bc$ and $w_0 \gets db^2dabc$. In both cases the algorithm returns a word equivalent to $w$ that begins with $d$.
 	\end{ex}
 
 \section{A practical normal form algorithm} 
 \label{section-main}
 
 In this section we describe a practical algorithm for computing lexicographically normal forms in $C(4)$ monoids. This section has four subsections: the first contains a description of the algorithm; the second a proof that the algorithm returns a word equivalent to the input word; the third contains a proof that the algorithm returns the lexicographically least word equivalent to the input word; and in the final section we consider the complexity of the algorithm.
 
 \subsection{Statement of the algorithm}

In this section, we describe the main algorithm of this paper for computing the lexicographically least word equivalent to an input word.
 Roughly speaking,  the input word is read from left to right, clean overlap prefixes of the form $uXY$ for $u \in A^*$ are found and replaced with a lexicographically smaller word if possible. Subsequently, the next clean overlap prefix of this form after $uXY$ is found, and the process is repeated.  The algorithm is formally defined in \cref{algorithm-normalform}.

\begin{algorithm}[H]
\caption{- \textbf{NormalForm}($w_0$)}
\label{algorithm-normalform}
\textbf{Input:} A word $w_0\in A ^ *$.\\
\textbf{Output:} The lexicographically least word $v\in A ^ *$ such that $v \equiv w_0$.

\raggedright
{\fontsize{10}{10}\selectfont
\begin{algorithmic}[1]
\STATE $W\gets \varepsilon$, $v\gets \varepsilon$, $w\gets w_0$
\WHILE{$w\neq \varepsilon$}
    \IF{$W=X_rY_rZ_r$, $w=Z_rw'$, $w'$ is $\overline{Z}_r$-active for some proper complement $\overline{Z_r}$ of $Z_r$, $w'$ is not $Z_r$-active and $a$ is a suffix of $\overline{Z_r}$ with $aw' = X_sY_sw''$ and \textbf{WpPrefix}$(w'',w'',Z_s)$=Yes}
    \IF{there exists a proper complement of $X_sY_sZ_s$ with prefix $a$ that is lexicographically less than $X_sY_sZ_s$}
     \STATE 
     $X_tY_tZ_t \gets$ \text{the lexicographically minimal proper complement of $X_sY_sZ_s$ that has prefix $a$}
     \ELSE
     \STATE $X_tY_tZ_t \gets \varepsilon$
    \ENDIF
       \IF{$X_tY_tZ_t \neq \varepsilon$, $X_t=ab$ and \textbf{WpPrefix}$(w_0,vZ_rbY_tZ_tt, \varepsilon)=\text{Yes}$ where $Z_st=$\text{\textbf{ReplacePrefix}}$(w'',Z_s)$}
	        \STATE $W\gets X_tY_tZ_t$
            \STATE $v\gets vZ_rbY_t$
	        \STATE $w\gets Z_tt$
	   \ELSE 
	       \STATE $W\gets X_sY_sZ_s$
	       \STATE $v\gets vZ_rX_s''Y_s$ where $X_s = aX_s''$
	       \STATE $w\gets \text{\textbf{ReplacePrefix}}(w'',Z_s)$
	   \ENDIF
	 \ELSIF{$w$ has a clean overlap prefix of the form $aXY$ and $w=aXYw'$}
  	    \IF {\textbf{WpPrefix}$(w',w',Z)=\text{No}$}
	        \STATE $W\gets \varepsilon$
		    \STATE $v\gets vaXY$
	        \STATE $w\gets w'$
	    \ELSE
	    	\STATE $W\gets $ the lexicographically minimal complement $X'Y'Z'$ of $XYZ$
	        \STATE $v\gets vaX'Y'$
	    	\STATE $w\gets Z'w''$, \text{where \textbf{ReplacePrefix}$(w', Z) = Zw''$} 
	   \ENDIF
    \ELSE 
      \STATE $v\gets vw$
	  \STATE $w \gets \varepsilon$  
    \ENDIF    
\ENDWHILE
\RETURN $v$
\end{algorithmic}}
\end{algorithm}

 \subsection{Equivalence}
 
In this section we show that \textbf{NormalForm} terminates and the word returned is equivalent to the input word $w_0$. We begin by observing that \textbf{NormalForm} rewrites $v$ and $w$ in lines 11-12, 15-16, 21-22, 25-26, and 29-30. For the remainder of this section, $v_i$ and $w_i$ will be $v$ and $w$ after the $i$-th time the algorithm has rewritten $v$ and $w$. 

The following result will be used to prove that \cref{algorithm-normalform} terminates and that the word returned by the algorithm is equivalent to its input. We have already proved that if a piece $p$ is a possible prefix of a word $v$, then algorithm \textbf{ReplacePrefix}($v, p$) returns a word equivalent to $v$ with prefix $p$.
If $w\in A^*$, $XY$ is a clean overlap prefix of $w$, $w'$ is the suffix of $w$ following $XY$, and $Z$ is a possible prefix of $w'$, then $w\equiv XYZu$ where $Zu=\text{\textbf{ReplacePrefix}}(w',Z)$.
This is straightforward since $Z$ is a possible prefix of $w'$ and hence $w' \equiv Zu$ for $Zu=\text{\textbf{ReplacePrefix}}(w',Z)$.  

\begin{lemma}\label{vw}
Assume that $w_0 \in A^*$ is the input to \textbf{NormalForm}. Then at each step of \textbf{NormalForm}($w_0$), $v_iw_i \equiv w_0$.
\end{lemma}

\begin{proof}
 We proceed by induction on $i$. For $i=0$, $v_0 = \varepsilon$ and hence $v_0w_0 = w_0$. Let $k \in \mathbb{N}$ and assume that $v_kw_k\equiv w_0$. We will prove that $v_{k+1}w_{k+1}\equiv w_0$. 

In the cases of lines 21-22 and 29-30, it is clear that some prefix of $w_k$ is transferred to the end of $v_{k + 1}$. In particular, $v_kw_k = v_{k+1}w_{k+1}\equiv w_0 $. In lines 15-16, a prefix of $w_k$ is transferred to the end of $v_{k + 1}$ again and \cref{algorithm-prefix} is applied. Hence $v_kw_k \equiv v_{k+1}w_{k+1}\equiv w_0 $. In lines 25-26 we rewrite the relation word $XYZ$ to $\overline{XYZ}$. Since $XYZ$ begins after the beginning of $w_k$, there exists some $s$ equivalent to $w_k$ which is obtained by the application of this rewrite. Hence $v_kw_k \equiv v_ks$ with $a\overline{XY}$ being a prefix of $s$, and $v_ks=v_{k+1}w_{k+1}$. It follows that $v_{k+1}w_{k+1}\equiv w_0$.
Finally, in the case of lines 11-12 the result follows immediately from the use of \textbf{WpPrefix} in line 9.
\end{proof}

In~\cite[Theorem 5.2.14]{higgins} it is shown that if $w_0, w \in A^*$ are such that $w\equiv w_0$, then
$$|w|<\delta |w_0|$$ where $\delta$ is the maximum length of a relation word in $R$.
Since in every step of \textbf{NormalForm}$(w_0)$, $v_iw_i \equiv w_0$ by \cref{vw}, we conclude that  $|v_iw_i| \leq \delta |w_0|$ for all $i$. \cref{algorithm-normalform} terminates when $w_i = \varepsilon$. Since the length of $v_{i + 1}$ is strictly greater than the length of $v_{i}$, \cref{algorithm-normalform} terminates for any $w_0 \in A^*$ and the while loop of line 2 will be repeated at most $\delta|w_0|$ times.

Combining \cref{vw} with the fact that \textbf{NormalForm} terminates, we obtain the following  corollary.

\begin{cor} \label{equivalence}
If $w_0\in A^*$ is arbitrary, then the word $v$ returned by \textbf{NormalForm}($w_0$) is equivalent to $w_0$.
\end{cor}

 \subsection{Minimality}
We require the following definition and a number of related results to establish that the word returned by \textbf{NormalForm}$(w_0)$ is the lexicographic minimum word equivalent to $w_0$.

\begin{de}
Let $w \in A^*$. A middle word $Y$ is called a \textit{special} middle word of $w$ if $w = pYq$ for some $p, q \in A^*$ and there exists a word $p'XYZq'$ that is equivalent to $w$, $p'X \equiv p$, and $Zq' \equiv q$.
\end{de}

In other words, $Y$ is a special middle word of $w$ if it is a subword of $w$ and there exists a word equivalent to $w$ containing $XYZ$ as a factor in the obvious place. Note that if a relation word $XYZ$ is a factor of $w$, then it follows directly from the definition that $Y$ is a special middle word of $w$.
Since middle words are not pieces, it follows that a middle word $Y_i$ will never occur as a factor of a middle word $Y_j$ unless $Y_i = Y_j$. So, if $Y_i$ and $Y_j$ are special middle words of a word $w$ and they begin at the same position in $w$, then $Y_i = Y_j$. In the following lemma, we prove that the special middle words of an arbitrary word $w$ do not overlap with each other.

\begin{lemma}\label{no-overlap}
Let $Y_i$ and $Y_j$ be special middle words of $w$ where $Y_i$ occurs strictly before $Y_j$. Then $w = pY_iqY_jr$ for some $p, q, r \in A^*$.
\end{lemma}

\begin{proof}
Assume that $Y_i$ and $Y_j$ are such that $Y_i$ and $Y_j$ overlap as factors in $w$. Let $Y_i = xy$ and $Y_j = yz$ be such that $w = pY_izr = pxY_jr = pxyzr$ for some $x, y, z\in A ^*$.

Since $Y_i$ is a special middle word of $w = pY_izr$, there exists $r' \in A^*$ such that $zr \equiv Z_ir'$. If $z$ is a prefix of $Z_i$, then $yz = Y_j$ is a factor of the relation word $X_iY_iZ_i$, a contradiction since $Y_j$ is not a piece. If $z$ is a prefix of $Z_ir'$ that is longer than $Z_i$, then the suffix $yZ_i$ of $X_iY_iZ_i$ which is longer than $Z_i$ is a factor of $Y_j$ and this contradicts the definition of $Z_i$. It follows that $z$ is not a prefix of $Z_ir$ and so $Z_ir' \neq zr$ and, in particular, $zr$ contains a relation word as a factor and hence $zr$ has a relation prefix and a clean overlap prefix. Hence $zr = aXYq'$ for $aXY$ a clean overlap prefix and some $q' \in A^*$. In addition, $|a| < |z|$ since $a$ is contained in every word equivalent to $aXYq'$ by \cref{lem-2}, $aXYq' \equiv Z_ir'$ and $z$ is not a prefix of $Z_i$. Since $|a|<|z|$, there exists a suffix $X'$ of $z$ that is a prefix of $X$ and $zr = aXYq' = zX''Yq'$ where $X''$ is such that $X = X'X''$.
Since $Y_j$ is a special middle word of $w$ and $w = pxY_jr = pxY_jX''Yq'$, $X''Yq'$ is equivalent to $Z_jt$ for some $t \in A^*$. Clearly, $Z_j$ is not a prefix of $X''Y$ because that would imply that a suffix of $X_jY_jZ_j$ longer than $Z_j$ is a factor of $XY$. In addition, $X''Y$ is not a prefix of $Z_j$ because $Y$ is not a piece. It follows that $Z_jt \neq X''Yq'$. Since $Z_jt \equiv X''Yq'$ but $Z_j$ is not a prefix of $X''Yq'$ it follows that $X''Yq'$ has a clean overlap prefix $bX_{*}Y_{*}$ with $|b| < |X''Y|$ because otherwise $X''Y$ would be a factor of all words equivalent to $X''Yq'$ and hence a factor of $Z_j$. If a prefix of $X''Y$ longer than $X''$ is a factor of $b$, then a prefix of $XY$ longer than $X$ is a factor of $Y_jZ_j$, a contradiction. 
It follows that either $X_{*}Y_{*}$ is a factor of $X''Y$ or $Y$ is a factor of $X_{*}Y_{*}$ and both cases lead to a contradiction.

We conclude that $Z_j$ cannot be a prefix of any word equivalent to $r$. In particular, it follows that $Y_j$ is not a special middle word of $w$ which contradicts the initial assumption and hence $Y_i$, $Y_j$ do not overlap.
\end{proof}

We can order the special middle words of a word $w \in A^*$ by their order of appearance as factors of $w$ from left to right. In particular, for every $ w \in A ^*$ we will refer to the \textit{sequence of special middle words} $(Y_0, Y_1, \ldots, Y_n)$ of $w$ where $i < j$ whenever $Y_i$ occurs to the left of $Y_j$ in $w$. 

The next lemma collects some basic facts about the decomposition of the relation words $u$ into $X_uY_uZ_u$ that follow more or less immediately from the definition of the $C(4)$ condition.

\begin{lemma}\label{lem-1}
    Let $W_i=X_iY_iZ_i$ and  $W_j=X_jY_jZ_j$ be arbitrary relation words in $R$ such that $W_i \not= W_j$. 
    \begin{enumerate}[\rm (i)]
        \item 
            If a suffix of $Y_iZ_i$ is a prefix of $X_jY_j$, then a suffix of  $Z_i$ is a prefix of $X_j$.
        \item 
            If $W_i$ overlaps $W_j$ in a word $w\in A ^ *$, then either: $Z_i$ overlaps with $X_j$ or $X_i$ overlaps with $Z_j$.
        \item 
        If $Y_u = sY_u'$ for some $s, Y_u'\in A ^*$ with $s\neq \varepsilon$, then $X_us$ is not a factor of any relation word other than $u$.
    \item
        Suppose that $w=X_iY_iZ_iX_j''Y_jZ_j$ where $X_j=X_j'X_j''$ and $X_j'$ is a suffix of $Z_i$. 
    If $W_k$ is a relation word in $R$ such that $w \equiv p W_k q$ for some $p, q\in A ^ *$, then $W_k$ equals a complement $\overline{W_i}$ of $W_i$, or a complement $\overline{W_j}$ of $W_j$.
    \end{enumerate}
\end{lemma}

To prove that the word returned by \textbf{NormalForm} is the lexicographically least word equal to the input word, we establish the following theorem.

\begin{thm} \label{theorem}
Suppose that $u, v\in A ^ *$ are such that $u \equiv v$,
that $u = u_0Y_{0} \cdots u_mY_{m}u_{m + 1}$, and that 
$v = v_0 \overline{Y_{0}} \cdots v_n\overline{Y_{n}}v_{n + 1}$ where $Y_{i}$ and $\overline{Y_{i}}$ are the special middle words in $u$ and $v$, respectively. Then $u = v$ if and only if $m = n$ and 
$Y_{0}= \overline{Y_{0}},\ldots, Y_{m} = \overline{Y_{m}}$.
\end{thm}

We establish the proof of \cref{theorem} in a sequence of lemmas. We start by showing that if $u \equiv v$, then there is a 1-1 correspondence between the special middle words of $u$ and $v$. Using the properties in \cref{lem-1} we obtain the following lemma to show that $Y$ is a special middle word of $u$, if and only if some complement $\overline{Y}$ of $Y$ is a special middle word of $v$. 

\begin{lemma}\label{yproperty}
Let $u, v \in A^*$. Assume that $Y$ is a special middle word of $u$ such that $u = pYq$ for some $p, q \in A^*$. Then $u \equiv v$ if and only if one of the following holds:
\begin{enumerate}[\rm (i)]
    \item $v = p'Yq'$ such that $p' \equiv p$ and $q' \equiv q$; or
    \item $u = pYq \equiv rXYZt \equiv r\overline{XYZ}t \equiv p'\overline{Y}q' = v$ and $p \equiv rX$, $q \equiv Zt$, $p' \equiv r\overline{X}$ and $q' \equiv \overline{Z}t$.
\end{enumerate}
\end{lemma}

\begin{proof}
Clearly, if (i) or (ii) hold then $u \equiv v$. It remains to show that if $u \equiv v$ then (i) or (ii) holds for $v$. 
Assume that $u = pYq$ and let $v \equiv u$. Then there exists a rewrite sequence $u = pYq \xrightarrow[]{*} v$. We will prove that no relation applied in this rewrite sequence can overlap $Y$ unless it is $XYZ$. 

It is clear that since $Y$ is not a piece, $Y$ is not a factor of any relation word except $XYZ$ and no relation word is a factor of $Y$. We start by showing that no relation word in the rewrite sequence $pYq \xrightarrow[]{*} v$ overlaps with a proper suffix of $Y$. Since $Y$ is a special middle word of $u$, $u = pYq \equiv rXYZt$ for some $r, t \in A^*$ such that $p \equiv rX$ and $q \equiv Zt$. Since $q \equiv Zt$, it follows by \cref{lem-2} that there are two cases to consider: either $q$ has a clean overlap prefix $aX_1Y_1$ with $|a| \geq |Z|$ and $Z$ is a prefix of all words equivalent to $q$ or $q$ has a clean overlap prefix $aX_1Y_1$ with $|a| < |Z|$ and $Z$ is a prefix of $a\overline{X_1Y_1}$ for some complement $\overline{X_1Y_1}$ of $X_1Y_1$. If the former holds, then no relation word in the rewrite sequence can overlap with a suffix of $Y$ because that would imply that either a suffix of $XYZ$ longer than $Z$ is a factor of a different relation word or that a relation word is a factor of $YZ$. Both of these lead to contradictions. Assume that $q = aX_1Y_1t'$ for some $t' \in A^*$. By \cref{lem-2} every word equivalent to $q$ has either $aX_1Y_1$ or $a\overline{X_1Y_1}$ as a prefix. If a relation word in the rewrite sequence overlaps a suffix of $Y$, then one of the following holds: 
\begin{itemize}
    \item it is a factor of $Ya$ which is a contradiction since $Ya$ is a factor of $YZ$;
    \item it is a factor of $Ya\overline{X_1Y_1}$ for some complement $\overline{X_1Y_1}$ of $X_1Y_1$ which is a contradiction because it implies that the relation word can be written as a product of 2 pieces; or
    \item the relation word contains $\overline{X_1Y_1}$ as a factor which is clearly a contradiction.
\end{itemize}
It follows that no relation word in the rewrite sequence overlaps with a suffix of $Y$ unless it is $XYZ$ in the obvious place.

Similarly, we can prove that no relation word in the rewrite sequence can overlap with a prefix of $Y$. Similar to the previous case, since $p \equiv rX$ there are two cases to consider: either $X$ is a suffix of all words equivalent to $p$ or it follows by \cref{lem-1}(i) and (ii) that $p \equiv r'X_2Y_2Z_2b$ for some $r', b \in A^*$ such that $|b|<|X|$ and there exists a word $r'\overline{X_2Y_yZ_2}b$ for $\overline{X_2Y_2Z_2}$ some complement of $X_2Y_2Z_2$, that has $X$ as a suffix. It follows by \cref{lem-1}(ii) and (iv) that any word overlapping with $\overline{X_2Y_2Z_2}$ that is not $\overline{X_2Y_2Z_2}$ in a rewrite sequence either overlaps with a prefix of $\overline{X_2}$ or with a suffix of $\overline{Z_2}$ and is $XYZ$. It follows that no relation word in the rewrite sequence overlaps with a prefix of $Y$ unless it is $XYZ$ in the obvious place.

In conclusion, no relation word in the rewrite sequence $u \xrightarrow[]{*} v$ overlaps with $Y$ unless it is a complement of $XYZ$ and hence (i) or (ii) holds for every $v$ such that $v \equiv u$.
\end{proof}

It follows by \cref{yproperty} that if $u \equiv v$ and $Y_i$ is a special middle word of $u$, then some complement $\overline{Y_i}$ of $Y_i$ is a special middle word of $v$. Note that it also follows by \cref{yproperty}(i) and (ii) that if $Y_i$, $Y_j$ are special middle words of $u$ and $Y_i$ occurs to the left of $Y_j$ in $u$, then the corresponding special middle word $\overline{Y_i}$ occurs to the left of $\overline{Y_j}$ in $v$. In particular, the following result holds as a corollary of \cref{yproperty}.

\begin{cor}\label{smw-correspondence}
If $u \equiv v$, then $(Y_0, Y_1, \ldots Y_n)$ is the sequence of special middle words of $u$ if and only if the sequence of special middle words of $v$ is $(\overline{Y_0}, \overline{Y_1}, \ldots \overline{Y_n})$ where $\overline{Y_i}$ is a complement of $Y_i$ for each $i$.
\end{cor}

The next lemma collects some basic facts about special middle words that follow more or less immediately from the definition of special middle words and the definition of clean overlap prefixes. We will use these properties in various results in the remainder of this section.

\begin{lemma}\label{lem-3}
Let $w \in A^*$.
    \begin{enumerate}[\rm (i)]
        \item 
           If $Y$ is a special middle word of $w$ and $w = pYq$ for some $p, q \in A^*$, then $Z$ is a possible prefix of $q$ and \textbf{WpPrefix}($q, q, Z$)=Yes.
        \item 
            If $w = XYw'$ for some $w' \in A ^*$ and \textbf{WpPrefix}($w', w', Z$)=Yes, then $XY$ is a clean overlap prefix of $w$.
        \item If $Y$ is the middle word in line 18 of \cref{algorithm-normalform} and the condition of line 19 is satisfied, then $Y$ is not a special middle word of $w = aXYw'$.
        \item 
        If $w = a_0X_0Y_0w'$ for some $w' \in A ^*$ and $Y_0$ is the left most special middle word in $w$, then any word equivalent to $w$ has prefix $a_0$.
    \end{enumerate}
\end{lemma}

We use some of the properties in \cref{lem-3} to prove the next lemma which refines the form of a word with respect to a C(4) monoid presentation. 

\begin{lemma}\label{form-of-word}
Let $u \in A^*$. Assume that $(Y_0, Y_1, \dots, Y_n)$, for some $n\geq 0$, is the sequence of special middle words of $u$. Then 
\[ u = a_0X_0Y_0a_1X_1''Y_1a_2X_2''Y_2a_3\dots a_nX_n''Y_na_{n+1}\]
where $a_i \in A^ *$ and either $X_i'' = X_i$ or $X_i''$ is a proper suffix of $X_i$ and $a_i = Z_{i-1}$ for all $i$.
\end{lemma}

\begin{proof}
Since $Y_0$ is the left most special middle word of $u$, $u = pY_0q$ and $p \equiv a_0X_0$ for some $a_0 \in A^*$. If a relation word $XYZ$ is a factor of $p$, then $Y$ is a factor of $p$ and hence it is a special middle word of $u$ occurring on the left of $Y_0$. This is a contradiction since $Y_0$ is the left most special middle word of $u$ and hence $p$ contains no relation words as factors. It follows that $p = a_0X_0$. 

Let $Y_{k-1}, Y_{k}$ be special middle words of $u$. Then $u = rY_{k-1}b_{k}Y_{k}t$ for some $r, b_{k}, t \in A ^ *$, by \cref{no-overlap}. We will prove that $b_{k} = a_{k}X_k''$ where either $X_k'' = X_k$ or $X_k''$ is a proper suffix of $X_k$ and $a_k = Z_{k-1}$. Since $Y_{k-1}$ is a special middle word of $u$, $b_{k}Y_{k}t \equiv Z_{k-1}s$ for some $s \in A^*$. It follows that either $Z_{k-1}$ is a prefix of all words equivalent to $b_kY_kt$ or $b_kY_kt$ has a clean overlap prefix $cXY$ with $|c| < |Z_{k-1}|$ and $Z_{k-1}$ is a prefix of $c\overline{XY}$ for some complement $\overline{XY}$ of $XY$.

In the latter case, since  $b_kY_kt = cXYq \equiv c\overline{XY}q'$ for some $q, q' \in A^*$, it follows that $cXYq \equiv cXYZq''$ for some $q'' \in A^*$ and hence $Y$ is a special middle word of $cXYq$. Since $cXYq$ is a suffix of $u$ it follows by the definition of special middle words that $Y$ is a special middle word of $u$.

We will prove that $Y = Y_k$ and hence $b_k = cX_k = a_kX_k''$ where $X_k'' = X_k$, as required.
Clearly, if $|cXY| = |b_kY_k|$ and $Y \neq Y_k$ then either $Y$ is a suffix of $Y_k$ or $Y_k$ is a suffix of $Y$. This is a contradiction since middle words are not pieces and hence $Y = Y_k$. 
Assume that $|cXY| < |b_kY_k|$. Then $Y$ is a special middle word that occurs after $Y_{k-1}$ and before $Y_k$ as a factor of $u$ and this is a contradiction.
Assume that $|cXY| > |b_kY_k|$. There are two cases to consider: either $Y_k$ is a factor of $XY$, which is clearly a contradiction, or $Y_k$ is a factor of $cXY$ that begins before the end of $c$. In this case, $Y_k$ must end before the start of $Y$ in $cXY$ because otherwise it would contain a factor of $XY$ longer than $X$. Since $Y_k$ is a special middle word, it follows that $b_kY_kt \equiv b_kY_kZ_kt'$ for some $t' \in A^*$. If $Z_k$ is a prefix of $t$ in this case, then either $Y$ is a factor of $X_kY_kZ_k$ or $XY$ contains a suffix of $X_kY_kZ_k$ longer than $Z_k$ and both of these lead to a contradiction.
It follows that $Z_k$ is not a prefix of $t$, and hence $t$ has a clean overlap prefix $dX_{*}Y_{*}$ with $|d| < |Z_k|$. 
If $d$ ends after the end of $Y$ in $cXYq$ then $Y$ is a factor of $X_kY_kZ_k$, a contradiction. It follows that $d$ ends before the end of $Y$. But then $XY$ overlaps with $X_{*}Y_{*}$ which is a contradiction because $cXY$ is a clean overlap prefix of $cXYq$. 
It follows that $Y = Y_k$.

In the former case, $Z_{k-1}$ is a prefix of all words equivalent to $b_kY_kt$. If $b_kY_kt$ has a clean overlap prefix $cXY$ with $|c| < |Z_{k-1}|$ and $b_kY_kt \equiv c\overline{XY}q'$ for some $q' \in A^*$ and a proper complement $\overline{XY}$ of $XY$, then as in the previous paragraph $Y = Y_k$ and hence $b_k = a_kX_k$, as required. Assume that $b_kY_kt$ does not have a clean overlap prefix $cXY$ with $|c| < |Z_{k-1}|$ such that $b_kY_kt \equiv c\overline{XY}q'$ for some $q' \in A^*$ and a proper complement $\overline{XY}$ of $XY$. Then, since $Y_k$ is a special middle word of $u$, if $X_k$ is a suffix of $b_k$ then it follows by \cref{lem-3}(ii) that $X_kY_k$ is a clean overlap prefix of $X_kY_kt$. If $b_kY_kt = b'X_kY_kt$ for some $b' \in A^*$ and there is not a clean overlap prefix contained in $b'$, then $b'X_kY_k$ is a clean overlap prefix of $b_kY_kt$ and $b_k = a_kX_k$, as required. If this is not the case, then either $X_k$ is a suffix of $b_k$ but $b_kY_kt$ has a clean overlap prefix $cXY$ with $|cXY| \leq |b_k|$ such that no equivalent of $b_kY_kt$ has prefix $c\overline{XY}$ for $\overline{XY}$ a proper complement of $XY$
or $X_k$ is not a suffix of $b_k$. In the first case, $b_k = a_kX_k$ and hence $X_k'' = X_k$ as required. In the second case, if $X_k$ is not a suffix of $b_k$, then since no word equivalent to $u$ contains a relation word between $Y_{k-1}$ and $Y_k$, it follows that $X_k = X_k'X_k''$ and the proper prefix $X_k'$ of $X_k$ is a suffix of some complement $\overline{Z_{k-1}}$ of $Z_{k-1}$. It follows that $b_k = Z_{k-1}X_k''$, as required.
\end{proof}

At this point, we have proved results that explain the connection between the sequence of special middle words of a word $w$ and the form of $w$ or a word equivalent to $w$ with respect to this sequence.
We would like to be able to compare words based on their sequences of special middle words.  
We utilize the algorithm \textbf{WpPrefix} from \cite{kambites2011} to do this. The following results highlight the connection between special middle words in equivalent words $u$ and $v$ and recursive calls from within \textbf{WpPrefix}$(u, v, \varepsilon)$.

\begin{lemma}\label{suffix}
If \textbf{WpPrefix}($u_j, v_j, p_j$) is a recursive call from within \textbf{WpPrefix}($u, v, \varepsilon$) for $u, v \in A^*$ and $u \equiv v$, then $u_j$ is a suffix of a word equivalent to $u$ and $v_j$ is a suffix of a word equivalent to $v$. In addition, if $u_j = XYu'$ and $XY$ is a clean overlap prefix of $u_j$, and $v_j = \overline{XY}v'$ and $\overline{XY}$ is a clean overlap prefix of $v_j$,
then $Yu'$ is a suffix of $u$ and $\overline{Y}v'$ is a suffix of $v$.
\end{lemma}

\begin{proof}
All line numbers in this proof refer to \textbf{WpPrefix} in \cite{kambites2011}. Clearly, if \textbf{WpPrefix}($u_j, v_j, p_j$) is the initial call to \textbf{WpPrefix}($u, v, \varepsilon$), then $u_j = u, v_j = v$ and hence $u_j$ and $v_j$ are suffixes of $u$ and $v$, respectively. In addition if $u_j = u = XYu'$ and $XY$ is a clean overlap prefix of $u_j$, and $v_j = v = \overline{XY}v'$ and $\overline{XY}$ is a clean overlap prefix of $v_j$ then clearly $Yu'$ is a suffix of $u$ and $\overline{Y}v'$ is a suffix of $v$. Assume that the result holds for the $j$-th recursive call \textbf{WpPrefix}($u_j, v_j, p_j$) from within \textbf{WpPrefix}($u, v, \varepsilon$). We will show that the result holds for the recursive call \textbf{WpPrefix}($u_{j+1}, v_{j+1}, p_{j+1}$) that occurs immediately after \textbf{WpPrefix}($u_j, v_j, p_j$). Since the result holds for $u_j$, there exists a word $w \equiv u$ with suffix $u_j$. If \textbf{WpPrefix}($u_{j+1}, v_{j+1}, p_{j+1}$) occurs in one of lines 15, 25, 28, 29, 31 or 42, then $u_{j+1}$ is a suffix of $u_j$ and the result holds for $u_{j+1}$. If \textbf{WpPrefix}($u_{j+1}, v_{j+1}, p_{j+1}$) occurs in line 24 or 33 then $u_j = XYZu''$ for some $u''\in A^*$ and $u_{j+1} = \hat{Z}u''$ for $\hat{Z}$ a complement of $Z$. Since $u_j = XYZu'' \equiv \hat{X}\hat{Y}\hat{Z}u''$ there exists a word equivalent to $w$ with suffix $\hat{X}\hat{Y}\hat{Z}u''$ and hence $u_{j+1} = \hat{Z}u''$ is a suffix of a word equivalent to $u$. The proof that $v_{j+1}$ is a suffix of a word equivalent to $v$ is analogous.

It remains to show that if $u_j$ has clean overlap prefix $XY$ and $u_j = XYu'$ and $v_j$ has clean overlap prefix $\overline{XY}$ and $v_j = \overline{XY}v'$,
then $Yu'$ is a suffix of $u$ and $\overline{Y}v'$ is a suffix of $v$. We will prove this for $u_j$, the proof for $v_j$ is analogous. As stated above, the result holds if \textbf{WpPrefix}($u_j, v_j, p_j$) is the initial call to \textbf{WpPrefix}($u, v, \varepsilon$) and $u$ has clean overlap prefix $XY$. If this is not the case, then the algorithm will make a number of calls in line 15 until a suffix of $u$ that has the form $XYu'$ with $XY$ a clean overlap prefix is found and hence the result holds for $Yu'$ in this case as well. We will assume that the result holds for \textbf{WpPrefix}($u_j, v_j, p_j$) such that $u_j$ has clean overlap prefix $X_jY_j$ and $u_j = X_jY_ju_j'$ and we will show that is holds for the recursive call \textbf{WpPrefix}($u_k, v_k, p_k$) that is the first recursive call after \textbf{WpPrefix}($u_j, v_j, p_j$) such that $u_k$ has clean overlap prefix $X_kY_k$ and $u_k = X_kY_ku'$. Since $u_j$ has clean overlap prefix $X_jY_j$ and $u_j \equiv v_j$, the recursive call to  \textbf{WpPrefix}($u_{j+1}, v_{j+1}, p_{j+1}$) that occurs immediately after \textbf{WpPrefix}($u_j, v_j, p_j$) occurs in one of lines 24, 25, 28, 29, 31, 33 or 42. If it occurs in line 25, 28, 29, 31 or 42, then $u_{j+1}$ is a suffix of $u_j$ and since $Y_ju_j'$ is a suffix of $u$ it follows that $u_{j+1}$ is a suffix of $u$. Since \textbf{WpPrefix}($u_k, v_k, p_k$) is the first recursive call after \textbf{WpPrefix}($u_j, v_j, p_j$) such that $u_k$ has clean overlap prefix $X_kY_k$ and $u_k = X_kY_ku'$, it follows that any recursive call after \textbf{WpPrefix}($u_{j+1}, v_{j+1}, p_{j+1}$) and before \textbf{WpPrefix}($u_k, v_k, p_k$) can only occur in line 15. It follows that $u_k$ is a suffix of $u_{j+1}$ and hence $Y_ku_k'$ is a suffix of $u$. If the recursive call to  \textbf{WpPrefix}($u_{j+1}, v_{j+1}, p_{j+1}$) occurs in line 24 or 33, then $u_j = X_jY_jZ_ju''$, $Y_jZ_ju''$ is a suffix of $u$ and $u_{j+1} = \hat{Z}_ju''$ for some complement $\hat{Z}_j$ of $Z_j$. Similar to the previous case, any recursive call after \textbf{WpPrefix}($u_{j+1}, v_{j+1}, p_{j+1}$) and before \textbf{WpPrefix}($u_k, v_k, p_k$) can only occur in line 15. Hence $u_{j+1} = aX_kY_ku_k'$ for $aX_kY_k$ a clean overlap prefix of $u_{j+1}$. Since $\hat{Z}_j$ is a prefix of $u_{j+1}$, $|aX_k| \geq |\hat{Z}_j|$, otherwise a prefix of $X_kY_kZ_k$ longer than $X_k$ would be a factor of $\hat{Z}_j$. It follows that $Y_ku_k'$ is a suffix of $u''$ and hence a suffix of $u$.
\end{proof}

The following result holds as a direct consequence of \cref{suffix} and the definition of special middle words and can be viewed as a tool to ``identify'' special middle words inside a word $w$ in some cases.

\begin{lemma}\label{wpprefix-smw}
Let $u, v \in A^*$ be such that $u \equiv v$. Assume that there exists a recursive call \textbf{WpPrefix}($XYu',$ $\overline{XY}v', p$) from within \textbf{WpPrefix}($u, v, \varepsilon$) for some $p, u', v' \in A^*$ such that $\overline{XY}$ is a proper complement of $XY$. Then $Y$ is a special middle word of $u$ and $\overline{Y}$ is the corresponding special middle word of $v$.
\end{lemma}

\begin{proof}
Since $u \equiv v$ it follows that \textbf{WpPrefix}($u, v, \varepsilon$) returns Yes and  hence since \textbf{WpPrefix}($XYu', \overline{XY}v', p$) is a recursive call from within \textbf{WpPrefix}($u, v, \varepsilon$), \textbf{WpPrefix}($XYu', \overline{XY}v', p$) returns Yes as well. It follows that $XYu' \equiv \overline{XY}v'$. Since $\overline{XY}$ is a proper complement of $XY$ it follows by Lemma 3 in \cite{kambites2011} that $u' \equiv Zu''$ and $v' \equiv \overline{Z}v''$ for some $u'', v'' \in A^*$. In addition, it follows by \cref{lem-3} (ii) that $XY$ is a clean overlap prefix of $XYu'$ and $\overline{XY}$ is a clean overlap prefix of $\overline{XY}v'$. By \cref{suffix}, $XYu'$ is a suffix of a word equivalent to $u$ and hence there exists a word $w$ such that $w \equiv u$ and $XYZu''$ is a suffix of $w$. It follows that $Y$ is a special middle word of $w$ and by \cref{smw-correspondence} some complements of $Y$ are the corresponding middle words in $u$ and $v$. By \cref{suffix}, $Yu'$ is a suffix of $u$ and $\overline{Y}v'$ is a suffix of $v$ and hence $Y$ and $\overline{Y}$ are the corresponding special middle words in $u$ and $v$, respectively.
\end{proof}

The next two results will be useful when comparing equivalent words based on the sequences of their special middle words.

\begin{lemma}\label{smw-wpprefix}
Let $u, v \in A^*$ be such that $u \equiv v$. Assume that $Y$ is a special middle word of $u$ and let $\overline{Y}$ be the corresponding special middle word in $v$. Then there exists a recursive call \textbf{WpPrefix}($XYu', \overline{XY}v', p$) from within \textbf{WpPrefix}($u, v, \varepsilon$) for some $p, u', v' \in A^*$.
\end{lemma}

\begin{proof}
All line numbers in this proof refer to \textbf{WpPrefix} in \cite{kambites2011}. 
Assume that $(Y_0, Y_1, \ldots, Y_n)$ is the sequence of special middle words of $u$. By \cref{form-of-word}, $u = a_0X_0Y_0a_1X_1''Y_1a_2X_2''Y_2a_3\dots a_nX_n''Y_na_{n+1}$ and $v = a_0\overline{X_0Y_0}b_1\overline{X_1''Y_1}b_2\overline{X_2''Y_2}b_3\dots b_n\overline{X_n''Y_n}b_{n+1}$ where $a_i, b_i \in A^*$ for all $i$ and $X_i''$ and $\overline{X_i''}$ are suffixes of $X_i$ and $\overline{X_i}$, respectively, for all $i$. We start by showing that the result holds for $Y_0$ and $\overline{Y_0}$. Since $u = a_0X_0Y_0u'$ for some $u' \in A^*$ and $Y_0$ is a special middle word of $u$, it follows by definition that $u' \equiv Z_0 u''$ for some $u'' \in A^*$. It follows by \cref{lem-3}(ii) applied to $X_0Y_0u'$ that $X_0Y_0$ is a clean overlap prefix of $X_0Y_0u'$. Since there are no relation words contained in $a_0$ and $a_0$ is also a prefix of $v$, \textbf{WpPrefix}($u, v, \varepsilon$) starts by making recursive calls in lines 15 and 28, until the prefix $a_0$ has been deleted and the recursive call to \textbf{WpPrefix}($X_0Y_0u', \overline{X_0Y_0}v', p_0$) occurs for some piece $p_0$.

We now assume that there exists a recursive call to \textbf{WpPrefix}($X_kY_ku_k', \overline{X_kY_k}v_k', p_k$) from within \textbf{WpPrefix}($u, v, \varepsilon$) for some $p_k, u_k', v_k' \in A^*$ for special middle words $Y_k$ and $ \overline{Y_k}$, respectively, and we will prove that there exists a recursive call to \textbf{WpPrefix}($X_{k+1}Y_{k+1}u_{k+1}', \overline{X_{k+1}Y_{k+1}}v_{k+1}', p_{k+1}$) from within \textbf{WpPrefix}($u, v, \varepsilon$) for some $p_{k+1}, u_{k+1}', v_{k+1}' \in A^*$ for special middle words $Y_{k+1}$ and $ \overline{Y_{k+1}}$, respectively. 

We will show this by examining the recursive calls that occur after \textbf{WpPrefix}($X_kY_ku_k', \overline{X_kY_k}v_k', p_k$). If $Z_k$ is not a prefix of $u_k'$ then the recursive call immediately after \textbf{WpPrefix}($X_kY_ku_k', \overline{X_kY_k}v_k', p_k$) occurs in one of lines 28, 29, 31 and 42. Since $Z_k$ is not a prefix of $u_k'$ and $Y_ku_k'$ is a suffix of $u$ by \cref{suffix}, it follows by Lemma ~\ref{form-of-word} that $a_{k+1}X_{k+1}'' = a_{k+1}X_{k+1}$. In addition, since $Y_k$ is a special middle word, $u_k'$ is equivalent to a word with prefix $Z_k$ and hence it has a clean overlap prefix $cXY$ with $|c| < |Z_k|$. Since $Z_k$ is not a prefix of $cX$ then it must be a prefix of $c\overline{X}$ for $\overline{X}$ a proper complement of $X$.

If the recursive call occurs in one of lines 28, 29 or 31 then the first argument is $u_k'$ which in this case has $a_{k+1}X_{k+1}Y_{k+1}$ as a prefix. If $a_{k+1}X_{k+1}Y_{k+1}$ is a clean overlap prefix of $u_k'$ the algorithm continues by making $|a_{k+1}|$ recursive calls in line 15. Then since $u \equiv v$ it makes the call \textbf{WpPrefix}($X_{k+1}Y_{k+1}u_{k+1}', V, p_{k+1}$) and $V = \overline{X_{k+1}Y_{k+1}}v_{k+1}'$ for some $p_{k+1}, u_{k+1}', v_{k+1}' \in A^*$. In addition, since this call was proceeded by $|a_{k+1}|$  calls in line 15 it follows that $v_k'$ has the prefix $a_{k+1}\overline{X_{k+1}Y_{k+1}}$ and since there is no clean overlap prefix in $a_{k+1}$, it follows by \cref{lem-3}(ii) that $\overline{Y_{k+1}}$ is the left most special middle word occurring to the right of $\overline{Y_k}$ in $v$ and hence it is the special middle word that corresponds to $Y_{k+1}$ in $u$. Assume that $a_{k+1}X_{k+1}Y_{k+1}$ is not a clean overlap prefix of $u_k'$. Since $Y_{k+1}$ is a special middle word of $u$, it follows by \cref{lem-3} (ii) that $X_{k+1}Y_{k+1}$ is a clean overlap prefix of a suffix of $u_k'$. It follows that the clean overlap prefix $cXY$ of $u_k'$ is such that $|cXY| \leq |a_{k+1}|$, otherwise either $cXY$ would overlap with $X_{k+1}Y_{k+1}$ or it would contain $X_{k+1}Y_{k+1}$ as a factor and both of these contradict the definition of a clean overlap prefix. It follows that $cXY$ is such that $|cXY| \leq |a_{k+1}|$ and since $Y$ is not a special middle word of $u$, $cXY$ is not followed by $Z$ and $u_k'$ is only equivalent to words that have $cXY$ as a prefix. It follows that in this case the algorithm makes recursive calls in lines 15 and 28 until $a_{k+1}$ gets deleted. Similar to the previous case, $a_{k+1}$ is a prefix of $v_k'$, the algorithm makes the recursive call \textbf{WpPrefix}($X_{k+1}Y_{k+1}u_{k+1}', \overline{X_{k+1}Y_{k+1}}v_{k+1}', p_{k+1}$) for some $p_{k+1}, u_{k+1}', v_{k+1}' \in A^*$ and $\overline{Y_{k+1}}$ is the special middle word of $v$ that corresponds to $Y_{k+1}$ in $u$. 

If the recursive call occurs in line 42 then by the assumptions of this case $a_{k+1}$ has prefix $z_1$ and any clean overlap prefix $XY$ of $u_k'$ begins after the end of $z_1$. It follows that the result holds in this case following the same argument as in the case of lines 28, 29 and 31 applied to the suffix of $u_k'$ that follows $z_1$.

It remains to show that the result holds when $Z_k$ is a prefix of $u_k'$. In this case the recursive call immediately after \textbf{WpPrefix}($X_kY_kZ_ku_k'', \overline{X_kY_k}v_k', p_k$) occurs in one of lines 24, 25 or 33. In case the call occurs in line 33 the result holds by symmetry with line 31. If the recursive call occurs in line 25, then $u_k''$ is not $\hat{Z}_{k+1}$-active for some complement $\hat{Z}_{k+1}$ of $Z_{k+1}$ and hence $a_{k+1}$ has the same form as in the case of lines 28, 29 and 31. Hence the same argument can show that the result holds in this case.

If $X_{k+1}''$ is a proper suffix of $X_{k+1}$ then $u_k''$ is $\hat{Z}_{k+1}$-active for some complement $\hat{Z}_{k+1}$ of $Z_{k+1}$, otherwise no word equivalent to $u$ would contain $X_{k+1}Y_{k+1}$ as a factor in this position. It follows that in this case the algorithm makes a recursive call in line 24 and $\hat{Z}_{k+1}u_k''$ has the clean overlap prefix $bX_{k+1}Y_{k+1}$ for some $b \in A^*$ with $|b| < |\hat{Z}_{k+1}|$. It follows that the algorithm makes $|b|$ recursive calls in line 15 and the result holds for $Y_{k+1}$. If $u_k''$ is $\hat{Z}_{k+1}$-active but the clean overlap prefix of $\hat{Z}_{k+1}u_k''$ is $cXY$ with $Y \neq Y_{k+1}$, then $a_{k+1}$ has the same form as in the cases of lines 25, 28, 29 and 31 and the result holds in this case.
\end{proof}

The following result holds as a corollary of \cref{smw-wpprefix}. 

\begin{cor}\label{cor-calls}
Let $u, v \in A^*$ be such that $u \equiv v$.  Assume that $Y_k$ and $Y_{k+1}$ are consecutive special middle words of $u$. Assume that $u_k, u_{k+1}, v_k, v_{k+1}$ are such that \textbf{WpPrefix}($u_k, v_k, p_k$), \textbf{WpPrefix}($u_{k+1}, v_{k+1}, p_{k+1}$) are the recursive calls corresponding to $Y_k, Y_{k+1}$ from within \textbf{WpPrefix}($u, v, \varepsilon$) from \cref{smw-wpprefix}. If \textbf{WpPrefix}($u_j, v_j, p_j$) is a recursive call from within \textbf{WpPrefix}($u, v, \varepsilon$) that occurs after \textbf{WpPrefix}($u_k, v_k, p_k$) and before \textbf{WpPrefix}($u_{k+1}, v_{k+1}, p_{k+1}$) and it is not the recursive call that occurs immediately after \textbf{WpPrefix} ($u_k, v_k, p_k$), then \textbf{WpPrefix}($u_j, v_j, p_j$) occurs either in line 15 or line 28 of \textbf{WpPrefix}.
\end{cor}

\begin{proof}
All line numbers in this proof refer to \textbf{WpPrefix} in \cite{kambites2011}. Assume that \textbf{WpPrefix}($u_j, v_j, p_j$) is the recursive call from within \textbf{WpPrefix}($u, v, \varepsilon$) as described in the statement of this lemma. If \textbf{WpPrefix}($u_j, v_j, p_j$) occurs in one of lines 31, 33 or 42 then $u_j = XYu'$ where $XY$ is a clean overlap prefix of $u_j$ and $v_j = \overline{XY}v'$ where $\overline{XY}$ is a clean overlap prefix of $v_j$ and $u', v' \in A^*$ and  $\overline{XY}$ is a proper complement of $XY$.
It follows by \cref{wpprefix-smw} that there exists a special middle word in a word equivalent to $u$ that occurs between $Y_k$ and $Y_{k+1}$, a contradiction. In addition, if \textbf{WpPrefix}($u_j, v_j, p_j$) occurs in lines 24, 25 or 29 then either $u_j = XYZu''$ for some relation word $XYZ$ and some $u'' \in A^*$ or $u_j = XYu'$ and $Z$ is a possible prefix of $u'$. By \cref{suffix} there exists a word equivalent to $u$ containing a relation word $XYZ$ as a factor and hence $Y$ is a special middle word. It follows that there exists a special middle word in a word equivalent to $u$ that occurs between $Y_k$ and $Y_{k+1}$, a contradiction. It follows that \textbf{WpPrefix}($u_j, v_j, p_j$) can only occur in line 15 or 28.
\end{proof}

We are now ready to prove the following lemma. \cref{theorem} will hold as a corollary of \cref{same-prefixes}. In addition, \cref{same-prefixes} will be used as a tool to compare prefixes of equivalent words when proving the correctness of the algorithm.

\begin{lemma}\label{same-prefixes}
Suppose that $u, v\in A ^ *$ are such that $u \equiv v$,
that $u = u_0Y_{0} \cdots u_mY_{m}u_{m + 1}$, and that 
$v = v_0 \overline{Y_{0}} \cdots v_n\overline{Y_{m}}v_{m - 1}$ where $Y_{i}$ are the special middle words in $u$ and $\overline{Y_{i}}$ are the special middle words in $v$. 
If $Y_{0}= \overline{Y_{0}},\ldots, Y_{k} = \overline{Y_{k}}$ for some $k$, then $u_0Y_{0} \cdots u_kY_{k}= v_0 \overline{Y_{0}} \cdots v_k\overline{Y_{k}}$.
\end{lemma}

\begin{proof}
All line numbers in this proof refer to \textbf{WpPrefix} in \cite{kambites2011}.
It follows by \cref{lem-3}(iv) and \cref{form-of-word},  that $u_0Y_{0} = a_0X_{0}Y_{0}$ and $a_0$ is a prefix of every word equivalent to $u$. Since $Y_{0} = \overline{Y_{0}}$, $u_0Y_{0} = v_0\overline{Y_{0}}$ by \cite[Lemma 2]{kambites2011}. 

Assume that the result holds for $Y_{0}= \overline{Y_{0}},\ldots, Y_{ k-1} = \overline{Y_{k-1}}$ for some $k \geq 1$.
Since $u \equiv v$ and $Y_{k}$ is a special middle word of $u$, 
it follows by \cref{smw-wpprefix} that 
there exist recursive calls to \textbf{WpPrefix}($u_{k-1}, v_{k-1}, p_{k-1}$) and \textbf{WpPrefix}($u_{k}, v_{k}, p_{k}$) from within \textbf{WpPrefix}($u, v, \varepsilon$) such that 
$u_{k-1} = X_{k-1}Y_{k-1}u_{k-1}'$, that $v_{k-1} = \overline{X_{k-1}Y_{k-1}}v_{k-1}'$, $u_{k} = X_{k}Y_{ k}u_{k}'$ and $v_{k} = \overline{X_{k}Y_{k}}v_{k}'$. Since $Y_{k-1} = \overline{Y_{k-1}}$ and $Y_{k} = \overline{Y_{k}}$ the next recursive call to \textbf{WpPrefix} within \textbf{WpPrefix}($u_{k-1}, v_{k-1}, p_{k-1}$) must occur in one of lines 24, 25, 28, or 29 (in the other cases the prefix of $u_{k - 1}$ is a proper complement of the prefix of $v_{k - 1}$). The only subsequent type of recursive call that can occur between \textbf{WpPrefix}($u_{k-1}, v_{k-1}, p_{k-1}$) and  \textbf{WpPrefix}($u_{k}, v_{k}, p_{k}$) is in lines 15 and 28 by Corollary~\ref{cor-calls}. 
Hence $u_k = v_k$ since in the recursive calls of lines 15 and 28 equal prefixes of the first two arguments are deleted. It follows that $u_0Y_{0} \cdots u_kY_{k}= v_0 \overline{Y_{0}} \cdots v_k\overline{Y_{k}}$.
\end{proof}

Having established \cref{theorem}, in the next 2 lemmas we consider how the special middle words $u, v\in A ^ *$ such that $u \equiv v$ interact with the lexicographic order.

\begin{lemma}\label{lem-refactor1}
Suppose that $u, v \in A ^ *$, that $u\equiv v$, and that  
$u = pY_{k - 1}Z_{k - 1}X_k''Y_kq$ for special middle words $Y_{k-1}$ and $Y_k$ in $u$, some $p, q \in A ^*$ and some proper suffix $X_k''$ of $X_k$. If  
$Y_{k - 1}$ and a proper complement $\overline{Y_k}$ of $Y_k$ are the corresponding middle words in $v$, then there is a suffix $X_k'$ of a complement $\overline{Z_{k-1}}$ of $Z_{k-1}$ such that $X_k = X_k'X_k''$ and $X_k'$ is also a prefix of $\overline{X_k}$.
\end{lemma}
\begin{proof}
Since $Y_{k-1}, Y_k$ are special middle words of $u$,
there is a rewrite sequence $u\stackrel{*}{\rightarrow}p'X_{k-1}Y_{k-1}Z_{k-1}X_{k}''Y_{k}q \rightarrow p'\overline{X_{k-1}Y_{k-1}Z_{k-1}}X_{k}''Y_{k}q$ where $\overline{Z_{k-1}} = zX_{k}'$ for some $z\in A ^ *$.
Hence 
\[
p'\overline{X_{k-1}Y_{k-1}Z_{k-1}}X_{k}''Y_{k}Z_{k}q_0= p'\overline{X_{k-1}Y_{k-1}}zX_{k}Y_{k}Z_{k}q_0\rightarrow  p'\overline{X_{k-1}Y_{k-1}}z\overline{X_{k}Y_{k}Z_{k}}q_0
\]
for some $q_0 \in A^*$.
Since $Z_{ k-1}$ is a factor of $v$, but not of $p'\overline{X_{k-1}Y_{k-1}}z\overline{X_{k}Y_{k}Z_{k}}q_0$ by the assumption of this case, 
\[
p'\overline{X_{ k-1}Y_{ k-1}}z\overline{X_{k}Y_{k}Z_{k}}q_0 \stackrel{*}{\rightarrow} 
p' \overline{X_{ k-1}Y_{ k-1}Z_{ k-1}}\, \overline{X_{k}''Y_{k}}q_1\rightarrow
p' X_sY_sZ_{ k-1} \overline{X_{k}''Y_{k}}q_1 \stackrel{*}{\rightarrow} v
\]
where $X_sY_sZ_{k-1}$ is a complement of $X_{k-1}Y_{k-1}Z_{k-1}$ with suffix $Z_{k-1}$ and $v = p_0Z_{ k-1}\overline{X_k''Y_k}q'$.
In particular, since  
\[p'\overline{X_{ k-1}Y_{ k-1}}z\overline{X_{k}Y_{k}Z_{k}}q_0 \stackrel{*}{\rightarrow} 
p' \overline{X_{ k-1}Y_{ k-1}Z_{ k-1}}\, \overline{X_{k}''Y_{k}}q_1,\]
$\overline{X_{k}} = X_k'\overline{X_k}''$ and so $X_k' = \overline{X_k}'$ is the required common prefix.
\end{proof}

In the next lemma, we make use of the following observation: if $u, v\in A ^ *$ are such that $u < v$ and $u$ is not a prefix of $v$, then $uw < vw'$ for all $w, w' \in A ^ *$.

\begin{lemma}\label{lem-xyz}
Suppose that $u, v\in A ^ *$ are such that $u\equiv v$. If $u < v$ and $Y_{ k}$ is the left most special middle word in $u$ such that the corresponding middle word $\overline{Y_{k}}$ in $v$ is a proper complement of $Y_{k}$, then $X_{k}Y_{k}Z_{k} < \overline{X_{k}Y_{k}Z_{k}}$.
\end{lemma}
\begin{proof}
Let $u := w_0, w_1, \ldots, w_n := v$ be any rewrite sequence. By Lemmas~\ref{form-of-word} and~\ref{same-prefixes}, there exist $a_k, b_k, p, q, q' \in A^*$ such that $u = pX_{k-1}''Y_{k-1}a_kX_{k}''Y_{k}q$ and $v = pX_{k-1}''Y_{k-1}b_k\overline{X_{k}}''\overline{Y_{k}}q'$ 
where $X_{k}''$ and $\overline{X_{k}}''$ are (not necessarily proper) suffixes of $X_{k}$ and $\overline{X_{k}}$, respectively. Since $Y_k$ is a special middle word of $u$, it follows by \cref{yproperty} (ii) that there exists $j\in \{0, \ldots, n\}$ such that $X_{k}Y_{k}Z_{k}$ is a factor of $w_j$ and $\overline{X_{k}Y_{k}Z_{k}}$ is a factor of $w_{j + 1}$.
It remains to show that $X_{k}Y_{k}Z_{k} < \overline{X_{k}Y_{k}Z_{k}}$. 

It follows by \cref{suffix} and \cref{smw-wpprefix} that there exists a recursive call to \textbf{WpPrefix}($X_{k-1}Y_{k-1}a_k$ $X_{k}''Y_{k}q$, $X_{k-1}Y_{k-1}b_k\overline{X_{k}}''\overline{Y_{k}}q', t$) for some piece $t$, from within \textbf{WpPrefix}($u, v, \varepsilon$). In addition, $X_{k-1}Y_{k-1}$ is a clean overlap prefix of the first two arguments of this call by \cref{lem-3} (i) and (ii). Since the first two arguments of this call begin with the clean overlap prefix $X_{k-1}Y_{k-1}$, the recursive call occurring immediately after \textbf{WpPrefix}($X_{k-1}Y_{k-1}a_kX_{k}''Y_{k}q$, $X_{k-1}Y_{k-1}b_k\overline{X_{k}}''\overline{Y_{k}}q', t$) must occur in one of lines 24, 25, 28, 29 and the only possible recursive calls that can occur before \textbf{WpPrefix}($X_kY_kq, \overline{X_kY_k}q', t'$) for some piece $t'$, are those in lines 15 or 28 by Corollary~\ref{cor-calls}. Since in the recursive calls of lines 15 and 28 equal prefixes of the first two arguments are deleted, it follows that $a_k = b_k$ and $v = pX_{k-1}Y_{k-1}a_k\overline{X_{k}''Y_{k}}q'$. Since $u < v$ and $X_kY_k\neq \overline{X_kY_k}$, it follows that $X_{k}''Y_{k} < \overline{X_{k}''Y_{k}}$.

We will show that this implies that $X_{k}Y_{k} < \overline{X_{k}Y_{k}}$ and hence $X_{k}Y_{k}Z_k < \overline{X_{k}Y_{k}Z_k }$.
There are two cases to consider: when $X_{k}'' = X_{k}$ and when $X_{k}''$ is a proper suffix of $X_{k}$. If $X_k'' = X_k$ then it follows that $\overline{X_{k}''} = \overline{X_{k}}$, since otherwise there would not be a recursive call to \textbf{WpPrefix}($X_kY_kq, \overline{X_kY_k}q', t'$) from within \textbf{WpPrefix}($u, v, \varepsilon$), which contradicts \cref{smw-wpprefix}. It follows that $X_kY_k < \overline{X_kY_k}$.
But $X_{k}Y_{k}$ is not a prefix of $\overline{X_{k}Y_{k}}$ because $Y_{k}$ is not a piece, and so $X_{k}Y_{k}Z_{k} < \overline{X_{k}Y_{k}Z_{k}}$.

Suppose that $X_{k}''$ is a proper suffix of $X_{k}$. 
It follows by \cref{form-of-word} that $a_k = b_k = Z_{k-1}$,
and so $u = pX_{k-1}''Y_{k-1}Z_{k-1}X_{k}''Y_{k}q$ and $v = pX_{k-1}''Y_{k-1}Z_{k-1}\overline{X_{k}''Y_{k}}q'$. In this case, as in the previous case, it follows that $X_{k}''Y_{k} < \overline{X_{k}''Y_{k}}$. By \cref{lem-refactor1}, if $X_{k}'$ and $\overline{X_{k}'}$ are such that $X_{k} = X_{k}'X_{k}''$ and $\overline{X_{k}} = \overline{X_{k}'}\overline{X_{k}''}$, then $X_{k}' = \overline{X_{k}'}$.  Hence 
 $X_kY_kZ_k = X_k'X_{k}''Y_{k}Z_k < X_k'\overline{X_{k}''Y_{k}Z_k} = \overline{X_{k}Y_{k}Z_k}$.
 \end{proof}

We are now ready to prove the correctness of the \textbf{NormalForm} algorithm. We have already shown that the output of \textbf{NormalForm}($w_0$) is a word equivalent to $w_0$ in \cref{equivalence}. By \cref{theorem} it suffices to show that if $v_n$ is the output of \textbf{NormalForm}($w_0$), then the sequences of special middle words of $v_n$ and $\min w_0$ are identical. We accomplish this in the next two lemmas.

\begin{lemma}\label{part1}
Let $w_0$ be the input to \textbf{NormalForm}($w_0$) and let $v_i$ with $i\geq 0$ be the value of $v$ after the $i$-th iteration of the while loop starting in line 2. Then for every special middle word $Y_k$ in $w_0$ there exists an $i$ such that $v_i$ contains a complement of $Y_k$.
\end{lemma}
\begin{proof}
All line numbers in this proof refer to \textbf{NormalForm} in \cref{algorithm-normalform}.
We will show that there exists an $i$ such that during the $i$-th iteration of the while loop in \textbf{NormalForm} one of the following holds:
\begin{enumerate} [\rm (i)]
    \item $\overline{Y_k} = Y_s$ in line 3 and all conditions of line 3 are satisfied; or
     \item $\overline{Y_k} = Y$  in line 18.
 \end{enumerate}
 If (i) holds, then $v_{i + 1}$ is assigned in line 11 or 15 and $v_{i + 1}$ contains a complement of $Y_k$.
 If (ii) holds, then $v_{i + 1}$ is assigned in line 21 or 25 and $v_{i + 1}$ contains a complement of $Y_k$.

 We proceed by induction on the number of special middle words in $w_0$.
By \cref{theorem}, if there are no special middle words in $w_0$, then the only
word equivalent to $w_0$ is itself. In particular, if there are no special
middle words in $w_0$, then neither $w_0$ nor any word equivalent to $w_0$
contains a relation word as a factor, by \cref{yproperty}. It follows that
either $w_0$ does not have a clean overlap prefix or if $w_0$ has a clean overlap prefix $cXY$ for some $c \in A^*$ such that $w_0 = cXYw'$ for $w' \in A^*$, then \textbf{WpPrefix}$(w', w', Z)$=No because otherwise there exists a word equivalent to $w_0$ that contains $XYZ$ as a factor in the obvious place. This is a contradiction since there are no special middle words in $w_0$. It follows that in every iteration of the while loop in line 2 the conditions of line 3 are not satisfied. If the conditions of line 18 are satisfied, then the condition of line 19 is satisfied as well and if the conditions of line 18 are not satisfied then we have the case of lines 29-30. In particular, in every iteration of the while loop in line 2 $v_iw_i = w_0$ and hence the algorithm returns $w_0$, as required.
 
Suppose that $w_0$ contains at least one special middle word. Let $w_0 = a_0X_0Y_0q$ where $Y_0$ is the left most special middle word of $w_0$ and $a_0, q \in A^*$. Since $Y_0$ is a special middle word, \textbf{WpPrefix}($q, q, Z_0$)=Yes by definition and hence $X_0Y_0$ is a clean overlap prefix of $X_0Y_0q$ by \cref{lem-3} (ii). 
Since $Y_0$ is the left most special middle word of $w_0$, the algorithm begins by finding the clean overlap prefix $aXY$ of $w_0$ such that $w_0 = aXYw'$. If $aXY = a_0X_0Y_0$, then (ii) holds for $Y_0$. If $aXY \neq a_0X_0Y_0$, then $|aXY| < |a_0X_0Y_0|$ because otherwise $X_0Y_0$ would be a factor of $aXY$, which contradicts the fact that $X_0Y_0$ is a clean overlap prefix of $X_0Y_0q$. In addition, $X_0Y_0$ does not overlap with a suffix of $aXY$ because $aXY$ is clean. In this case, $aXY$ satisfies the conditions of line 18 and 19 and $w_1 = w'$. Since $W$ is assigned to be equal to $\varepsilon$ in line 20 and since no clean overlap prefix of $w'$ overlaps with $X_0Y_0$, the same steps as in the first iteration of the while loop are repeated until $w_i = bX_0Y_0w_i'$ for some $b, w_i'$ such that (ii) is satisfied for $Y_0$. 

We assume that (i) or (ii) holds for the special middle words $Y_0, \ldots, Y_{k - 1}$ of $w_0$ for some $k\geq 1$. We will show that (i) or (ii) holds for $Y_k$ also. 

Let $v_jw_j$ be the word equivalent to $w_0$ after the $j$-th iteration of the while loop in line 2. 

Suppose that either (i) or (ii) was satisfied for $Y_{k-1}$ during the $j$-th iteration of the while loop.
In this case,
$v_j$ is defined in one of lines 11, 15, or 25, and in any of these cases:
$$v_j = p\overline{X_{k-1}''Y_{k-1}} \text{ and } w_j =  b_k\overline{X_k''Y_k}q$$ 
by \cref{form-of-word} for some $p, b_k, q\in A ^ *$,  $\overline{X_{k-1}''}$ and $\overline{X_{k}''}$ are suffixes of $\overline{X_{k-1}}$ and $\overline{X_{k}}$, respectively, and $\overline{Z_{k-1}}$ is a prefix of $b_k$ since $w_j$ was assigned in one of lines 12, 16 or 26.
Since $\overline{Y_k}$ is a special middle word of $v_jw_j$ it follows by definition that \textbf{WpPrefix}$(q, q, \overline{Z_k})=$Yes. In addition, by \cref{form-of-word}, either $\overline{X_k''} = \overline{X_k}$; or $\overline{X_k''}$ is a proper suffix of $\overline{X_k}$ and $b_k = \overline{Z_{k-1}}$. 
If $\overline{X_k''}$ is a proper suffix of $\overline{X_k}$ and $b_k = \overline{Z_{k-1}}$, then, since $\overline{Y_k}$ is a special middle word of $v_jw_j$, $\overline{X_k''Y_k}q$ is either $\overline{Z_{k-1}}$-active or $Z$-active for some proper complement $Z$ of $\overline{Z_{k-1}}$ since there exists a word equivalent to $v_jw_j$ containing $\overline{X_kY_kZ_k}$ as a factor in the obvious place. If the latter holds, then the conditions of line 3 are satisfied. In particular, we have that $W = \overline{X_{k-1}Y_{k-1}Z_{k-1}}$ and that $\overline{Z_{k-1}}$ is a prefix of $w_j$ since $v_j$ was assigned in one of lines 11, 15 or 25. In addition \textbf{WpPrefix}$(q, q, \overline{Z_k})$ = Yes since $\overline{Y_k}$ is a special middle word. Hence (i) holds for $\overline{Y_k}$. If the former holds, then the conditions of line 3 are not satisfied and since $\overline{X_k''Y_k}q$ is $\overline{Z_{k-1}}$-active it follows by definition that $z\overline{X_kY_k}$ is a clean overlap prefix of $w_j = b_k\overline{X_k''Y_k}q$ for some $z \in A^*$ with $|z| < |b_k|$. Hence (ii) holds for $\overline{Y_k}$.

In the case that $X_k'' = X_k$ then
$w_j = \overline{Z_{k-1}}c_k\overline{X_kY_k}w'$ for some $c_k \in A^*$. If the conditions of line 3 are satisfied for $w_j$ then $c_k\overline{X_kY_k}w'$
is $Z$-active for some proper complement $Z$ of $\overline{Z_{k-1}}$,
$c_k\overline{X_kY_k}w' = X_s''Y_sw''$ for some $w'' \in A^*$ and \textbf{WpPrefix}$(w'', w'', Z_s)$= Yes and hence there exist $p, q \in A^*$ such that $v_jw_j = pY_sq$ and a word $p'X_sY_sZ_sq'$ with $p'X_s \equiv p, Z_sq' \equiv q$. But this implies that $Y_s$ is a special middle word of $w_0$ that occurs between $Y_{k-1}$ and $Y_k$, a contradiction. It follows that in this case the conditions of line 3 cannot be satisfied. Since $w_j = \overline{Z_{k-1}}c_k\overline{X_kY_k}w'$ and $\overline{Y_k}$ is a special middle word of $v_jw_j$, \textbf{WpPrefix}($w', w', \overline{Z_k}$) = Yes and hence $\overline{X_kY_k}$ is a clean overlap prefix of $\overline{X_kY_k}w'$ by \cref{lem-3} (ii). It follows by the same argument applied to prove the base case that either $c_k\overline{X_kY_k}$ is a clean overlap prefix of $w_j$ or $w_j$ has a clean overlap prefix $cXY$ such that $|cXY| \leq c_k$ and hence (ii) holds for $Y_k$.
\end{proof}

\begin{lemma}\label{W-is-smw}
Let $v_i$, $w_i$ and $W_i$ be the values of $v$, $w$ and $W$ after the $i$-th iteration of the while loop of line 2 in \cref{algorithm-normalform}. If $W_i = XYZ$ then $Y$ is a special middle word of $v_iw_i$.
\end{lemma}

\begin{proof}
All line numbers in this proof refer to \textbf{NormalForm} in \cref{algorithm-normalform}. The value of $W_i$ gets assigned in one of lines 10, 14, 20 and 24. In line 24, $W \gets \varepsilon$ and hence it suffices to examine the cases of lines 10, 14 and 24. In each of these cases, $W_i \gets XYZ$ for some relation word $XYZ$ and $v_i \gets v'Y$ for some $v' \in A^*$. It follows that $Y$ is a factor of $v_iw_i$.

We prove that $Y$ is a special middle word of $v_iw_i$ by induction. Assume that the $k$-th iteration of the while loop of line 2 is the first iteration of \textbf{NormalForm}($w_0$) such that a value not equal to $\varepsilon$ gets assigned to $W_k$. It follows that $W_{k-1}$ did not satisfy the condition of line 3 and the values of $W_k, v_k, w_k$ get assigned in lines 24, 25 and 26, respectively. Hence $w_{k-1}$ has a clean overlap prefix $aXY$ for some $a \in A^*$, $w_{k-1} = aXYw'$ and $Z$ is a possible prefix of $w'$. It follows that $Y$ is a factor of $v_{k-1}w_{k-1} \equiv v_{k-1}aXYZq \equiv v_{k-1}aX'Y'Z'q = v_kw_k$ for $X'Y'Z'$ the lexicographically minimal equivalent of $XYZ$ and for some $q \in A^*$. Hence, by definition, $Y'$ is a special middle word of $v_kw_k$.

We now assume that the result holds for the first $j$ iterations of the while loop of line 2 and assume that $m \in \mathbb{N}$ is such that the $(j+m)$-th iteration of the while loop of line 2 is the first iteration after the $j$-th iteration of \textbf{NormalForm}($w_0$) such that a value not equal to $\varepsilon$ gets assigned to $W_{j+m}$. Then the value of $W_{j+m}$ gets assigned in one of lines 10, 14 or 24. In the cases of lines 10 and 14, $W_{j+m-1} \neq \varepsilon$ and hence $W_{j+m-1} = W_{j-1}$ and $W_{j+m} = W_j$. It follows that $W_{j-1} = X_rY_rZ_r$ and $v_{j-1}w_{j-1} = pY_rq$ for some $p, q \in A^*$ and $v_{j-1}w_{j-1} \equiv p'X_rY_rZ_rq'$ for some $p', q' \in A^*$. Since the value of $W_j$ gets assigned in one of lines 10 and 14, $w_{j-1}$ satisfies the conditions of line 3. In particular, $w_{j-1} = Z_rw'$, $w'$ is $\overline{Z_r}$-active, for $\overline{Z_r}$ a proper complement of $Z_r$ and $\overline{Z_r}w' = aX_sY_sw''$ and $Z_s$ is a possible prefix of $w''$. It follows that $v_{j-1}w_{j-1} \equiv p'\overline{X_rY_rZ_r}q' = p'\overline{X_rY_r}aX_sY_sw''$. Since $Z_s$ is a possible prefix of $w''$, it follows that $X_sY_sZ_s$ is a factor of a word equivalent to $v_{j-1}w_{j-1}$ and hence $Y_s$ is a special middle word of $v_{j-1}w_{j-1}$. If the values of $W_j, v_j, w_j$ get assigned in lines 14, 15 and 16, respectively, then $v_{j-1}w_{j-1} = v_jw_j$, $W_j = X_sY_sZ_s$ and hence $Y_s$ is a special middle word of $v_jw_j$. If the values of $W_j, v_j, w_j$ get assigned in lines 10, 11 and 12, respectively, then
$W \gets X_tY_tZ_t$ is a complement of $X_sY_sZ_s$ and since $Y_s$ is a special middle word of $v_{j-1}w_{j-1}$, it follows by \cref{yproperty} that $Y_t$ is a special middle word of $v_jw_j$.

In the case of line 24 the result follows by an argument that is identical to the argument in the proof of the base case of this proof.
\end{proof}

\begin{lemma}\label{part2}
Let $w_0$ be the input to \textbf{NormalForm}($w_0$) and let $v_i$ with $i\geq 0$ be the value of $v$ after the $i$-th iteration of the while loop in line 2. Then $v_i$ is a prefix of the lexicographically minimal word $\min w_0$ equivalent to $w_0$ for every $i$.
\end{lemma}
\begin{proof}
All line numbers in this proof refer to \textbf{NormalForm} in \cref{algorithm-normalform}.
Certainly, $v_0 = \varepsilon$ is a prefix of the lexicographically least word $\min w_0$ equivalent to $w_0$.
Assume for $j\geq 1$ that $v_{j-1}$ is a prefix of $\min w_0$.
We will show that $v_j$ is also a prefix of $\min w_0$.
Since $v_{j - 1}$ is a prefix of $\min w_0$,
the special middle words in $v_{j - 1}$ are the initial $k$ special middle words in $\min w_0$ for some $k$. 
The value of $v_j$ is assigned in one of lines 11, 15, 21, 25, and 29 and in every case $v_{j-1}$ is a prefix of $v_j$.
We consider each of these cases separately.
\medskip

\noindent\textbf{line 11:}
If $v_{j}$ is defined in line 11, then the conditions of lines 3, 4 and 9 are satisfied. Since the conditions of line 3 are satisfied, $W_{j-1} = X_rY_rZ_r$ and by \cref{W-is-smw}, $Y_r$ is a special middle word of $v_{j-1}w_{j-1}$. Since the value of $W_{j-1}$ is assigned such that $Y_r$ is a suffix of $v_{j-1}$, it follows that $Y_r = Y_{k-1}$. In addition, since the value of $W_j$ gets assigned in line 10, it follows by \cref{W-is-smw} that $Y_s = Y_k$ in line 3 and $Y_t = \overline{Y_k}$ in line 11. It suffices by \cref{same-prefixes} to prove that the special middle word $\overline{Y_k}$ in $v_j$ is equal to the complement of $\overline{Y_k}$ in $\min w_0$.

 Since $\min w_0 \equiv v_j w_j$, it follows that $\min w_0 \leq v_jw_j$. If the $k+1$-th special middle word $\overline{\overline{Y_k}}$ in $\min w_0$ is 
 not $\overline{Y_k}$, then $\min w_0 < v_jw_j$, and so, by \cref{lem-xyz}, $\overline{\overline{X_kY_kZ_k}} < \overline{X_kY_kZ_k}$.  If $a$ is the suffix of $\overline{Z_{k-1}}$ given in line 3, then 
 $\overline{X_kY_kZ_k}$ is chosen in line 5 to be the least complement of $X_kY_kZ_k$ with prefix $a$. Seeking a contradiction we will show that $\overline{\overline{X_kY_kZ_k}}$ also has a prefix $a$. In order to accomplish this, we show that $\min w_0$ and $v_jw_j$ satisfy the assumption of \cref{lem-refactor1}. In other words, we will show that $v_jw_j = p Y_{k-1}Z_{k-1}\overline{X_k''Y_k}q$ for some $p, q\in A ^*$ and some suffix $\overline{X''_k}$ of $\overline{X_k}$. 
 
 The word $v_{j - 1}$ was defined to be $v' Y_{k -1}$ for some $v'$.
 Since the condition in line 3 holds, $w_{j -1} = Z_{k - 1}w'$ for some $w'$, and $\overline{Z_{k-1}}w' = bX_kY_kw''$ for some $b\in A ^ *$ such that $|b|<|\overline{Z_{k-1}}|$.
 This implies that $X_k = aX_k''$ where $a$ is the suffix of $\overline{Z_{k-1}}$ given in line 3, and $X_k''$ is a prefix of $w'$. Hence, since $w'$ is $Z_{k-1}$-active,
 $w' = X_k''Y_kw''$ and so
 \[
 v_{j-1}w_{j-1} 
 = v' Y_{k - 1}Z_{k - 1}w'
 = v'Y_{k-1}Z_{k-1}X_k''Y_kw'',
 \]
 and 
 \[
 v_jw_j = v' Y_{k-1}Z_{k-1}\overline{X_k''Y_k}q
\] 
for some $q\in A ^ *$. Hence, by \cref{lem-refactor1}, 
the word $a$ is a prefix of both $\overline{X_k}$ and $\overline{\overline{X_k}}$, giving the required contradiction. 
\medskip

\noindent \textbf{line 15:}
Similar to the case of line 11, $W_{j-1} = X_{k-1}Y_{k-1Z_{k-1}}$, $v_j = v_{j-1}Z_{k-1}X_{k}''Y_k$, and we must show that $Y_k$ is in $\min w_0$.
If the conditions of line 3 are satisfied but the conditions of line 4 are not satisfied, then $X_kY_kZ_k$ is the lexicographically minimum relation word with prefix $a$.  If $\min w_0$ does not contain $Y_k$, then it contains a proper complement $\overline{Y_{k}}$. 
As in the previous case, $v_jw_j = v' Y_{k-1}Z_{k-1}\overline{X_k''Y_k}q$, and so by \cref{lem-refactor1} (applied to $v_jw_j$ and $\min w_0$) the word $a$ is a prefix of both $X_k$ and $\overline{X_k}$. Thus it follows by \cref{lem-xyz} that $\overline{X_kY_kZ_k} < X_kY_kZ_k$. But in this case $\overline{X_kY_kZ_k}$ is a proper complement of $X_kY_kZ_k$ that has prefix $a$ and the condition of line 4 is satisfied, a contradiction.

If the conditions of lines 3 and 4 are satisfied but the condition of line 9 is not satisfied, then \textbf{WpPrefix}($w_0, v_{j-1}Z_{k-1}\overline{X_k''Y_kZ_k}t, \varepsilon$) = No. Hence no word equivalent to $w_0$ contains $Y_{k-1}$ and a proper complement $\overline{Y_k}$ of $Y_k$ where $\overline{X_kY_kZ_k}$ has prefix $a$. 
But, by \cref{lem-refactor1}, every word equivalent to $w_0$ that contains $Y_{k-1}$ and a proper complement of $Y_k$, has the property that the proper complement of $Y_k$  is the middle word of a relation word with prefix $a$. Hence no word equivalent to $w_0$ contains both $Y_{k -1}$ and a proper complement of $Y_k$.
In particular, $\min w_0$ contains $Y_k$, as required.
\medskip

\noindent \textbf{line 21:} In this case, $v_j = v_{j -1}aXY$, $w_{j-1} = aXYw'$ and \textbf{WpPrefix}($w', w', Z$)=No. It follows that $Y$ is not a special middle word of $v_{j-1}w_{j-1}$. By assumption $v_{j-1}$ is a prefix of $\min w_0$ and by \cref{lem-2} $aXY$ is a prefix of all words equivalent to $w_{j-1}$. It follows that $v_j = v_{j-1}aXY$ is a prefix of $\min w_0$.
\medskip

\noindent \textbf{line 25:}
In this case, $v_j = v_{j -1}a\overline{XY}$, $w_{j-1} = aXYw'$ and \textbf{WpPrefix}($w', w', Z$) = Yes. It follows that there exists a word equivalent to $v_{j-1}w_{j-1}$ containing $XYZ$ as a factor and hence $Y$ is a special middle word of $v_{j-1}w_{j-1}$. Since $v_{j-1}$ contains the initial $k$ special middle words of $\min w_0$ and $Y$ is a special middle word occurring after $Y_{k-1}$, it follows that $Y = Y_k$.
In this case, $v_j = v_{j - 1}a \overline{X_kY_k}$ and by \cref{same-prefixes} it suffices to show that $\overline{Y_k}$ is a factor of $\min w_0$.
In this case, $w_{j-1} = b_kX_kY_kw'$ and $Z_k$ is a possible prefix of $w'$. By \cref{lem-xyz}, $\min w_0$ must contain the middle word in $\min X_kY_kZ_k$. 
Hence $\min w_0$ contains $\overline{Y_k}$, as required. 
\medskip

\noindent \textbf{line 29:}
In this case, neither of the conditions in lines 3 or 18 are satisfied. Since the condition of line 18 is not satisfied, $w_{j-1}$ contains no relation words as factors and it is only equivalent to itself. Hence $v_j = v_{j-1}w_{j-1}$ is the normal form of $w_0$.
\end{proof}

The proof of the correctness of \textbf{NormalForm} is concluded in the following proposition.

\begin{prop}
If $w_0\in A^*$ is arbitrary, then the word $v$ returned by \textbf{NormalForm}($w_0$) is the lexicographical least word equivalent to $w_0$.
\end{prop}

\begin{proof}
All line numbers in this proof refer to \textbf{NormalForm} in \cref{algorithm-normalform}.
 In \cref{equivalence} it is shown that the word returned by \textbf{NormalForm} is equivalent to $w_0$. We will use the same notation as in the proof of \cref{equivalence}; $v_i, w_i$ will be used to denote the values of $v$ and $w$ after the $i$-th iteration of the while loop in line 2. 
 
  In \cref{part1}, we showed that for every special middle word $Y_k$ in $w_0$ there exists an $i$ such that $v_i$ contains a complement of $Y_k$. Since $v_{i}$ is a proper prefix of $v_{i + 1}$ for every $i$, it follows that eventually $v_i$ contains a complement of every special middle word in $w_0$.
  In \cref{part2}, we showed that  $v_i$ is a prefix of $\min w_0$ for all $i$. Together these two statements imply that when \textbf{NormalForm} terminates, the middle words in $v_i$ coincide with the special middle words in $\min w_0$ 
 and hence $v_i$ is the lexicographically least word  equivalent to $w_0$ by \cref{theorem}.
 \end{proof}

\subsection{Complexity}
In this section we analyze the complexity of \textbf{NormalForm}. 
Throughout this section we suppose that the maximal piece prefix $X$, suffix $Z$, and middle word $Y$ has been computed already for every relation word in the given presentation $\langle A | R \rangle$. The time complexity for doing this is discussed in Section~\ref{section-uniform-wp}.
As such we do not include the complexity of determining that the presentation satisfies $C(4)$, nor that of finding the $X$, $Y$, and $Z$, in the statements in this section. 
We start with two results regarding the complexity of finding a clean overlap prefix for a word $w$ and deciding if a word $w$ is $p$-active for a piece $p$. Finally, we show that for a given C(4) presentation $\langle A \mid R \rangle$, the complexity of \textbf{NormalForm}$(w)$ is $O(|w|^2)$ where $w \in A^*$ is the input.

\begin{lemma} \label{cop - time}
If $w \in A ^ *$ is arbitrary, then the clean overlap prefix of $w$, if any, can be found in time linear to the length of $w$.
\end{lemma}
\begin{proof}
Let $M$ denote the number of relation words and let $\delta$ be the length of the longest relation word in our $C(4)$ presentation. According to Lemma 7 in \cite{overlaps1} to check if a word $v$ has a clean overlap prefix of the form $X_iY_i$ where $X_iY_iZ_i = W_i$, $1 \leq i \leq r$ it suffices to check if $v'$ has a clean overlap prefix of this form, where $v'$ is a prefix of $v$ such that $|v| = 2\delta$. 

Hence, in order to find the clean overlap prefix of $w$ that has the form $sX_iY_i$ for $s \in A^*$ it suffices to check at most $|w|$ suffixes of $w$ for clean overlap prefixes of the form $X_iY_i$. This can be done in $O(|w|)$ time.
\end{proof}

\begin{lemma}
If $w \in A ^ *$ is arbitrary and $p$ is a piece, then deciding if $w$ is $p$-active takes constant time.
\end{lemma}

\begin{proof}
Again, let $M$ denote the number of relation words and let $\delta$ be the length of the longest relation word in our $C(4)$ presentation.
According to Lemma 7 in \cite{overlaps1} it suffices to check if $w'$ is $p$-active, where $w'$ is a prefix of $w$ of length $2\delta$. Since $p$ is a piece, then clearly $|p| < \delta$. A string searching algorithm, such as, for example, Boyer-Moore-Horspool~\cite{Horspool},  can check if there exists some $i$, $1 \leq i \leq M$ such that the factor $X_iY_i$ occurs in $pw'$ before the end of $p$. This takes $O(M\delta |pw'|) = O(3M\delta^2)$ time.
\end{proof}

\begin{prop}
The complexity of \textbf{NormalForm} is $O(|w_0|^2)$ where $w_0\in A ^ *$ is the input, given that the decompositions of the relation words in the presentation into $XYZ$ are known.
\end{prop}
\begin{proof} 
Let $\left\langle A \, | \, R \right\rangle$ be the presentation, let $M$ be the number of distinct relation words in $R$ and let $\delta$ be the length of the longest relation word in $R$. We have already shown that the while loop of line 2 will be repeated at most $|w_0|$ times.
We analyze the complexity of each step of the procedure in the loop.

     In line 3 the algorithm tests if the word $w'$ is $\overline{Z_r}$-active for $\overline{Z_r}$ some complement of $Z_r$ in constant time. In addition, checking if $Z_r \neq \overline{Z_r}$ requires comparing at most $\delta$ characters. Finding the suffix $a$ of $\overline{Z_r}$ such that $aw'=X_sY_sw''$ also requires checking at most $\delta$ characters, hence these checks can be performed in constant time.
    
    In lines 3,9 and 19, \textbf{WpPrefix}$(u, v, p)$ is called. According to \cite{overlaps1}, the algorithm can be implemented with execution time bounded above by a linear function of the length of the shortest of the words $u$ and $v$. Since every time \textbf{WpPrefix}$(u, v, p)$ is called either $u = w_0$ or $u$ is a suffix of some word equivalent to $w_0$, this step can be executed in $O(\delta|w_0|)$ time. 
    
    In lines 4-5 we search for proper complements of $X_sY_sZ_s$ that have the prefix $a$. Since $a$ is a piece, this step also requires constant time.
    
    In line 9 the algorithm finds the suffix $b$ of $\overline{X_s}$ such that $\overline{X_s}=ab$ and the suffix $t$ of \textbf{ReplacePrefix}($w'', Z_s$) that follows $Z_s$. This is also done in constant time since $a$ and $Z_s$ are pieces.
    
    In lines 9, 16 and 26 \cref{algorithm-prefix} is called. Each time, \cref{algorithm-prefix} takes as input a suffix of some word $s$ equivalent to $w_0$. Since $|s|< \delta |w_0|$, this step can be completed in $O(\delta|w_0|)$ time.
    
   In lines 5 and 25 we search for the lexicographically minimal complement of some relation word $X_iY_iZ_i$. Clearly, this check can be done by comparing at most $\delta$ characters $M$ times, hence it is constant for a given presentation. 
    
    In line 18 the algorithm finds the clean overlap prefix of $w$. As shown in \cref{cop - time}, this can be done in $O(|w|)$ time. Since $w$ is always a suffix of some word equivalent to $w_0$, this step can be executed in $O(\delta|w_0|)$ time.
    
    In lines 11-12, 15, 21-22, 25-26 and 29 \cref{algorithm-normalform} concatenates $v$ with a word of length at most $\delta |w_0|$ and deletes a prefix of length at most $\delta |w_0|$ from $w$. Hence these steps require at most $2\delta|w_0|$ time.
\end{proof}

We end the paper with an example of the application of \textbf{NormalForm} to specific $C(4)$ presentation.

   \begin{ex}
  	Let $$\left\langle a,b,c,d \, |\,ab^3a=cdc \right\rangle $$ be the presentation and let $w_0=cdcdcab^3ab^3ab^2cd$. The set of relation words of the presentation is $\{ab^3a, cdc\}$ and each relation word has a single proper complement. The set of pieces of $\mathcal{P}$ is $P =\{\varepsilon, a, b, c, b^2\}$. Let $W_0 = ab^3a, \, W_1 = cdc$. Clearly, $X_{W_0} = a, \, Y_{W_0} = b^3, \, Z_{W_0}  = a$ and $X_{W_1}  = c, \, Y_{W_1} =d, \, Z_{W_1}  = c$.
  	
  	\cref{algorithm-normalform} begins with $v \gets \varepsilon$, $W \gets \varepsilon$ and $w \gets cdcdcab^3ab^3ab^2cd$. Since $u= \varepsilon$ the conditions of line 3 are not satisfied. The word $w$ has a clean overlap prefix $X_{W_1}Y_{W_1}=cd$ followed by $Z_{W_1}=c$ hence \textbf{WpPrefix}($cdcab^3ab^3ab^2cd, cdcab^3ab^3ab^2cd,c$) returns Yes and \textbf{ReplacePrefix}($cdcab^3ab^3ab^2cd,c$) returns $cdcab^3ab^3ab^2cd$. Since $W_0 < W_1$, $v \gets X_{W_0}Y_{W_0}=ab^3$, $w \gets adcab^3ab^3ab^2cd$ and $W \gets ab^3a$ in lines 24-26.
  	
  	Now $W=ab^3a$, $w$ begins with $Z_{W_0}=a$, $w'=dcab^3ab^3ab^2cd$ is $\overline{Z_{W_1}}$-active and the prefix $Y_{W_1}=d$ of $w'$ is followed by $Z_{W_1}=c$, hence the conditions in line 3 are satisfied. In addition, $ab^3a < cdc$ but $X_{W_0}$ and $X_{W_1}$ do not have a common prefix. Hence, in lines 14-16 $v \gets ad$, $w\gets$ \textbf{ReplacePrefix}($cab^3ab^3ab^2cd, c$) = $cab^3ab^3ab^2cd $ and $W \gets cdc$.
  	
  	At this point, $W=cdc$ and $w=cab^3ab^3ab^2cd$ begins with $Z_{W_1}=c$ but $ab^3ab^3ab^2cd$ is not $Z_{W_0}$-active. The word $w$ has the clean overlap prefix $cab^3$ that is followed by $a$, hence in lines 24-25, $v \gets vcab^3$, $w \gets$ \textbf{ReplacePrefix}($ab^3ab^2cd,a)=ab^3ab^2cd$ and $W \gets ab^3a$.
  	
  	Next, $W=ab^3a$ and $w=ab^3ab^2cd$ begins with $Z_{W_0}$ and $b^3ab^2cd$ is not $Z_{W_1}$-active but it has the clean overlap prefix $ab^3$. The clean overlap prefix is followed by $a$, hence in lines 24-25 $v \gets vab^3$, $w\gets ab^2cd$ and $W \gets ab^3a$.
  	
  	Finally, $W=ab^3a$, $w=ab^2cd$ begins with $Z_{W_0}$ but $b^2cd$ is not $Z_{W_1}$-active. Now $w$ has the clean overlap prefix $ab^2cd$ that is followed by $\varepsilon$, hence in line 19 \textbf{WpPrefix}($\varepsilon, \varepsilon, c$) returns No and in lines 20-22 $v \gets vab^2cd$, $W \gets \varepsilon$ and $w \gets \varepsilon$.
  	Since $w=\varepsilon$, the algorithm returns $v=ab^3adcab^3ab^3ab^2cd$. \medskip	\\	
  Next, we will apply \textbf{NormalForm}, to find the normal form of $w_0=cdab^3cdc$. 
  	We begin with $v \gets \varepsilon$, $W \gets \varepsilon$ and $w \gets cdab^3cdc$. Since $ W = \varepsilon$ we do not have the case of line 3. The word $w$ has the clean overlap prefix $cd$ of the form $X_{W_1}Y_{W_1}$ and \textbf{WpPrefix}($ab^3cdc, ab^3cdc, c$)=Yes, and \textbf{ReplacePrefix}($ab^3cdc$, $c$) returns $cdcb^3a$. In lines 24-26 $v\gets ab^3$, $w\gets adcb^3a$ and $W \gets ab^3a$.
  	
  	For this iteration, $u=ab^3a$ and $w$ begins with $Z_{W_0}=a$, $dcb^3a$ is $Z_{W_1}$-active and clearly \textbf{WpPrefix}($cb^3a, cb^3a, c$)=Yes. In addition, $ab^3a<cdc$ but $X_{W_0}$ and $X_{W_1}$ do not have a common prefix, hence in lines 14-16 $v\gets vad$, $w\gets$ \textbf{ReplacePrefix}($cb^3a, c$)= $cb^3a$ and $W \gets cdc$.
  	
  	At this stage, $ W = cdc$, $w=cb^3a$ begins with $Z_{W_1}=c$ and $b^3a$ is $Z_{W_0}$-active but $X_{W_0}$ and $X_{W_1}$ do not have a common prefix hence $v\gets cb^3$, $w\gets$\textbf{ReplacePrefix}($a,\, a$) $=a$, $W \gets ab^3a$ in lines 14-16.
  	
  	Finally, $a$ begins with $Z_{W_0}$ but clearly $\varepsilon$ is not $Z_{W_1}$-active. Also $a$ does not have a clean overlap prefix and in lines 29-30 $v\gets va$, $w\gets \varepsilon$ and the algorithm returns $v=ab^3adcb^3a$. \\
\end{ex}

\section*{Acknowledgements}
The second author would like to thank the School of Mathematics and Statistics of the University of St Andrews and the Cyprus State Scholarship Foundation for their financial support.
The authors would also like to thank Mark Kambites for his correspondence on the topic of this paper, and the anonymous referee for their comments that helped to improve the paper. 

\bibliographystyle{halpha}
\bibliography{main}{}

\newcommand{\etalchar}[1]{$^{#1}$}
\begin{thebibliography}{Kam11b}
\expandafter\ifx\csname url\endcsname\relax
  \def\url#1{\texttt{#1}}\fi
\expandafter\ifx\csname doi\endcsname\relax
  \def\doi#1{\burlalt{doi:#1}{http://dx.doi.org/#1}}\fi
\expandafter\ifx\csname urlprefix\endcsname\relax\def\urlprefix{URL }\fi
\expandafter\ifx\csname href\endcsname\relax
  \def\href#1#2{#2}\fi
\expandafter\ifx\csname burlalt\endcsname\relax
  \def\burlalt#1#2{\href{#2}{#1}}\fi

\bibitem[Ber79]{Berstel}
J.~Berstel.
\newblock {\em Transductions and Context-Free Languages}.
\newblock Teubner-Verlag, 1979.

\bibitem[Gil79]{Gilman}
Robert~H Gilman.
\newblock Presentations of groups and monoids.
\newblock {\em Journal of Algebra}, 57(2):544--554, 1979.
\newblock \doi{https://doi.org/10.1016/0021-8693(79)90238-2}.

\bibitem[Gus97]{suffix_trees}
Dan Gusfield.
\newblock {\em Algorithms on Strings, Trees, and Sequences: Computer Science
  and Computational Biology}.
\newblock Cambridge University Press, 1997.
\newblock \doi{10.1017/CBO9780511574931}.

\bibitem[Hig92]{higgins}
P.M. Higgins.
\newblock {\em Techniques of semigroup theory}.
\newblock Oxford science publications. Oxford University Press, Incorporated,
  1992.

\bibitem[Hor80]{Horspool}
R.~Nigel Horspool.
\newblock Practical fast searching in strings.
\newblock {\em Software: Practice and Experience}, 10(6):501--506, 1980.
\newblock \doi{https://doi.org/10.1002/spe.4380100608}.

\bibitem[Joh85]{Johnson2}
J.~H. Johnson.
\newblock Do rational equivalence relations have regular cross-sections?
\newblock In Wilfried Brauer, editor, {\em Automata, Languages and
  Programming}, pages 300--309, Berlin, Heidelberg, 1985. Springer Berlin
  Heidelberg.

\bibitem[Joh86]{Johnson1}
J.~Howard Johnson.
\newblock Rational equivalence relations.
\newblock In Laurent Kott, editor, {\em Automata, Languages and Programming},
  pages 167--176, Berlin, Heidelberg, 1986. Springer Berlin Heidelberg.

\bibitem[Kam09a]{overlaps1}
Mark Kambites.
\newblock Small overlap monoids {I}: The word problem.
\newblock {\em Journal of Algebra}, 321:2187--2205, 04 2009.
\newblock \doi{10.1016/j.jalgebra.2008.09.038}.

\bibitem[Kam09b]{overlaps2}
Mark Kambites.
\newblock Small overlap monoids {II}: Automatic structures and normal forms.
\newblock {\em Journal of Algebra}, 321:2302--2316, 04 2009.
\newblock \doi{10.1016/j.jalgebra.2008.12.028}.

\bibitem[Kam11a]{Kambites2011aa}
Mark Kambites.
\newblock Generic complexity of finitely presented monoids and semigroups.
\newblock {\em computational complexity}, 20(1):21--50, March 2011.
\newblock \doi{10.1007/s00037-011-0005-5}.

\bibitem[Kam11b]{kambites2011}
Mark Kambites.
\newblock A note on the definition of small overlap monoids.
\newblock {\em Semigroup Forum}, 83:499--512, 12 2011.
\newblock \doi{10.1007/s00233-011-9350-6}.

\bibitem[KB70]{knuth-bendix}
Donald~E. Knuth and Peter~B. Bendix.
\newblock Simple word problems in universal algebras.
\newblock In John Leech, editor, {\em Computational Problems in Abstract
  Algebra}, pages 263--297. Pergamon, 1970.
\newblock \doi{https://doi.org/10.1016/B978-0-08-012975-4.50028-X}.

\bibitem[LS77]{lyndon}
R.C. Lyndon and P.E. Schupp.
\newblock {\em Combinatorial Group Theory}.
\newblock Number $\tau$. 89 in Classics in mathematics. Springer-Verlag, 1977.

\bibitem[M{\etalchar{+}}22]{libsemigroups}
J.~D. Mitchell et~al.
\newblock {\em libsemigroups - C++ library for semigroups and monoids, Version
  2.4.0}, Dec 2022.
\newblock \doi{http://dx.doi.org/10.5281/zenodo.1437752}.

\bibitem[Mar47]{markov}
A.~A. Markov.
\newblock On the impossibility of certain algorithms in the theory of
  associative systems.
\newblock {\em Dokl. Akad. Sci. USSR}, 55:587--590, 1947.

\bibitem[Pos47]{post_1947}
Emil~L. Post.
\newblock Recursive unsolvability of a problem of thue.
\newblock {\em The Journal of Symbolic Logic}, 12(1):1–11, 1947.
\newblock \doi{10.2307/2267170}.

\bibitem[Rem71]{Remmers}
John~Hermann Remmers.
\newblock {\em Some Algorithmic Problems for Semigroups: A Geometric Approach}.
\newblock PhD thesis, University of Michigan, USA, 1971.
\newblock AAI7123856.

\bibitem[Sim94]{sims_1994}
Charles~C. Sims.
\newblock {\em Computation with Finitely Presented Groups}.
\newblock Encyclopedia of Mathematics and its Applications. Cambridge
  University Press, 1994.
\newblock \doi{10.1017/CBO9780511574702}.

\bibitem[TC36]{todd-coxeter}
J.~A. Todd and H.~S.~M. Coxeter.
\newblock A practical method for enumerating cosets of a finite abstract group.
\newblock {\em Proceedings of the Edinburgh Mathematical Society}, 5(1):26--34,
  1936.
\newblock \doi{10.1017/S0013091500008221}.

\bibitem[Ukk95]{Ukkonen}
E.~Ukkonen.
\newblock On-line construction of suffix trees.
\newblock {\em Algorithmica}, 14(3):249--260, September 1995.
\newblock \doi{10.1007/BF01206331}.

\end{thebibliography}

\appendix
\section{Experimental results}\label{appendix-experimental}

In this section we present some empirical evidence related to the performance of our implementation of, and the complexity of, the algorithms presented in this paper. The implementation is that in the C++ library \textsf{libsemigroups}~\cite{libsemigroups}. The benchmarks presented in this section are not intended to be exhaustive, but rather to give a taste of what the implementation provides. 
The timings in this section were generated on an 2020 Apple M1 MacBook Air.

In \cref{figure-c4-check}, we plot the mean time to find the largest $k$ such that the $C(k)$ condition holds for all of the $2$-generated $1$-relation semigroups 
where the maximum length of a relation word is $n\in \{4, \ldots, 12\}$.

\begin{figure}[h]
\centering
\includegraphics[width=0.48\textwidth]{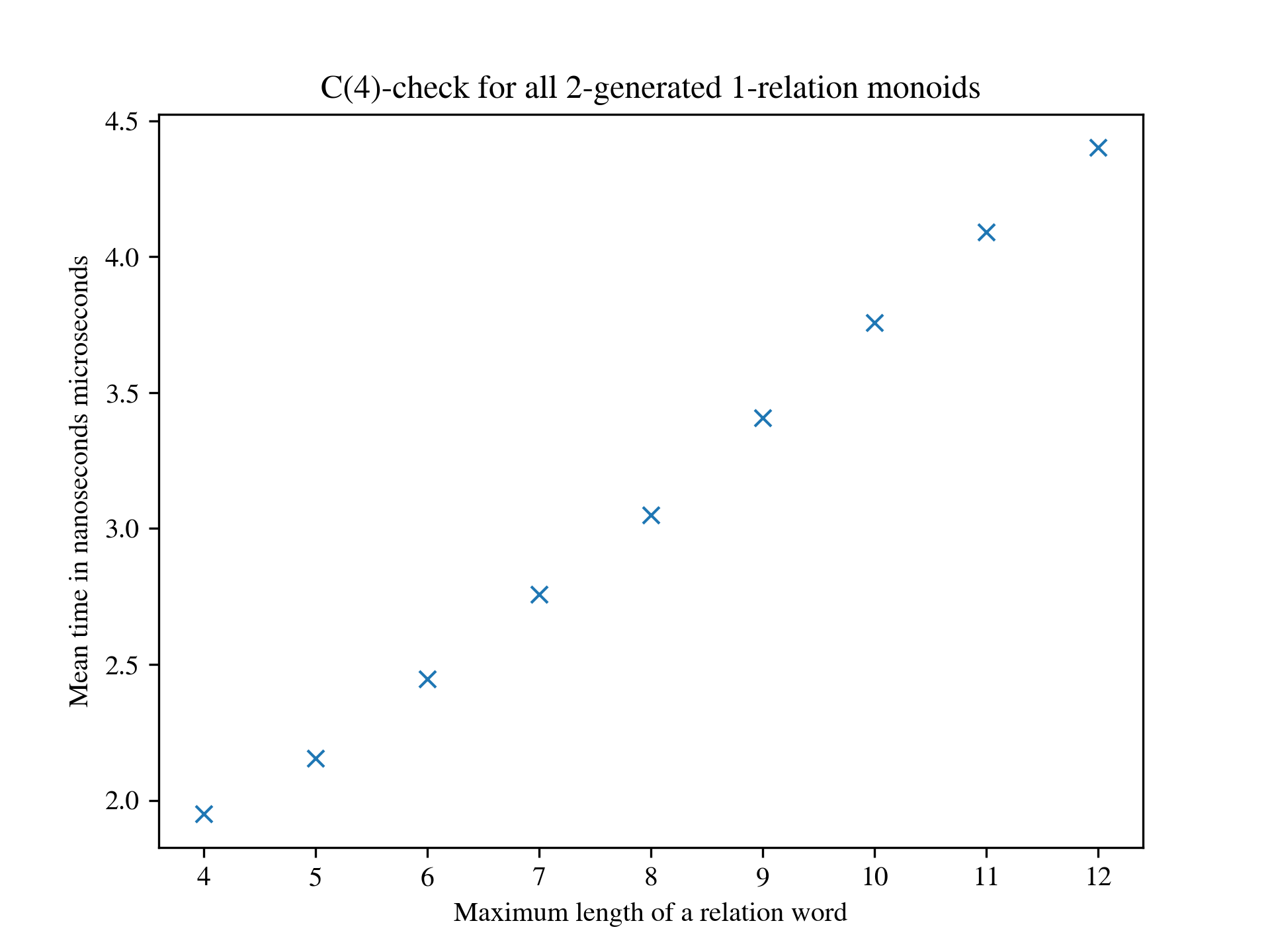}
    \includegraphics[width=0.48\textwidth]{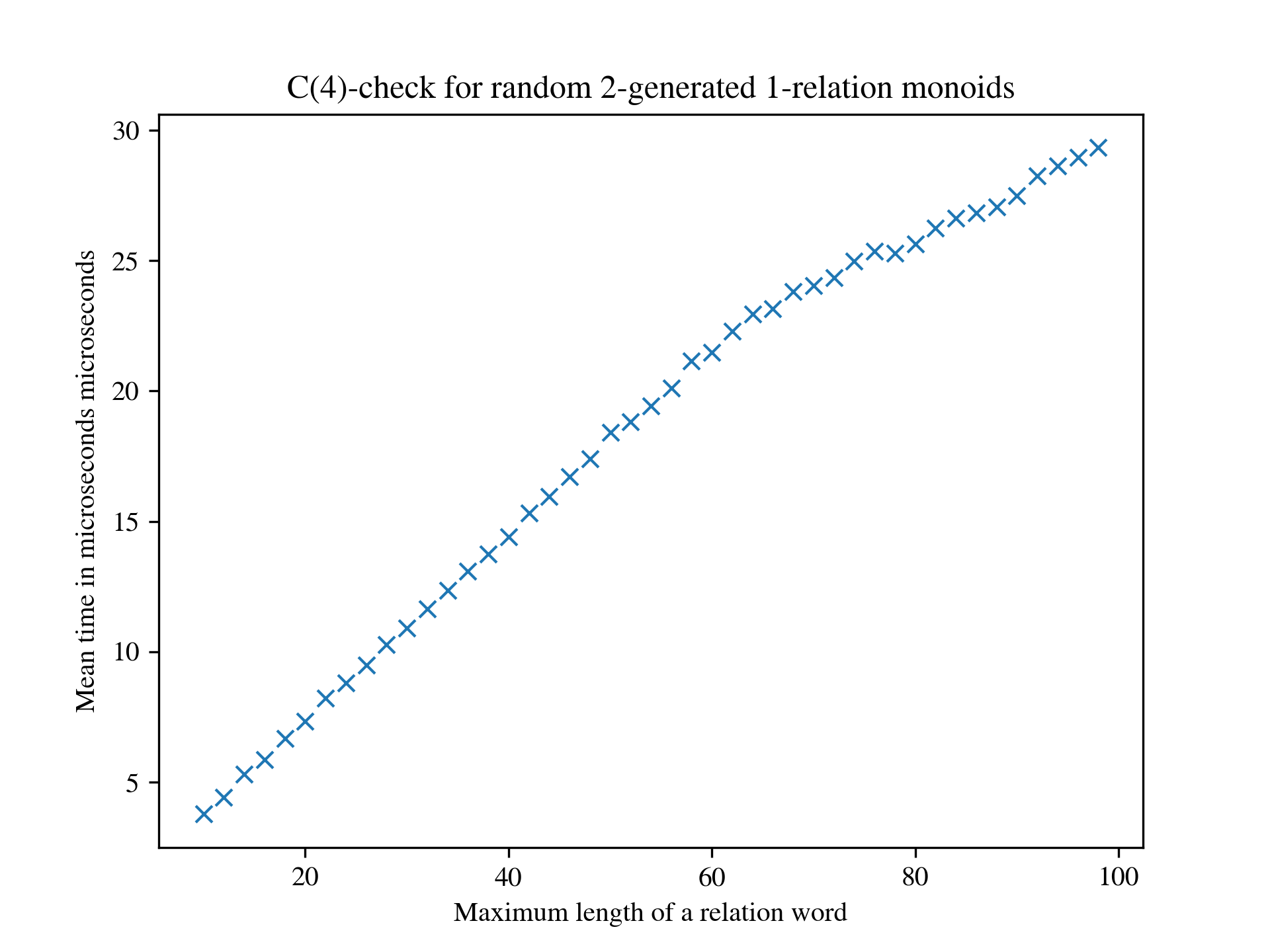}
    \caption{Mean time to check the $C(4)$ condition for: all $2$-generated $1$-relation semigroups where the maximum length of a relation word is $n\in \{4, \ldots, 12\}$ (left); and for a random sample of size $1000$ of $2$-generated $1$-relation semigroups where the maximum length of a relation word is $n\in \{10, 12, \ldots, 100\}$ (right).}
    \label{figure-c4-check}
\end{figure}

In Figures~\ref{figure-equal-to-all-2-gen-1-rel} and~\ref{figure-equal-to-2-gen-2-rel-100}, we provide some empirical data about the implementation in \textsf{libsemigroups}~\cite{libsemigroups} of \textbf{WpPrefix}. 
 If $\langle A|R\rangle $ is a monoid presentation, then the input words $u$ and $v$ were generated as follows: for some $N\in \N$, we choose words $w_0, \ldots, w_{N-1}\in A ^*$  uniformly at random among all words of length $0$ to $4N + 4$. If $W_0$ and $W_1$ are arbitrary relation words in the same relation in $R$ and $(s(0), s(1), \ldots, s(N-1)), (t(0), t(1), \ldots, t(N-1))\in \{0, 1\} ^ N$ are chosen uniformly at random from $\{0, 1\} ^ N$, then $u$ and $v$ are defined to be $w_0W_{s(0)}w_1W_{s(1)} \cdots w_{N-1} W_{s(N-1)}$ and $w_0W_{t(0)}w_1W_{t(1)} \cdots w_{N-1} W_{t(N-1)}$, respectively. In this way, we are guaranteed that $u$ equals $v$ in the presented semigroup, but the words $u$ and $v$ are far from being identical. For comparison, we have also included the time for comparing words chosen uniformly at random, but, as can be seen in the figures, it appears that the probability that two words chosen at random from $A ^ +$ are equal in the presented semigroups is very small. As such, \textbf{WpPrefix} can terminates without considering all letters in either $u$ or $v$. 
 
In \cref{figure-normal-form-all-2-gen-1-rel}, and~\cref{figure-normal-form-rand-2-gen-2-rel} we plot the time taken to compute normal forms of words, chosen in the same way as for \cref{figure-equal-to-all-2-gen-1-rel} and~\ref{figure-equal-to-2-gen-2-rel-100}, against the length of these words for each of the presentations defined above. 
 
\begin{figure}
    \centering
    \includegraphics[width=0.48\textwidth]{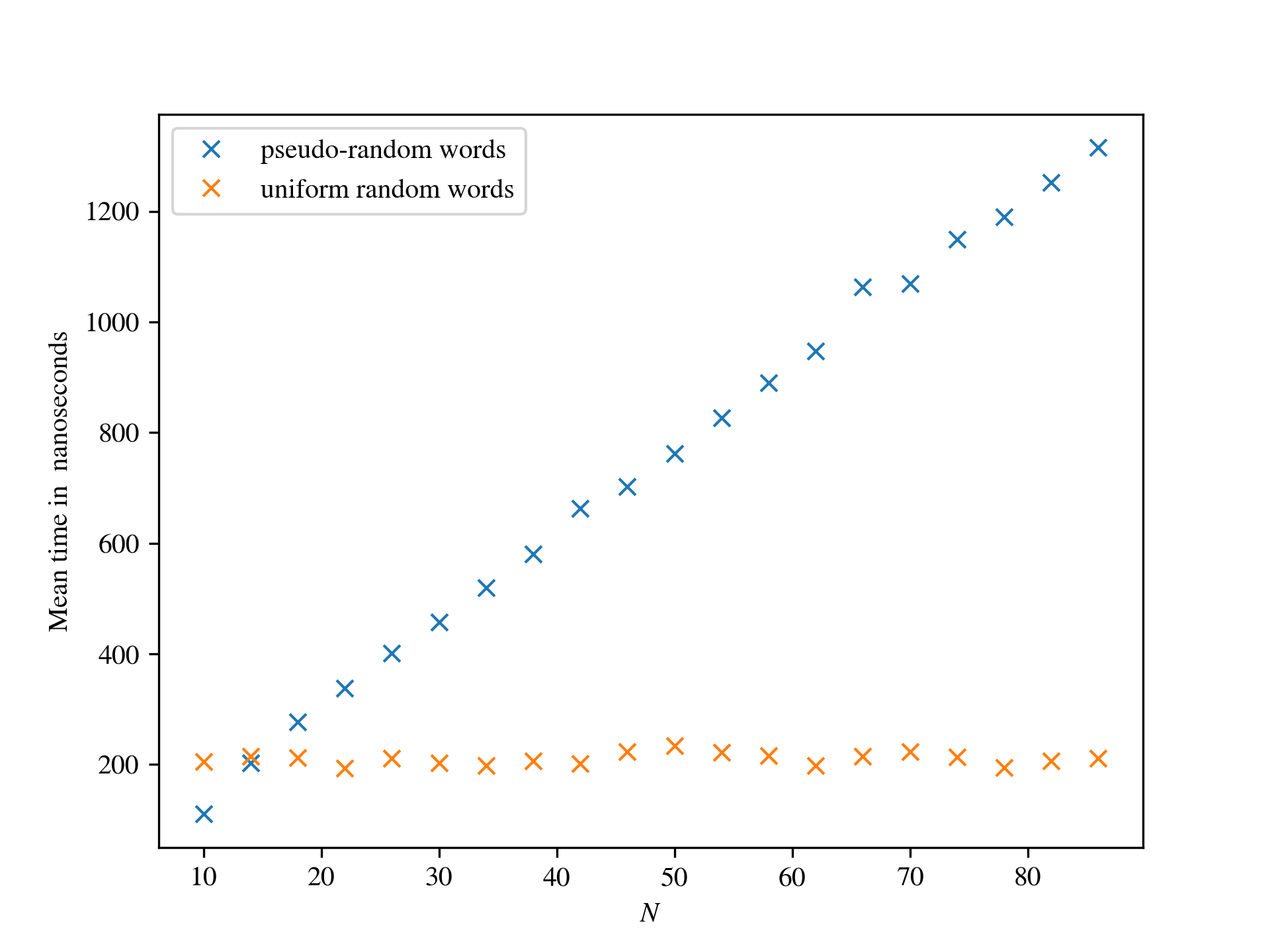}
    \includegraphics[width=0.48\textwidth]{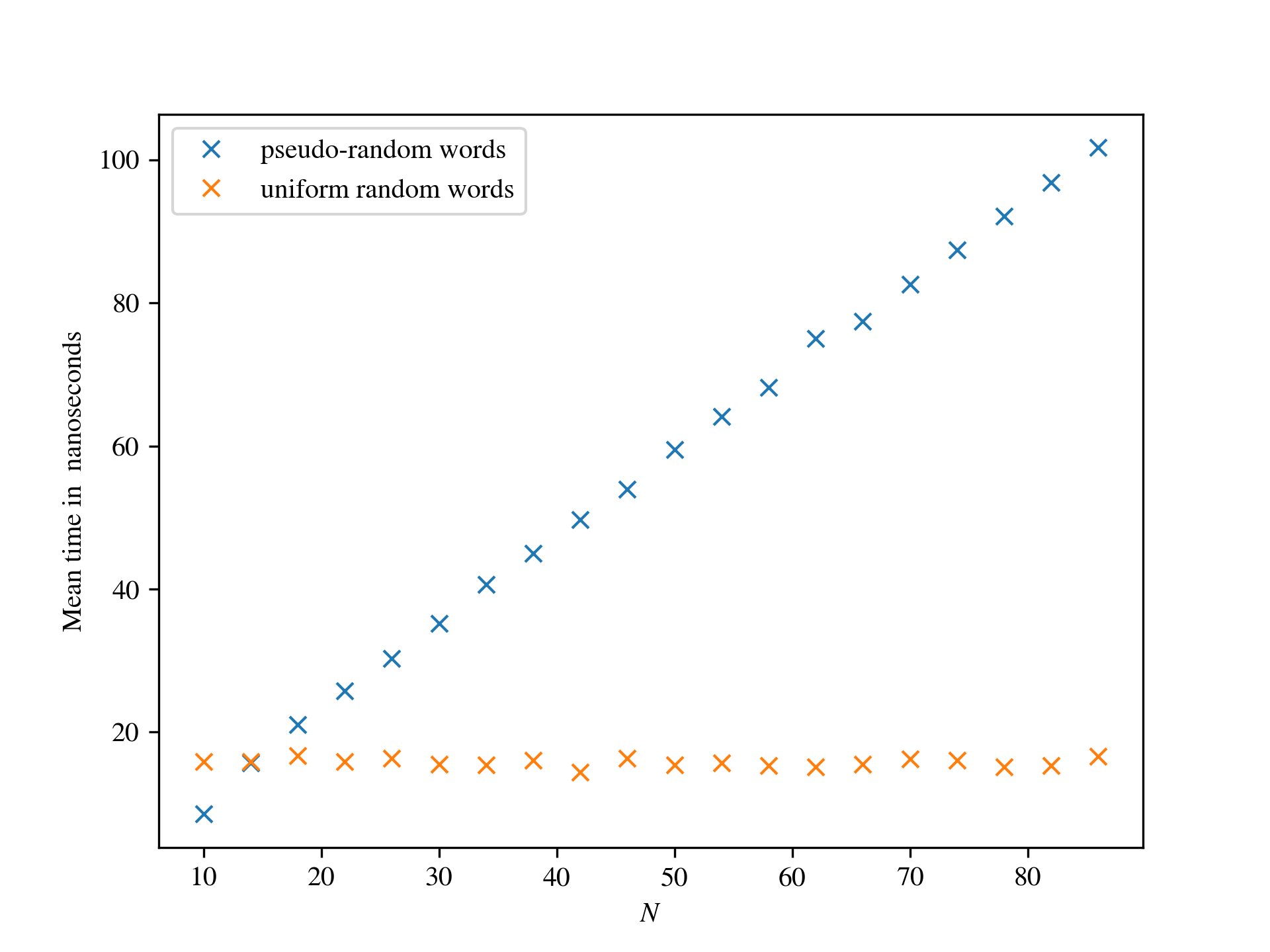}
    \caption{The mean time, over all 2-generated 1-relation $C(4)$ monoids where the maximum length of a relation word is $7$ (left) or $8$ (right), to check equality of 10 random pairs of words of the form
    $w_0W_{s(0)}w_1W_{s(1)} \cdots w_{N-1} W_{s(N-1)}$ and $v = w_0W_{t(0)}w_1W_{t(1)} \cdots w_{N-1} W_{t(N-1)}$, and 10 pairs of words of length $4N^2 + 7N + 4$ chosen uniformly at random,
    for $N \in \{10, 14, \ldots, 86\}$.}
    \label{figure-equal-to-all-2-gen-1-rel}
\end{figure}

\begin{figure}
    \centering
    \includegraphics[width=0.48\textwidth]{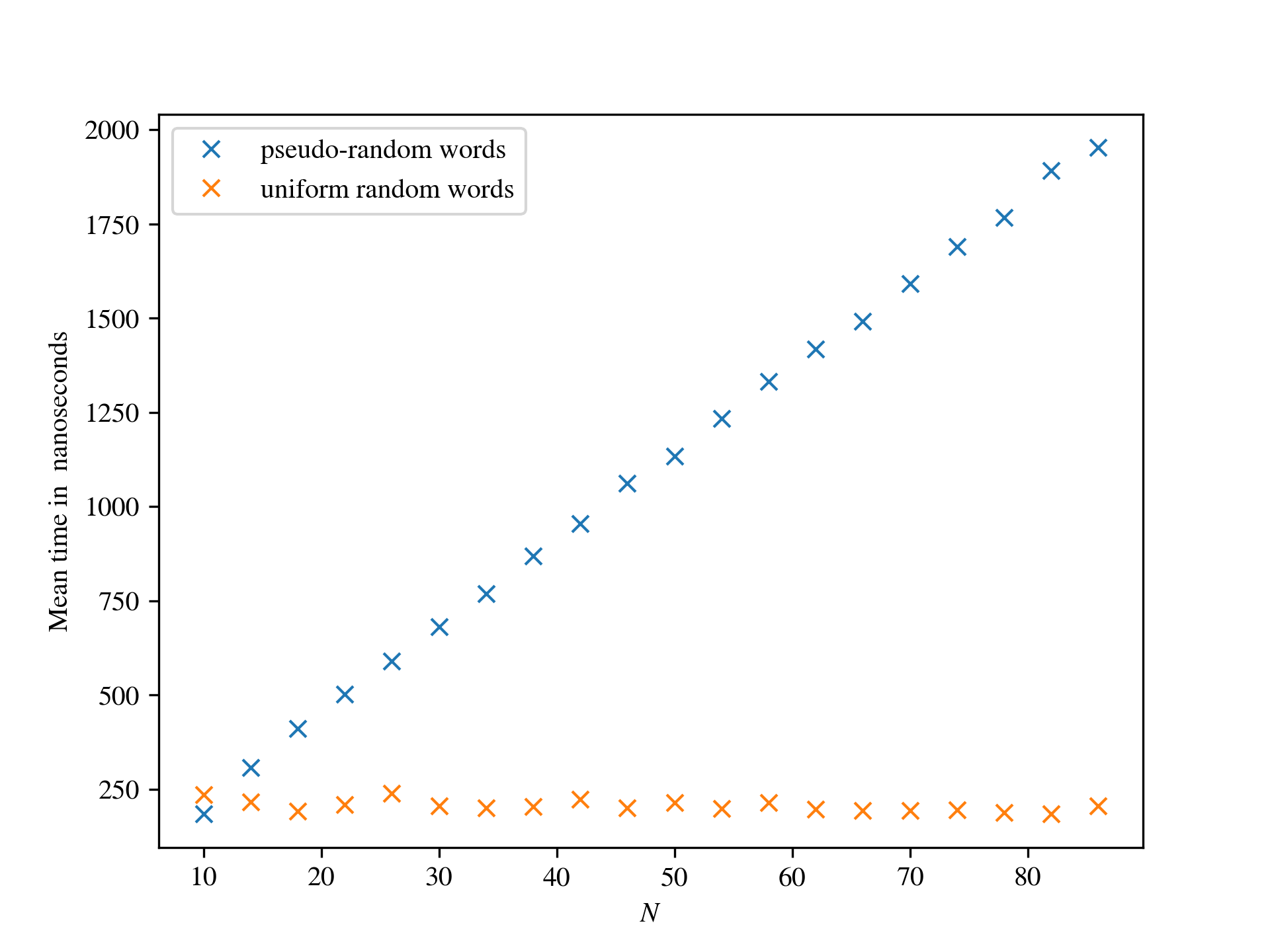}
    \includegraphics[width=0.48\textwidth]{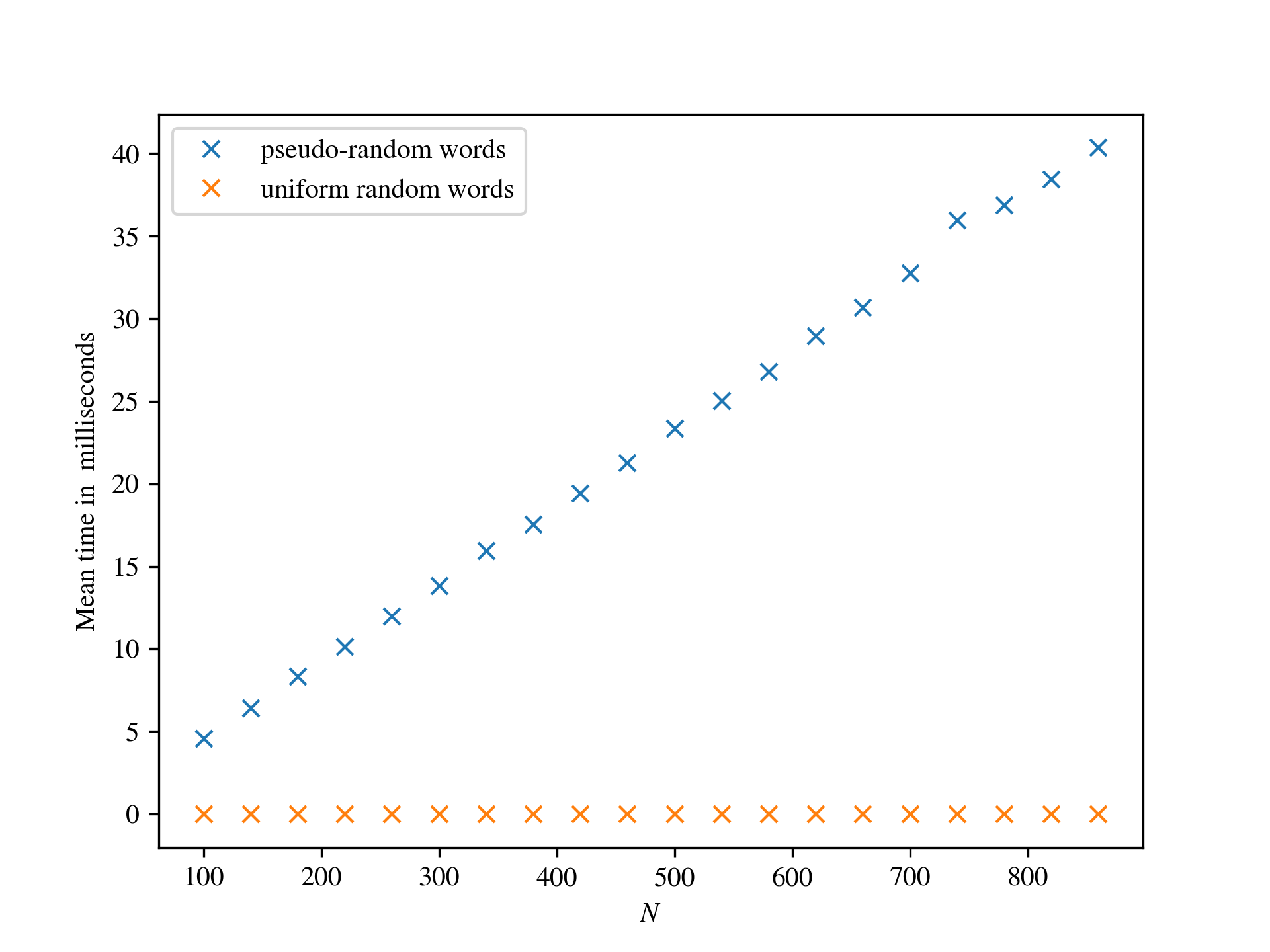}
    \caption{The time, for a randomly chosen 2-generated 2-relation $C(4)$ monoid where the maximum length of a relation word is $100$, to check equality of 10 random pairs of words of the form
    $w_0W_{s(0)}w_1W_{s(1)} \cdots w_{N-1} W_{s(N-1)}$ and $v = w_0W_{t(0)}w_1W_{t(1)} \cdots w_{N-1} W_{t(N-1)}$, and 10 pairs of words of length $4N^2 + 7N + 4$ chosen uniformly at random,
    for $N \in \{10, 14, \ldots, 86\}$ (left) and $N \in \{100, 140, \ldots, 860\}$ (right).}
    \label{figure-equal-to-2-gen-2-rel-100}
\end{figure}

\begin{figure}
    \centering
    \includegraphics[width=0.48\textwidth]{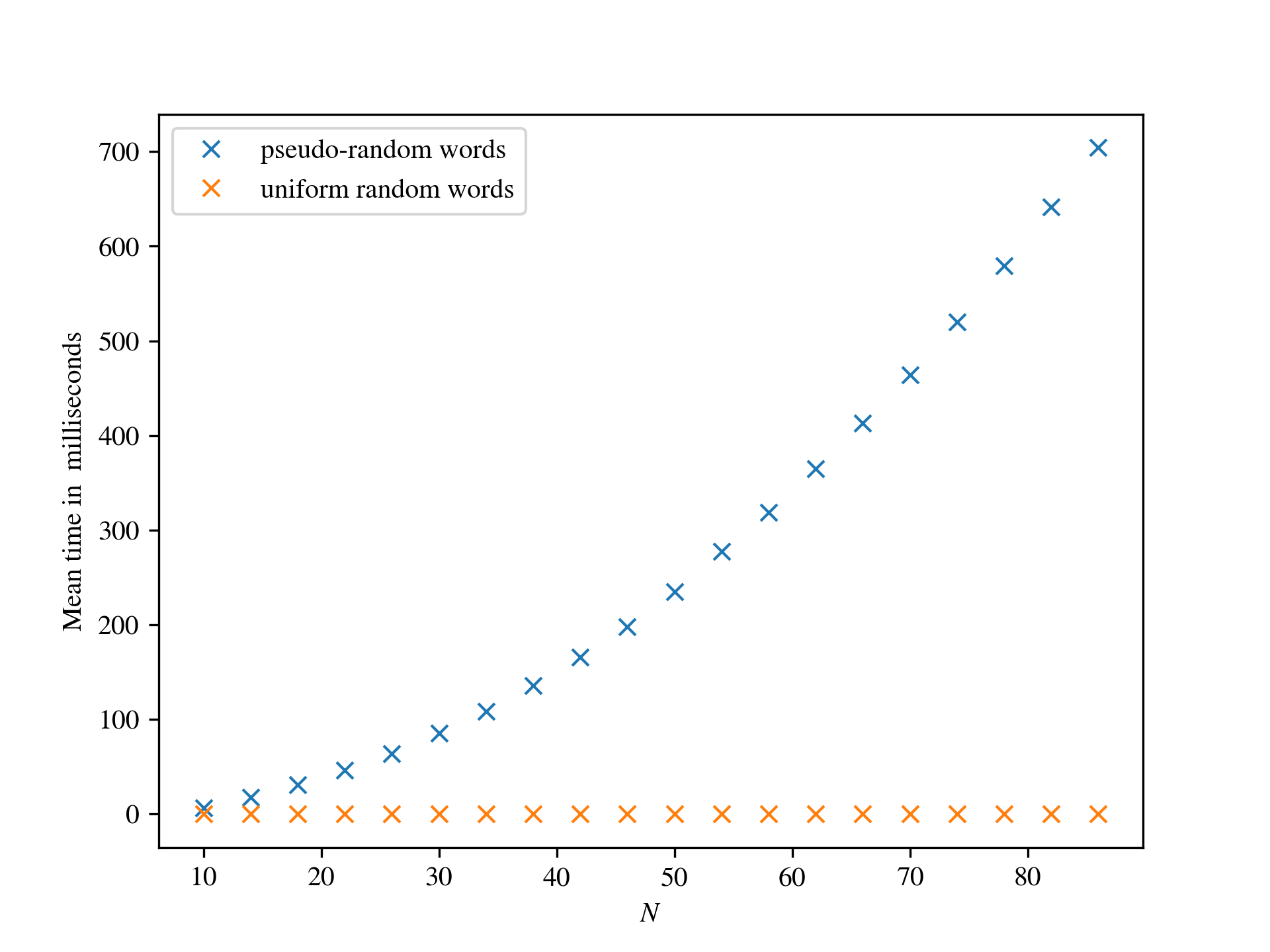}
    \includegraphics[width=0.48\textwidth]{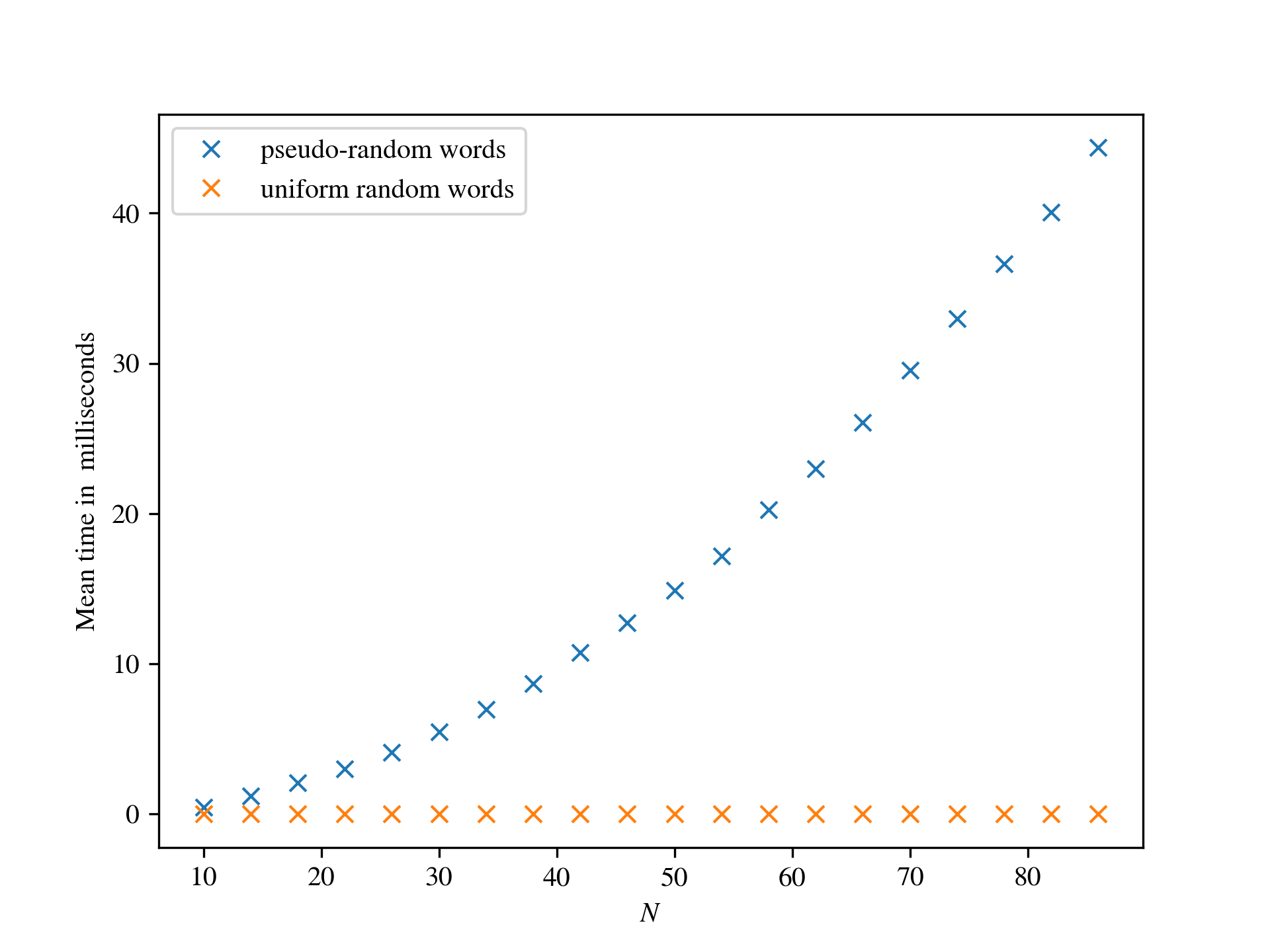}
    \caption{The mean time, over all 2-generated 1-relation $C(4)$ monoids where the maximum length of a relation word is $7$ (left) or $8$ (right), to compute normal forms for 20 random words of the form
    $w_0W_{s(0)}w_1W_{s(1)} \cdots w_{N-1} W_{s(N-1)}$, and $20$ words of length $4N^2 + 7N + 4$ chosen uniformly at random,
    for $N \in \{10, 14, \ldots, 86\}$. Performing a linear regression on a log-log plot of the data in these figures indicates a time complexity of $O(N ^ {2.08831})$ (left) and $O(N ^ {2.03774})$ (right).}
    \label{figure-normal-form-all-2-gen-1-rel}
\end{figure}

\begin{figure}
    \centering
    \includegraphics[width=0.48\textwidth]{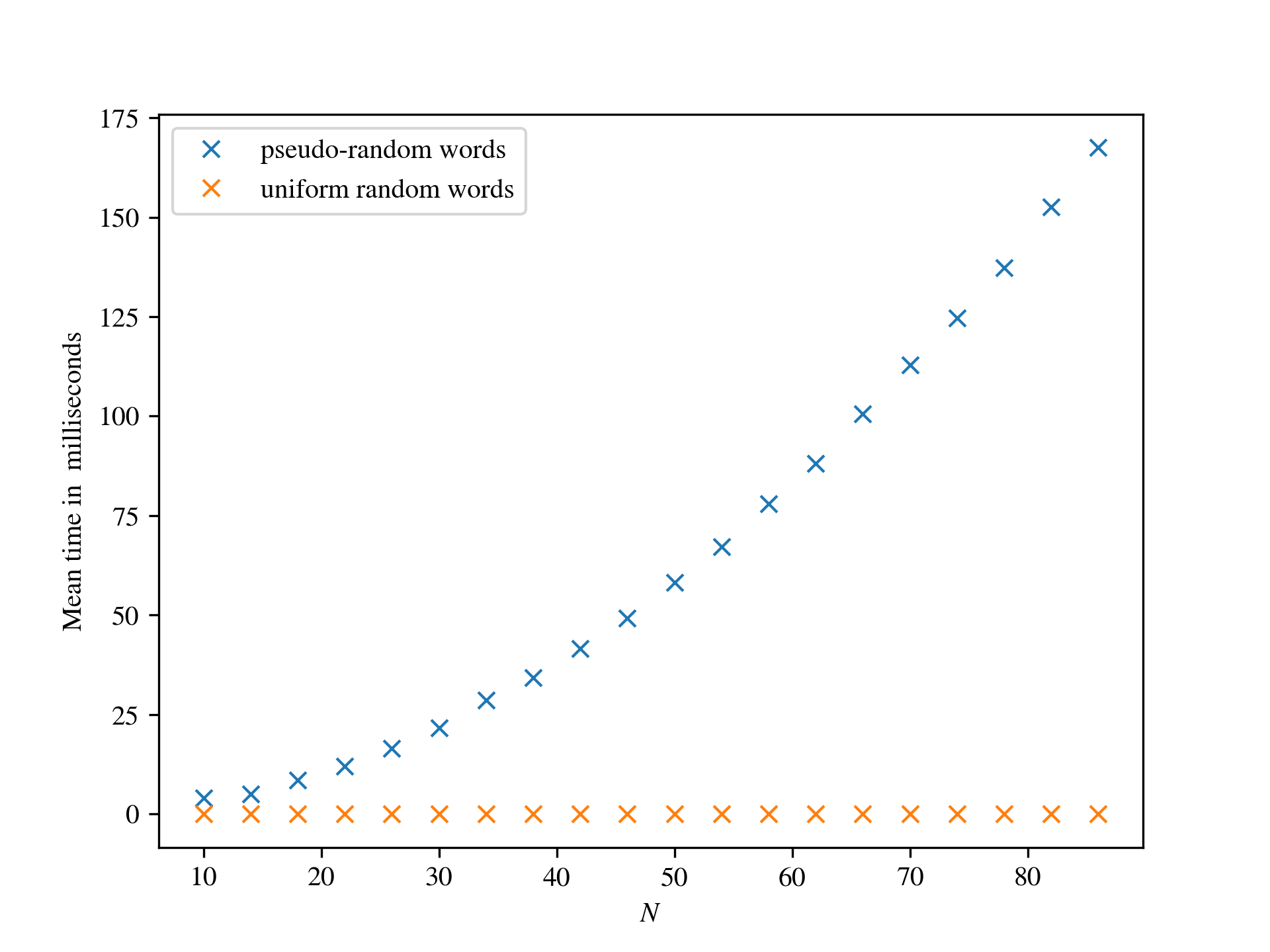}
    \includegraphics[width=0.48\textwidth]{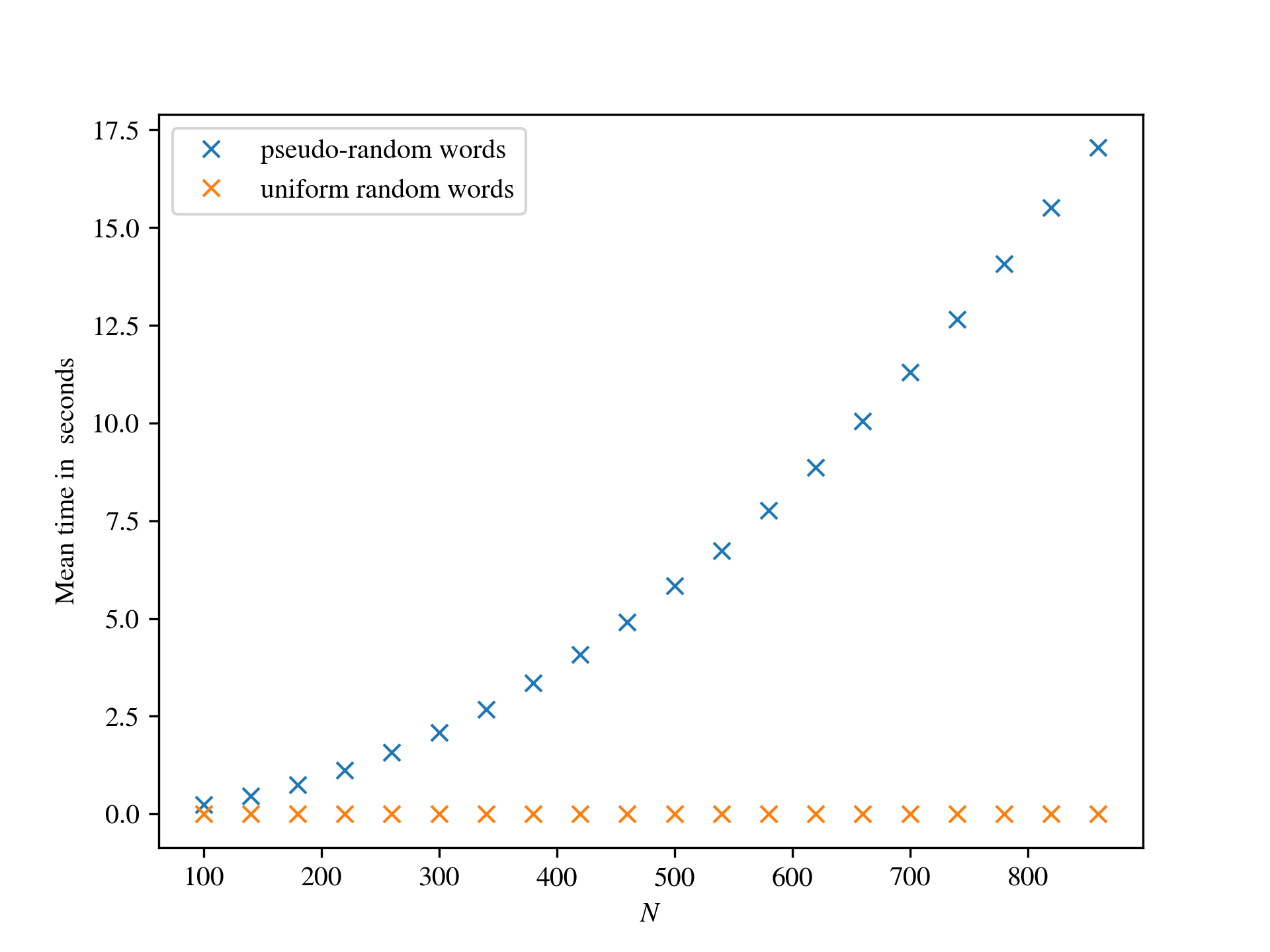}
    \caption{The mean time for a randomly chosen 2-generated 2-relation $C(4)$ monoid where the maximum length of a relation word is $100$ to compute normal forms for $20$ random words of the form
    $w_0W_{s(0)}w_1W_{s(1)} \cdots w_{N-1} W_{s(N-1)}$, and 20 words of length $4N^2 + 7N + 4$ chosen uniformly at random,
    for $N \in \{10, 14, \ldots, 86\}$ (left) and $N \in \{100, 140, \ldots, 860\}$ (right). Performing a linear regression on a log-log plot of the data in these figures indicates a time complexity of $O(N ^ { 1.83817})$ (left) and $O(N ^ {1.99715})$ (right).}
    \label{figure-normal-form-rand-2-gen-2-rel}
\end{figure}

\end{document}